\documentclass[english,final]{IEEEtran}
\usepackage[T1]{fontenc}
\usepackage[latin9]{inputenc}
\usepackage{amsthm}
\usepackage{amsmath}
\usepackage{amssymb}
\usepackage{graphicx}
\usepackage{booktabs} 
\usepackage{epstopdf}
\usepackage{dsfont}
 
\makeatletter 
 
\theoremstyle{plain}
\newtheorem{thm}{\protect\theoremname}
  \theoremstyle{plain}
  \newtheorem{prop}{\protect\propositionname}
  \theoremstyle{plain}
  \newtheorem{lem}{\protect\lemmaname}

\usepackage[noadjust]{cite} 
\usepackage{enumerate}
\usepackage{amsmath}
\usepackage{babel}

  \providecommand{\lemmaname}{Lemma}
  \providecommand{\propositionname}{Proposition}
\providecommand{\theoremname}{Theorem}

\newcommand{\Yvi}{\boldsymbol{Y}_{\hspace*{-0.35ex}1}}
\newcommand{\Yvii}{\boldsymbol{Y}_{\hspace*{-0.35ex}2}}

\newcommand{\Yvu}{\boldsymbol{Y}_{\hspace*{-0.35ex}u}}

\newcommand{\pebar}{\overline{p}_{e}}
\newcommand{\peobar}{\overline{p}_{e,0}}
\newcommand{\peibar}{\overline{p}_{e,1}}
\newcommand{\penubar}{\overline{p}_{e,\nu}}
\newcommand{\peiibar}{\overline{p}_{e,2}}
\newcommand{\peiiibar}{\overline{p}_{e,12}}

\newcommand{\rcuiiia}{\mathrm{rcu}_{12,1}}
\newcommand{\rcuiiib}{\mathrm{rcu}_{12,2}}
\newcommand{\rcuiiinu}{\mathrm{rcu}_{12,\nu}}

\newcommand{\Ercc}{E_{r}^{\mathrm{cc}}}
\newcommand{\Erccprime}{E_{r,12}^{\mathrm{cc}^{\prime}}}

\newcommand{\Eiiiazcc}{E_{0,12,1}^{\mathrm{cc}}}

\newcommand{\Eiiiazcost}{E_{0,12,1}^{\mathrm{cost}}}

\newcommand{\Eircc}{E_{r,1}^{\mathrm{cc}}}
\newcommand{\Eiircc}{E_{r,2}^{\mathrm{cc}}}
\newcommand{\Enurcc}{E_{r,\nu}^{\mathrm{cc}}}
\newcommand{\Eiiircc}{E_{r,12}^{\mathrm{cc}}}

\newcommand{\LM}{I_{\mathrm{LM}}}

\newcommand{\RegLM}{\mathcal{R}_{\mathrm{LM}}}

\newcommand{\SetScc}{\mathcal{S}}
\newcommand{\SetSncc}{\mathcal{S}_{n}}

\newcommand{\SetTocc}{\mathcal{T}_{0}}
\newcommand{\SetToncc}{\mathcal{T}_{0,n}}
\newcommand{\SetTicc}{\mathcal{T}_{1}}

\newcommand{\SetTiiicc}{\mathcal{T}_{12}}
\newcommand{\SetTiiincc}{\mathcal{T}_{12,n}}

\newcommand{\SetTnucc}{\mathcal{T}_{\nu}}

\newcommand{\Ptilde}{\widetilde{P}}

\newcommand{\ubar}{\overline{u}}
\newcommand{\Ubar}{\overline{U}}

\newcommand{\xbar}{\overline{x}}
\newcommand{\Xbar}{\overline{X}}
\newcommand{\xtilde}{\widetilde{x}}
\newcommand{\Xtilde}{\widetilde{X}}
\newcommand{\zbar}{\overline{z}}

\newcommand{\qv}{\boldsymbol{q}}
\newcommand{\Qv}{\boldsymbol{Q}}
\newcommand{\uvbar}{\overline{\boldsymbol{u}}}
\newcommand{\uv}{\boldsymbol{u}}
\newcommand{\Uvbar}{\overline{\boldsymbol{U}}}
\newcommand{\Uv}{\boldsymbol{U}}
\newcommand{\Wv}{\boldsymbol{W}}

\newcommand{\xvbar}{\overline{\boldsymbol{x}}}

\newcommand{\xv}{\boldsymbol{x}}

\newcommand{\Xvbar}{\overline{\boldsymbol{X}}}

\newcommand{\Xv}{\boldsymbol{X}}
\newcommand{\yv}{\boldsymbol{y}}
\newcommand{\Yv}{\boldsymbol{Y}}
\newcommand{\zv}{\boldsymbol{z}}

\newcommand{\Fsf}{\mathsf{F}}
\newcommand{\Tsf}{\mathsf{T}}

\newcommand{\Ac}{\mathcal{A}}
\newcommand{\Cc}{\mathcal{C}}
\newcommand{\Dc}{\mathcal{D}}
\newcommand{\Ec}{\mathcal{E}}

\newcommand{\Kc}{\mathcal{K}}
\newcommand{\Pc}{\mathcal{P}}
\newcommand{\Rc}{\mathcal{R}}

\newcommand{\Uc}{\mathcal{U}}

\newcommand{\Xc}{\mathcal{X}}
\newcommand{\Yc}{\mathcal{Y}}
\newcommand{\Zc}{\mathcal{Z}}

\newcommand{\EE}{\mathbb{E}}
\newcommand{\PP}{\mathbb{P}}
\newcommand{\RR}{\mathbb{R}}

\newcommand{\defeq}{\triangleq}

\long\def\symbolfootnote[#1]#2{\begingroup\def\thefootnote{\fnsymbol{footnote}}\footnote[#1]{#2}\endgroup}

\DeclareMathOperator*{\argmax}{arg\,max}

\newcommand{\openone}{\mathds{1}}

\makeatother

\begin{document}
 
\title{Multiuser Random Coding \\ Techniques for Mismatched Decoding}
 
\author{Jonathan Scarlett, \IEEEmembership{Member, IEEE}, Alfonso Martinez, \IEEEmembership{Senior Member, IEEE}, \\ and Albert Guill\'en i F\`abregas, \IEEEmembership{Senior Member, IEEE}} 

\maketitle

\begin{abstract}  
    This paper studies multiuser random coding techniques for channel coding with 
    a given (possibly suboptimal) decoding rule. For the mismatched discrete memoryless 
    multiple-access channel, an error exponent is obtained that is tight with respect to the ensemble
    average, and positive within the interior of Lapidoth's achievable rate region.
    This exponent proves the ensemble tightness of the exponent of Liu and Hughes
    in the case of maximum-likelihood decoding.  An equivalent dual form of
    Lapidoth's achievable rate region is given, and the latter is shown to extend
    immediately to channels with infinite and continuous alphabets.
    In the setting of single-user mismatched decoding, similar analysis 
    techniques are applied to a refined version of superposition coding, which
    is shown to achieve rates at least as high as standard superposition coding 
    for any set of random-coding parameters.
\end{abstract}
\begin{IEEEkeywords}
    Mismatched decoding, multiple-access channel, superposition coding, random coding, error exponents, ensemble tightness, Lagrange duality, maximum-likelihood decoding.
\end{IEEEkeywords}

\symbolfootnote[0]{
J. Scarlett was with the Department of Engineering, University of Cambridge, Cambridge, CB2 1PZ, U.K.
He is now with the Laboratory for Information and Inference Systems, 
\'Ecole Polytechnique F\'ed\'erale de Lausanne, CH-1015, Switzerland (e-mail: jmscarlett@gmail.com). 
A. Martinez is with the Department of Information and Communication Technologies, Universitat Pompeu Fabra, 08018 Barcelona, 
Spain (e-mail: alfonso.martinez@ieee.org).  A. Guill\'en i F\`abregas is with the Instituci\'o Catalana de Recerca i Estudis Avan\c{c}ats (ICREA), 
the Department of Information and Communication Technologies, Universitat Pompeu Fabra, 08018 Barcelona, Spain, 
and also with the Department of Engineering, University of Cambridge, Cambridge, CB2 1PZ, U.K. (e-mail: guillen@ieee.org).

This work has been funded in part by the European Research Council under ERC grant agreement 259663, by the European Union's 7th Framework Programme   
under grant agreement 303633 and by the Spanish Ministry of Economy and Competitiveness under grants RYC-2011-08150 and TEC2012-38800-C03-03.

This work was presented in part at the 50th Annual Allerton Conference on Communication, Control and Computing 2012, and at the 2013 IEEE International Symposium on Information Theory.
}

\thispagestyle{empty}

\section{Introduction} \label{sec:MU_INTRODUCTION}

The mismatched decoding problem \cite{Hui,Csiszar1,Csiszar2,Merhav,ConverseMM,MacMM,MMRevisited,JournalSU,MMSomekh} 
seeks to characterize the performance of coded communication systems when the 
decoding rule is fixed and possibly suboptimal. This problem is of interest,
for example, when the optimal decoding rule is infeasible due to channel uncertainty 
or implementation constraints.  Finding a single-letter expression for the mismatched 
capacity (i.e.~the highest achievable rate with mismatched decoding; see Section \ref{sub:MU_SETUP} for formal definitions) remains
an open problem even for single-user discrete memoryless channels.  The vast
majority of existing works have focused on achievability results via random coding.

The most notable early works are by Hui
\cite{Hui} and Csisz\'{a}r and K\"{o}rner \cite{Csiszar1}, who independently
derived the achievable rate known as the LM rate, using random codes in which 
each codeword has a constant or nearly-constant composition.  A generalization
to infinite and continuous alphabets was given by Ganti \emph{et al.} \cite{MMRevisited} using cost-constrained coding techniques, relying
on a Lagrange dual formulation of the LM rate that first appeared in \cite{Merhav}.
In general, the LM rate can be strictly smaller than the mismatched capacity 
\cite{MacMM,Csiszar2}.
Motivated by the lack of converse results, the concept of ensemble tightness
has been addressed in \cite{Merhav,MMRevisited,JournalSU}, where it has been shown
that, for any DMC, the LM rate is the best rate possible for the constant-composition
and cost-constrained random-coding ensembles.
In \cite{Csiszar2}, Csisz\'{a}r and Narayan showed that
better achievable rates can be obtained by applying the LM rate to
the second-order product channel, and similarly for higher-order products.
Random-coding error exponents for mismatched decoding were given in 
\cite{Compound,Variations,JournalSU}, and ensemble tightness
was addressed in \cite{JournalSU}.  

The mismatched multiple-access channel (MAC) was 
considered by Lapidoth \cite{MacMM}, who 
obtained an achievable rate region and showed the surprising fact that 
the single-user LM rate can be improved by treating the single-user 
channel as a MAC. Thus, as well as being of independent interest, 
network information theory problems with mismatched decoding
can also provide valuable insight into the single-user
mismatched decoding problem. In recent work that  
developed independently of ours, Somekh-Baruch \cite{MMSomekh} 
gave error exponents and rate regions for the cognitive MAC (i.e.~the 
MAC where one user knows both messages and the other only knows its own) using two 
multiuser coding schemes: superposition coding and random binning.  When 
applied to single-user mismatched channels, these yield achievable rates
that can improve on those by Lapidoth when certain auxiliary variables
are fixed.  

In this paper, we build on the work of \cite{MacMM} and study multiuser 
coding techniques for channels with mismatched
decoding.  Our main contributions are as follows:
\begin{enumerate}
    \item We develop a variety of tools for studying multiuser random coding ensembles
    in mismatched decoding settings.  Broadly speaking, our techniques permit the
    derivations of ensemble-tight error exponents for channels with finite input
    and output alphabets, as well as generalizations to continuous alphabets based on Lagrange duality analogous to those for the single-user setting mentioned above.
    \item By applying our techniques to the mismatched MAC, we provide an alternative
    derivation of Lapidoth's rate region \cite{MacMM} that also yields the ensemble-tight
    error exponent, and the appropriate generalization to continuous alphabets.
    By specializing to the case of ML decoding, we prove the ensemble tightness
    of the exponent given in \cite{MACExponent4} for constant-composition random coding,
    which was previously unknown.
    \item For the single-user channel, we introduce a refined version of 
    superposition coding that yields rates at 
    least as high as the standard version \cite{MMSomekh,Thesis} for any choice 
    of parameters, with strict improvements possible when the input 
    distribution is fixed. 
\end{enumerate}
To avoid overlap with \cite{MMSomekh}, we have omitted 
the parts of our work that appeared therein; 
however, these can also be found in \cite{Thesis}.

For mismatched DMCs, the results of this paper and various previous works
can be summarized by the following list of  random-coding constructions, in decreasing order of achievable rate:
\begin{enumerate}
  \item Refined superposition coding (Theorems \ref{thm:RSC_Main} and \ref{thm:RSC_Dual}),
  \item Standard superposition coding (Theorems \ref{thm:SC_Rate_SC} and \ref{thm:SC_Rate_Dual}; see \cite{MMSomekh,Thesis}),
  \item Expurgated parallel coding \cite{MacMM},
  \item Constant-composition or cost-constrained coding with independent codewords (LM Rate \cite{Hui,Csiszar1,MMRevisited}),
  \item i.i.d.~coding with independent codewords (generalized mutual information \cite{Compound}). 
\end{enumerate}
The gap between 1) and 2) can be strict for a given input distribution;
no examples are known where the gap between 2) and 3) is strict; and the gaps between
the remaining three can be strict even for an optimized input distribution. 
Numerical examples are provided in Section \ref{sec:MU_NUMERICAL}.


\subsection{System Setup} \label{sub:MU_SETUP}

Throughout the paper, we consider both the mismatched single-user channel and
the mismatched multiple-access channel.  Here we provide a description
of each. 

\subsubsection{Mismatched Single-User Channel} \label{sub:SU_SETUP}

The input and output alphabets are denoted by $\Xc$ and $\Yc$
respectively, and the channel transition law is denoted by $W(y|x)$,
thus yielding an $n$-letter transition law given by
\begin{equation}
    W^{n}(\yv|\xv)\defeq\prod_{i=1}^{n}W(y_{i}|x_{i}).
\end{equation}
If $\Xc$ and $\Yc$ are finite, the channel is referred
to as a discrete memoryless channel (DMC).
We consider length-$n$ block coding, in which a codebook
$\Cc=\{\xv^{(1)},\dotsc,\xv^{(M)}\}$
is known at both the encoder and decoder. The encoder takes as input
a message $m$ uniformly distributed on the set $\{1,\dotsc,M\}$,
and transmits the corresponding codeword $\xv^{(m)}$.
The decoder receives the vector $\yv$ at the output of
the channel, and forms the estimate 
\begin{equation}
    \hat{m}=\argmax_{j\in\{1,\dotsc,M\}} q^{n}(\xv^{(j)},\yv),\label{eq:SU_DecodingRule}
\end{equation}
where $n$ is the length of each codeword, and
$q^{n}(\xv,\yv)\defeq\prod_{i=1}^{n}q(x_{i},y_{i})$.
The function $q(x,y)$ is called the \emph{decoding metric}, and is 
assumed to be non-negative. In the case of a tie, a codeword achieving 
the maximum in \eqref{eq:SU_DecodingRule} is selected uniformly at random.
In the case that $q(x,y)=W(y|x)$, the decoding rule in \eqref{eq:SU_DecodingRule} 
is that of optimal maximum-likelihood (ML) decoding. 

A rate $R$ is said to be achievable if, for all $\delta > 0$, there
exists a sequence of codebooks $\Cc_n$ with at least $\exp(n(R-\delta))$ codewords
of length $n$ such that $\lim_{n\to\infty} p_{e}(\Cc_{n}) = 0$ under the decoding metric $q$. The mismatched 
capacity of a given channel and metric is defined to be the supremum of all achievable rates. 

An error exponent $E(R)$ is said to be achievable if there exists a 
sequence of codebooks $\Cc_{n}$ with at least $\exp(nR)$ codewords of length $n$ such that
\begin{equation}
    \liminf_{n\to\infty}-\frac{1}{n}\log p_{e}(\Cc_{n})\ge E(R). \label{eq:SU:AchievableExp}
\end{equation}
We let $\pebar(n,M)$ denote the average error probability
with respect to a given random-coding ensemble that will be clear
from the context. A random-coding error exponent $E_{r}(R)$ is
said to exhibit ensemble tightness if
\begin{equation}
\lim_{n\to\infty}-\frac{1}{n}\log\pebar(n,e^{nR})=E_{r}(R).\label{eq:SU_EnsTight}
\end{equation}
For all of the cases of interest in this paper, the limit will exist.

With these definitions, the above-mentioned LM rate is given as follows 
for an arbitrary input distribution $Q$:
\begin{equation}
    \LM(Q)\defeq\min_{\substack{\Ptilde_{XY}\,:\,\Ptilde_{X}=Q,\Ptilde_{Y}=P_{Y}\\
        \EE_{\Ptilde}[\log q(X,Y)]\ge\EE_{P}[\log q(X,Y)]}}I_{\Ptilde}(X;Y),\label{eq:SU_PrimalLM}
\end{equation}
where $P_{XY} = Q\times W$. This rate can equivalently be expressed as \cite{Merhav} 
\begin{equation}
    \LM(Q)=\sup_{s\ge0,a(\cdot)}\EE\left[\log\frac{q(X,Y)^{s}e^{a(X)}}{\EE[q(\Xbar,Y)^{s}e^{a(\Xbar)}\,|\, Y]}\right],\label{eq:SU_DualLM}
\end{equation}
where $(X,Y,\Xbar) \sim Q(x)W(y|x)Q(\xbar)$.
In the terminology of \cite{MMRevisited}, \eqref{eq:SU_PrimalLM}
is the primal expression and \eqref{eq:SU_DualLM} is the dual expression. 

\subsubsection{Mismatched Multiple-Access Channel}  \label{sub:MAC_SETUP}

We also consider a 2-user memoryless MAC $W(y|x_{1},x_{2})$ with input
alphabets $\Xc_{1}$ and $\Xc_{2}$ and output
alphabet $\Yc$.  In the case that each alphabet is finite, the MAC is
referred to as a discrete memoryless MAC (DM-MAC). 
The decoding metric is denoted by $q(x_{1},x_{2},y)$, and
we write $W^{n}(\yv|\xv_{1},\xv_{2})\defeq\prod_{i=1}^{n}W(y_{i}|x_{1,i},x_{2,i})$
and $q^{n}(\xv_{1},\xv_{2},\yv)\defeq\prod_{i=1}^{n}q(x_{1,i},x_{2,i},y_{i})$.

Encoder $\nu=1,2$ takes as input a message $m_{\nu}$ uniformly
distributed on the set $\{1,\dotsc,M_{\nu}\}$,
and transmits the corresponding codeword $\xv_{\nu}^{(m_{\nu})}$
from the codebook $\Cc_{\nu}=\{\xv_{\nu}^{(1)},\dotsc,\xv_{\nu}^{(M_{\nu})}\}$.
Given the output sequence $\yv$,
the decoder forms an estimate $(\hat{m}_{1},\hat{m}_{2})$ of the message pair, given by 
\begin{equation}
    (\hat{m}_{1},\hat{m}_{2})=\argmax_{(i,j) \in\{1,\dotsc,M_{1}\} \times \{1,\dotsc,M_{2}\}}q^{n}(\xv_{1}^{(i)},\xv_{2}^{(j)},\yv).\label{eq:MAC_Metric}
\end{equation}
We assume that ties are resolved uniformly at random. Similarly to the single-user case,
optimal ML decoding is recovered by setting $q(x_1,x_2,y)=W(y|x_1,x_2)$.

An error is said to have
occurred if the estimate $(\hat{m}_{1},\hat{m}_{2})$ differs from
$(m_{1},m_{2})$. The error probability for a given pair of codebooks $(\Cc_{1},\Cc_{2})$
is denoted by $p_e(\Cc_{1},\Cc_{2})$, and the error probability for a given 
random-coding ensemble is denoted by $\pebar(n,M_{1},M_{2})$.  We define
achievable rate pairs, error exponents, and ensemble tightness analogously to
the single-user setting.  

\subsection{Notation} \label{sub:MAC_NOTATION}

We use bold symbols for vectors (e.g.~$\xv_1$, $\yv$), and denote the corresponding
$i$-th entry using a non-bold symbol with a subscript (e.g.~$x_{1,i}$, $y_i$).
All logarithms have base $e$. Moreover, all rates are in units of nats
except in the examples, where bits are used. We define $[c]^{+}=\max\{0,c\}$,
and denote the indicator function by $\openone\{\cdot\}$.

The symbol $\sim$ means ``distributed as''.
The set of all probability distributions on an alphabet, say $\Xc$,
is denoted by $\Pc(\Xc)$, and the set of all empirical
distributions on a vector in $\Xc^{n}$ (i.e.~types \cite[Ch. 2]{CsiszarBook}, \cite{GallagerCC}) 
is denoted by $\Pc_{n}(\Xc)$.  Similar notations $\Pc(\Yc|\Xc)$
and $\Pc_n(\Yc|\Xc)$ are used for conditional distributions, with the
latter adopting the convention that the empirical distribution of $\yv$
given $\xv$ is uniform for values of $x$ that do not appear in $\xv$.
For a given $Q\in\Pc_{n}(\Xc)$,
the type class $T^{n}(Q)$ is defined to be the set of all sequences
in $\Xc^{n}$ with type $Q$.  For a given joint type $P_{XY}\in\Pc_{n}(\Xc\times\Yc)$  
and sequence $\xv \in T^{n}(P_{X})$, the conditional type
class $T_{\xv}^{n}(P_{XY})$ is defined to be the set of all sequences
$\yv$ such that $(\xv,\yv) \in T^{n}(P_{XY})$. 

The probability of an event is denoted by $\PP[\cdot]$. The
marginals of a joint  distribution $P_{XY}(x,y)$ are denoted by $P_{X}(x)$
and $P_{Y}(y)$. We write $P_{X}=\Ptilde_{X}$ to denote element-wise
equality between two probability distributions on the same alphabet.
Expectation with respect to a joint distribution $P_{XY}(x,y)$ is
denoted by $\EE_{P}[\cdot]$, or simply $\EE[\cdot]$ when the associated 
probability distribution is understood from the context.  Similarly, mutual 
information with respect to $P_{XY}$ is written as $I_{P}(X;Y)$, or simply 
$I(X;Y)$. Given a 
distribution $Q(x)$ and conditional distribution $W(y|x)$, we write $Q\times W$
to denote the joint distribution defined by $Q(x)W(y|x)$.

For two positive sequences $f_{n}$ and $g_{n}$, we write $f_{n}\doteq g_{n}$
if $\lim_{n\to\infty}\frac{1}{n}\log\frac{f_{n}}{g_{n}}=0$, $f_{n}\,\dot{\le}\,g_{n}$
if $\limsup_{n\to\infty}\frac{1}{n}\log\frac{f_{n}}{g_{n}}\le0$, and analogously for $\dot{\ge}$. We make use of the standard asymptotic notations $O(\cdot)$, 
$o(\cdot)$ and $\Omega(\cdot)$.  When studying the MAC, we index the users as
$\nu=1,2$, and let $\nu^c$ denote the unique index differing from $\nu$.


\section{Multiple-Access Channel} \label{sec:MU_MAC}

In this section, we study the mismatched multiple-access channel introduced
in Section \ref{sub:MU_SETUP}.  We consider random coding, in which each
codeword of user $\nu={1,2}$ is generated independently according to some
distribution $P_{\Xv_{\nu}}$. We let $\Xv_{\nu}^{(i)}$ be the random variable 
corresponding to the $i$-th codeword of user $\nu$, yielding
\begin{multline}
    \Big(\{\Xv_{1}^{(i)}\}_{i=1}^{M_{1}},\{\Xv_{2}^{(j)}\}_{i=1}^{M_{2}}\Big)\sim\prod_{i=1}^{M_{1}}P_{\Xv_{1}}(\xv_{1}^{(i)})\prod_{j=1}^{M_{2}}P_{\Xv_{2}}(\xv_{2}^{(j)}). \label{eq:MAC_Distr1}
\end{multline}
We assume without loss of generality that message $(1,1)$
is transmitted, and write $\Xv_{1}$ and $\Xv_{2}$ in place of
$\Xv_{1}^{(1)}$ and $\Xv_{2}^{(1)}$. We write $\Xvbar_{1}$ and $\Xvbar_{2}$
to denote arbitrary codewords that are generated independently of $\Xv_{1}$ and $\Xv_{2}$. 
The random sequence at the output of the channel is denoted by $\Yv$.
It follows that
\begin{multline}
    (\Xv_{1},\Xv_{2},\Yv,\Xvbar_{1},\Xvbar_{2})\sim P_{\Xv_{1}}(\xv_{1})P_{\Xv_{2}}(\xv_{2})W^{n}(\yv|\xv_{1},\xv_{2}) \\ \times P_{\Xv_{1}}(\xvbar_{1})P_{\Xv_{2}}(\xvbar_{2}). \label{eq:MAC_VecDistr}
\end{multline}
For clarity of exposition, we focus
primarily on the case that there is no time-sharing (e.g.~see \cite{MACExponent4}).  
In Section \ref{sub:MAC_TIME_SHARING}, we discuss some of the corresponding
results with time-sharing.

We study the random-coding error probability by considering the following events:
\begin{tabbing}
    ~{\emph{(Type 1)}} ~~~ \= $\displaystyle \frac{q^n(\Xv_1^{(i)},\Xv_2,\Yv)}{q^n(\Xv_1,\Xv_2,\Yv)} \ge 1$ for some $i \ne 1$; \\
    ~{\emph{(Type 2)}} ~~~ \> $\displaystyle \frac{q^n(\Xv_1,\Xv_2^{(j)},\Yv)}{q^n(\Xv_1,\Xv_2,\Yv)} \ge 1$ for some $j \ne 1$; \\
    ~{\emph{(Type 12)}} ~~~ \> $\displaystyle \frac{q^n(\Xv_1^{(i)},\Xv_2^{(j)},\Yv)}{q^n(\Xv_1,\Xv_2,\Yv)} \ge 1$ for some $i \ne 1$, $j \ne 1$.
\end{tabbing}
We refer to these as error events, though they
do not necessarily imply decoder errors when the inequalities 
hold with equality, since we have assumed that the decoder resolves
ties uniformly at random.

The probabilities of the error events are denoted by 
$\peibar(n,M_{1})$, $\peiibar(n,M_{2})$ and $\peiiibar(n,M_{1},M_{2})$,
and the overall random-coding error probability is denoted by  
$\pebar(n,M_{1},M_{2})$.  Since breaking ties as errors increases
the error probability by at most a factor of two \cite{TwoChannels}, we have
\begin{equation}
    \frac{1}{2}\max\{\peibar,\peiibar,\peiiibar\}\le \pebar \le \peibar+\peiibar+\peiiibar.\label{eq:MAC_ErrorProbs}
\end{equation} 

\subsection{Exponents and Rates for the DM-MAC} \label{sub:MAC_DISCRETE}

In this subsection, we study the DM-MAC using the constant-composition ensemble. 
For $\nu=1,2$, we fix $Q_{\nu}\in\Pc(\Xc_{\nu})$ and let $P_{\Xv_{\nu}}$
be the uniform distribution on $T^{n}(Q_{\nu,n})$,
where $Q_{\nu,n}\in\Pc_{n}(\Xc_{\nu})$ is a type with the same support
as $Q_{\nu}$ such that $\max_{x_{\nu}}|Q_{\nu,n}(x_{\nu}) - Q_{\nu}(x_{\nu})| \le \frac{1}{n}$.
Thus,  
\begin{equation}
    P_{\Xv_{\nu}}(\xv_{\nu})=\frac{1}{|T^{n}(Q_{\nu,n})|}\openone\big\{\xv_{\nu}\in T^{n}(Q_{\nu,n})\big\}. \label{eq:MAC_Q_X1}
\end{equation}
Our analysis is based on the method of types \cite[Ch.~2]{CsiszarBook}.
Throughout the section, we write $f(\Qv)$ to denote a quantity
$f$ that depends on $Q_{1}$ and $Q_{2}$. Similarly, we write $f(\Qv_{n})$
to denote a quantity that depends on $Q_{1,n}$ and $Q_{2,n}$.

\subsubsection{Error Exponents} \label{sub:MAC_DM_EXPONENT}

The error exponents and achievable rates are expressed in terms
of the following sets ($\nu=1,2$):
\begin{multline}
    \SetScc(\Qv)\defeq\Big\{ P_{X_{1}X_{2}Y}\in\Pc(\Xc_{1}\times\Xc_{2}\times\Yc)\,:\,\\ P_{X_{1}}=Q_{1},\, P_{X_{2}}=Q_{2}\Big\}\label{eq:MAC_SetS}
\end{multline}
\vspace*{-3ex}
\begin{align}
    \SetTnucc(P_{X_{1}X_{2}Y})\defeq\bigg\{\Ptilde_{X_{1}X_{2}Y}\in\Pc(\Xc_{1}\times\Xc_{2}\times\Yc)\,: & \nonumber \\
    \Ptilde_{X_{\nu}}=P_{X_{\nu}},\Ptilde_{X_{\nu^c}Y}=P_{X_{\nu^c}Y}, & \nonumber \\
    \EE_{\Ptilde}[\log q(X_{1},X_{2},Y)] \ge\EE_{P}[\log q(X_{1},X_{2},Y)]\bigg\} \label{eq:MAC_SetT1}
\end{align}
\vspace*{-6ex}
\begin{align}
    \SetTiiicc(P_{X_{1}X_{2}Y})\defeq\bigg\{\Ptilde_{X_{1}X_{2}Y}\in\Pc(\Xc_{1}\times\Xc_{2}\times\Yc)\,: & \nonumber \\
    \Ptilde_{X_{1}}=P_{X_{1}},\Ptilde_{X_{2}}=P_{X_{2}},\Ptilde_{Y}=P_{Y}, \nonumber & \\
    \EE_{\Ptilde}[\log q(X_{1},X_{2},Y)]\ge\EE_{P}[\log q(X_{1},X_{2},Y)]\bigg\}, \label{eq:MAC_SetT12}
\end{align}
where we recall that for $\nu=1,2$,  $\nu^c$ denotes the unique element differing from $\nu$.

\begin{thm} \label{thm:MAC_Primal}
    For any mismatched DM-MAC, for the constant-composition 
    ensemble in \eqref{eq:MAC_Q_X1} with input distributions $Q_1$ and $Q_2$, the ensemble-tight error exponents are given as follows for $\nu=1,2$:
    \begin{gather} 
        \lim_{n\to\infty} -\frac{1}{n}\log\penubar(n,e^{nR_{\nu}}) = \Enurcc(\Qv,R_{\nu}) \label{eq:MAC_pe1_exp} \\
        \lim_{n\to\infty} -\frac{1}{n}\log\peiiibar(n,e^{nR_{1}},e^{nR_{2}}) = \Eiiircc(\Qv,R_{1},R_{2}), \label{eq:MAC_pe12_exp} 
    \end{gather}
    where
    \begin{align}
        & \Enurcc(\Qv,R_{\nu}) \defeq \min_{P_{X_{1}X_{2}Y}\in\SetScc(\Qv)}\min_{\Ptilde_{X_{1}X_{2}Y}\in\SetTnucc(P_{X_{1}X_{2}Y})} \nonumber \\
            & D(P_{X_{1}X_{2}Y}\|Q_{1}\times Q_{2}\times W) 
         +\big[I_{\Ptilde}(X_{\nu};X_{\nu^c},Y)-R_{\nu}\big]^{+} \label{eq:MAC_Er1_LM}
     \end{align}
     \vspace*{-5ex}
      \begin{align}
        &\Eiiircc(\Qv,R_{1},R_{2}) \defeq \min_{P_{X_{1}X_{2}Y}\in\SetScc(\Qv)}\min_{\Ptilde_{X_{1}X_{2}Y}\in\SetTiiicc(P_{X_{1}X_{2}Y})} \nonumber \\ 
        & D(P_{X_{1}X_{2}Y}\|Q_{1}\times Q_{2}\times W)  + \Big[\max\Big\{I_{\Ptilde}(X_{1};Y)-R_{1}, \nonumber \\
        & I_{\Ptilde}(X_{2};Y)-R_{2},D\big(\Ptilde_{X_{1}X_{2}Y}\|Q_{1}\times Q_{2}\times P_{Y}\big)-R_{1}-R_{2}\Big\}\Big]^{+}.\label{eq:MAC_Er12_LM}
    \end{align}
\end{thm}
\begin{proof}
    The random-coding error probabilities $\peibar$ and $\peiibar$ can be 
    handled similarly to the single-user setting \cite{JournalSU}.
    Furthermore, equivalent error exponents to \eqref{eq:MAC_Er1_LM} ($\nu=1,2$)
    were given in \cite{MACExponent5}. We therefore focus 
    on $\peiiibar$, which requires a more careful analysis. We first 
    rewrite
    \begin{align}
        &\peiiibar \nonumber \\ &=\EE\Bigg[\PP\Bigg[\bigcup_{i\ne1,j\ne1}\bigg\{\frac{q^{n}(\Xv_{1}^{(i)},\Xv_{2}^{(j)},\Yv)}{q^{n}(\Xv_{1},\Xv_{2},\Yv)}\ge1\bigg\}\bigg|\Xv_{1},\Xv_{2},\Yv\Bigg]\Bigg].\label{eq:MAC_pe12_ExactIsh}
    \end{align}
    in terms of the possible joint types of $(\Xv_{1},\Xv_{2},\Yv)$
    and $(\Xv_{1}^{(i)},\Xv_{2}^{(j)},\Yv)$.  To this end, we define 
    \begin{align}
        \SetSncc(\Qv_{n}) &\defeq\Big\{ P_{X_{1}X_{2}Y}\in\Pc_{n}(\Xc_{1}\times\Xc_{2}\times\Yc)\,:\, \nonumber \\
        & \hspace{10ex} P_{X_{1}}=Q_{1,n},\, P_{X_{2}}=Q_{2,n}\Big\}\label{eq:MAC_SetSn} \\
        \SetTiiincc(P_{X_{1}X_{2}Y}) &\defeq\SetTiiicc(P_{X_{1}X_{2}Y})\,\cap\,\Pc_{n}(\Xc_{1}\times\Xc_{2}\times\Yc).\label{eq:MAC_SetT12n}                    
    \end{align}
    Roughly speaking, $\SetSncc$ is the set of possible joint types of 
    $(\Xv_{1},\Xv_{2},\Yv)$, and $\SetTiiincc(P_{X_{1}X_{2}Y})$ is the set
    of types of $(\Xv_{1}^{(i)},\Xv_{2}^{(j)},\Yv)$ that lead to decoding errors 
    when $(\Xv_{1},\Xv_{2},\Yv)\in T^{n}(P_{X_{1}X_{2}Y})$. The constraints on 
    $P_{X_{\nu}}$ and $\Ptilde_{X_{\nu}}$ arise from the fact that we are using 
    constant-composition random coding, and the constraint
    $\EE_{\Ptilde}[\log q(X_{1},X_{2},Y)]\ge\EE_{P}[\log q(X_{1},X_{2},Y)]$
    holds if and only if $q^{n}(\xvbar_{1},\xvbar_{2},\yv)\ge q^{n}(\xv_{1},\xv_{2},\yv)$
    for $(\xv_{1},\xv_{2},\yv)\in T^{n}(P_{X_{1}X_{2}Y})$
    and $(\xvbar_{1},\xvbar_{2},\yv)\in T^{n}(\Ptilde_{X_{1}X_{2}Y})$.
    Fixing $P_{X_{1}X_{2}Y}\in\SetSncc(\Qv_{n})$ and letting $(\xv_{1},\xv_{2},\yv)$
    be an arbitrary triplet of sequences such that $(\xv_{1},\xv_{2},\yv)\in T^{n}(P_{X_{1}X_{2}Y})$,
    it follows that the event in \eqref{eq:MAC_pe12_ExactIsh} can be written as 
    \begin{equation}
        \bigcup_{i\ne1,j\ne1}\bigcup_{\Ptilde_{X_{1}X_{2}Y}\in\SetTiiincc}\bigg\{(\Xv_{1}^{(i)},\Xv_{2}^{(j)},\Yv)\in T^{n}(\Ptilde_{X_{1}X_{2}Y})\bigg\}.\label{eq:MAC_Event2}
    \end{equation}
    Expanding the probability and expectation in \eqref{eq:MAC_pe12_ExactIsh} in 
    terms of types, substituting \eqref{eq:MAC_Event2}, and interchanging 
    the order of the unions, we obtain
    \begin{align}
        &\peiiibar = \nonumber \\ &\sum_{P_{X_{1}X_{2}Y}\in\SetSncc(\Qv_{n})}\PP\big[(\Xv_{1},\Xv_{2},\Yv)\in T^{n}(P_{X_{1}X_{2}Y})\big] \nonumber \\ 
                  & \quad\times \PP\Bigg[\bigcup_{\Ptilde_{X_{1}X_{2}Y}\in\SetTiiincc(P_{X_{1}X_{2}Y})}\bigcup_{i\ne1,j\ne1} \nonumber \\
                  &\hspace{15ex}\bigg\{(\Xv_{1}^{(i)},\Xv_{2}^{(j)},\yv) \in T^{n}(\Ptilde_{X_{1}X_{2}Y})\bigg\}\Bigg] \\
                  &\doteq \max_{P_{X_{1}X_{2}Y}\in\SetSncc(\Qv_{n})}\PP\big[(\Xv_{1},\Xv_{2},\Yv)\in T^{n}(P_{X_{1}X_{2}Y})\big] \nonumber \\ 
                  & \quad\times \max_{\Ptilde_{X_{1}X_{2}Y}\in\SetTiiincc(P_{X_{1}X_{2}Y})}\PP\Bigg[\bigcup_{i\ne1,j\ne1} \nonumber \\
                  & \hspace{15ex}\bigg\{(\Xv_{1}^{(i)},\Xv_{2}^{(j)},\yv)\in T^{n}(\Ptilde_{X_{1}X_{2}Y})\bigg\}\Bigg], \label{eq:MAC_ExpProof2}
    \end{align}
    where $\yv$ is an arbitrary element of $T^{n}(P_{Y})$ (hence depending \emph{implicitly} on $P_{X_1X_2Y}$), and \eqref{eq:MAC_ExpProof2} follows from the 
    union bound and since the number of joint types is polynomial in $n$.
    
    By a standard property of types \cite[Ch.~2]{CsiszarBook}, the exponent 
    of the first probability in \eqref{eq:MAC_ExpProof2} is 
    given by $D(P_{X_{1}X_{2}Y} \| Q_1 \times Q_2 \times W)$, so it only remains to 
    determine the exponential behavior of the second probability.  To this end, we make
    use of Lemma \ref{lem:MU_MatchingInd} in Appendix \ref{sub:MU_BOUNDS} with $Z_{1}(i)=\Xv_{1}^{(i)}$, 
    $Z_{2}(j)=\Xv_{2}^{(j)}$, $\Ac = T_{\yv}^{n}(\Ptilde_{X_{1}X_{2}Y})$,
    $\Ac_{1} = T_{\yv}^{n}(\Ptilde_{X_{1}Y})$ and $\Ac_{2} = T_{\yv}^{n}(\Ptilde_{X_{2}Y})$.
    Using \eqref{eq:MU_MatchingUB}--\eqref{eq:MU_MatchingLB}
    and standard properties of types \cite[Ch.~2]{CsiszarBook}, it follows
    that the second probability in \eqref{eq:MAC_ExpProof2} has an exponent of 
    \begin{multline}
        \Big[\max\Big\{I_{\Ptilde}(X_{1};Y)-R_{1},I_{\Ptilde}(X_{2};Y)-R_{2}, \\ D\big(\Ptilde_{X_{1}X_{2}Y}\|Q_{1}\times Q_{2}\times P_{Y}\big)-R_{1}-R_{2}\Big\}\Big]^{+}. \label{eq:MAC_pE_exp}
    \end{multline}
    Upon substituting \eqref{eq:MAC_pE_exp} into \eqref{eq:MAC_ExpProof2}, it
    only remains to replace the sets $\SetSncc$ and $\SetTiiincc$ by $\SetScc$ 
    and $\SetTiiicc$ respectively.  This is seen to be valid since the underlying objective function
    is continuous in $\Ptilde_{X_1X_2Y}$, and since any joint distribution has a corresponding
    joint type which is within $\frac{1}{n}$ in each value of the probability mass function.
    See the discussion around \cite[Eq. (30)]{DyachkovCC} for the analogous continuity argument in the single-user setting.
\end{proof}
Theorem \ref{thm:MAC_Primal} and \eqref{eq:MAC_ErrorProbs} reveal that
the overall ensemble-tight error exponent is given by
\begin{multline}
    \Ercc(\Qv,R_{1},R_{2}) \defeq\min\Big\{ \Eircc(\Qv,R_{1}), \\ \Eiircc(\Qv,R_{2}),\Eiiircc(\Qv,R_{1},R_{2})\Big\}. \label{eq:MAC_Er_LM}
\end{multline}

The proof of Theorem \ref{thm:MAC_Primal} made use of the refined union bound
given in Lemma \ref{lem:MU_MatchingInd}.  If we had instead used the standard
truncated union bound in \eqref{eq:MU_TruncatedBound}, we would have obtained 
the weaker type-12 exponent
\begin{multline}
    \Erccprime(\Qv,R_{1},R_{2})\defeq\min_{P_{X_{1}X_{2}Y}\in\SetScc(\Qv)}\min_{\Ptilde_{X_{1}X_{2}Y}\in\SetTiiicc(P_{X_{1}X_{2}Y})}\\
    D(P_{X_{1}X_{2}Y}\|Q_{1}\times Q_{2}\times W)+\big[D(\Ptilde_{X_{1}X_{2}Y}\|Q_{1}\times Q_{2}\times P_{Y}) \\ -(R_{1}+R_{2})\big]^{+},\label{eq:MAC_Er12'}
\end{multline}
which coincides with an achievable exponent given in \cite{MACExponent5}.
 
\subsubsection{Achievable Rate Region} \label{sub:MAC_RATE}

The following theorem is a direct consequence of Theorem \ref{thm:MAC_Primal}, 
and provides an alternative proof of Lapidoth's ensemble-tight
achievable rate region \cite{MacMM}.
\begin{thm} \label{thm:MAC_Rate} 
    The overall error exponent $\Ercc(\Qv,R_{1},R_{2})$ in \eqref{eq:MAC_Er_LM}
    is positive for all rate pairs $(R_{1},R_{2})$ in the interior of
    $\RegLM(\Qv)$, defined to be the set of all rate pairs $(R_{1},R_{2})$ satisfying the following for $\nu=1,2$:
    \begin{align}
    & \quad R_{\nu}        \le\min_{\Ptilde_{X_{1}X_{2}Y}\in\SetTnucc(Q_{1}\times Q_{2}\times W)}I_{\Ptilde}(X_{\nu};X_{\nu^c},Y) \label{eq:MAC_R1_LM} \\
        & R_{1}+R_{2}  \le\min_{\substack{\Ptilde_{X_{1}X_{2}Y}\in\SetTiiicc(Q_{1}\times Q_{2}\times W) \\
        I_{\Ptilde}(X_{1};Y)\le R_{1},\, I_{\Ptilde}(X_{2};Y)\le R_{2}}} \nonumber \\ & \hspace{22ex}D(\Ptilde_{X_{1}X_{2}Y}\|Q_{1}\times Q_{2}\times P_{Y}).\label{eq:MAC_R12_LM}
    \end{align}
\end{thm}
\begin{proof}
    The conditions in \eqref{eq:MAC_R1_LM}--\eqref{eq:MAC_R12_LM} are
    obtained from \eqref{eq:MAC_Er1_LM}--\eqref{eq:MAC_Er12_LM}
    respectively. Focusing on \eqref{eq:MAC_R12_LM}, we see that the
    objective in \eqref{eq:MAC_Er12_LM} is always positive when $D(P_{X_{1}X_{2}Y}\|Q_{1}\times Q_{2}\times W)>0$,
    $I_{\Ptilde}(X_{1};Y)>R_{1}$ or $I_{\Ptilde}(X_{2};Y)>R_{2}$.
    Moreover, by a similar argument to \cite[Lemma 1]{Csiszar2}, the
    right-hand side of \eqref{eq:MAC_Er12_LM}, with only the second
    minimization kept, is continuous as a function of $P_{X_1X_2Y}$ when
    restricted to distributions with the same support as $Q_1 \times Q_2 \times W$.
    Hence, we may substitute $Q_1 \times Q_2 \times W$
    for $P_{X_1X_2Y}$ (thus forcing the first divergence to zero) and 
    introduce the constraints $I_{\Ptilde}(X_{1};Y)\le R_{1}$
    and $I_{\Ptilde}(X_{2};Y)\le R_{2}$ to obtain the condition in \eqref{eq:MAC_R12_LM}.
\end{proof}

Using a time-sharing argument \cite{NetworkBook,MacMM} (see
also Section \ref{sub:MAC_TIME_SHARING}), it follows from Theorem
\ref{thm:MAC_Rate} that we can achieve any rate pair in the convex hull of
$\bigcup_{\Qv}\RegLM(\Qv)$, where the union is over all distributions $Q_{1}$ and $Q_{2}$ on
$\Xc_{1}$ and $\Xc_{2}$ respectively.

Using a similar argument to the proof of Theorem \ref{thm:MAC_Rate},
we see that \eqref{eq:MAC_Er12'} yields the rate condition
\begin{equation}
    R_{1}+R_{2}\le\min_{\Ptilde_{X_{1}X_{2}Y}\in\SetTiiicc(Q_{1}\times Q_{2}\times W)}D(\Ptilde_{X_{1}X_{2}Y}\|Q_{1}\times Q_{2}\times P_{Y}).\label{eq:MAC_R12'}
\end{equation}
In Section \ref{sub:MAC_NUMERICAL}, we compare \eqref{eq:MAC_Er12_LM} and \eqref{eq:MAC_R12_LM} with the weaker
expressions in \eqref{eq:MAC_Er12'} and \eqref{eq:MAC_R12'}.

\subsection{Exponents and Rates for General Alphabets} \label{sub:MAC_GENERAL} 

In this section, we present equivalent dual expressions for the  
rates given in Theorem \ref{thm:MAC_Rate}, and extend them to the 
memoryless MAC with general alphabets. 
While we focus on rates for brevity,
dual expressions and continuous-alphabet generalizations 
for the exponents in Theorem \ref{thm:MAC_Primal} can
be obtained similarly; see \cite[Sec.~4.2]{Thesis} for details.

We use  the cost-constrained ensemble \cite{Variations,JournalSU},
defined as follows. We fix $Q_{1}\in\Pc(\Xc_{1})$ and $Q_{2}\in\Pc(\Xc_{2})$, and choose 
\begin{equation}
    P_{\Xv_{\nu}}(\xv_{\nu})=\frac{1}{\mu_{\nu,n}}\prod_{i=1}^{n}Q_{\nu}(x_{\nu,i})\openone\big\{\xv_{\nu}\in\Dc_{\nu,n}\big\} \label{eq:CNT_EnsemebleCost}
\end{equation}
for $\nu=1,2$, where $\mu_{\nu,n}$ is a normalizing constant, and
\begin{multline}
    \Dc_{\nu,n}\defeq\Bigg\{ \xv_{\nu}\,:\,\left|\frac{1}{n}\sum_{i=1}^{n}a_{\nu,l}(x_{\nu,i})-\phi_{\nu,l}\right|\le\frac{\delta}{n},\,\\ l=1,\dotsc,L_{\nu}\Bigg\}, \label{eq:CNT_Domain}
\end{multline}
where $\{a_{\nu,l}\}_{l=1}^{L_{\nu}}$ are auxiliary cost functions, $\delta$
is a positive constant, and $\phi_{\nu,l}\defeq\EE_{Q_{\nu}}[a_{\nu,l}(X_{\nu})]$. 
Thus, the codewords for user $\nu$ are 
constrained to satisfy $L_{\nu}$ cost constraints in which the empirical 
mean of $a_{\nu,l}(\cdot)$ is close to the true mean.  We allow each 
of the parameters to be optimized, including the cost functions.
The case $L_{\nu}=0$ should be understood as
corresponding to the case that $\Dc_{\nu,n}$ contains all $\xv_\nu$ 
sequences, thus recovering the i.i.d.~distribution studied in \cite{MACExponent2}.
In the case of finite input alphabets, the constant-composition ensemble 
can also be recovered by setting $L_{\nu} = |\Xc_{\nu}|$ and letting each 
auxiliary cost function be the indicator function of its argument equaling 
a given input symbol \cite{JournalSU}.

The cost-constrained ensemble has primarily been used with $L_{\nu}=1$ \cite{Gallager,Variations},
but the inclusion of multiple cost functions has proven beneficial in
the mismatched single-user setting \cite{JournalSU,PaperExpurg}. 
We will see that the use of multiple costs is beneficial for both
the matched and mismatched MAC.  We note that \emph{system} costs
(as opposed to the \emph{auxiliary} costs used here) can easily
be handled (e.g.~see \cite[Sec. VII]{JournalSU}, \cite{PaperExpurg}),
but in this paper we assume for simplicity that the channel is unconstrained.

The following proposition from \cite{JournalSU} will be useful.
\begin{prop} \label{prop:MAC_SubExp} 
    \emph{\cite[Prop. 1]{JournalSU}} For $\nu=1,2$, fix the input distribution $Q_{\nu}$
    along with $L_{\nu}$ and the auxiliary cost functions $\{a_{\nu,l}\}_{l=1}^{L_{\nu}}$.
    Then $\mu_{\nu,n}=\Omega(n^{-L_{\nu}/2})$ provided that $\EE_{Q_{\nu}}[a_{\nu,l}(X_{\nu})^2]<\infty$
    for $l=1,\dotsc,L_{\nu}$.
\end{prop}

The main result of this subsection is the following theorem.

\begin{thm} \label{thm:MAC_DualRate}
    The region $\RegLM(\Qv)$
    in \eqref{eq:MAC_R1_LM}--\eqref{eq:MAC_R12_LM} can be expressed
    as the set of rate pairs $(R_{1},R_{2})$ satisfying 
    \begin{align}
        R_{1} &\le \sup_{s\ge0,a_{1}(\cdot)}\EE\left[\log\frac{q(X_{1},X_{2},Y)^{s}e^{a_{1}(X_{1})}}{\EE\big[q(\Xbar_{1},X_{2},Y)^{s}e^{a_{1}(\Xbar_{1})}\,|\, X_{2},Y\big]}\right]\label{eq:MAC_R1_DualLM} \\
        R_{2} &\le \sup_{s\ge0,a_{2}(\cdot)}\EE\left[\log\frac{q(X_{1},X_{2},Y)^{s}e^{a_{2}(X_{2})}}{\EE\big[q(X_{1},\Xbar_{2},Y)^{s}e^{a_{2}(\Xbar_{2})}\,|\, X_{1},Y\big]}\right],\label{eq:MAC_R2_DualLM}
    \end{align}
    and at least one of
    \begin{align}
        & R_{1} \le \sup_{\rho_{2}\in[0,1],s\ge0,a_{1}(\cdot),a_{2}(\cdot)} - \rho_{2}R_{2} \nonumber \\
        &\hspace*{-0.7ex}+\EE\Bigg[\log\frac{\big(q(X_{1},X_{2},Y)^{s}e^{a_{2}(X_{2})}\big)^{\rho_{2}}e^{a_{1}(X_{1})}}{\EE\Big[\Big(\EE\big[q(\Xbar_{1},\Xbar_{2},Y)^{s}e^{a_{2}(\Xbar_{2})}\,\big|\,\Xbar_{1}\big]\Big)^{\rho_{2}}e^{a_{1}(\Xbar_{1})}\big|Y\Big]}\Bigg] \label{eq:MAC_R12_1_DualLM} \\
        & R_{2} \le \sup_{\rho_{1}\in[0,1],s\ge0,a_{1}(\cdot),a_{2}(\cdot)} - \rho_{1}R_{1} \nonumber \\ &\hspace*{-0.7ex}+\EE\Bigg[\log\frac{\big(q(X_{1},X_{2},Y)^{s}e^{a_{1}(X_{1})}\big)^{\rho_{1}}e^{a_{2}(X_{2})}}{\EE\Big[\Big(\EE\big[q(\Xbar_{1},\Xbar_{2},Y)^{s}e^{a_{1}(\Xbar_{1})}\,\big|\,\Xbar_{2}\big]\Big)^{\rho_{1}}e^{a_{2}(\Xbar_{2})}\big|Y\Big]}\Bigg], \label{eq:MAC_R12_2_DualLM}
    \end{align}
    where $(X_{1},X_{2},Y,\Xbar_{1},\Xbar_{2})$ is distributed as $Q_{1}(x_{1})Q_{2}(x_{2})W(y|x_{1},x_{2})Q_{1}(\xbar_{1})Q_{2}(\xbar_{2})$.
    
    Moreover, this region is achievable for any memoryless MAC (possibly
    having infinite or continuous alphabets) and any pair $(Q_1,Q_2)$, where
    each supremum is subject to $\EE_{Q_\nu}[a_{\nu}(X_{\nu})^2]<\infty$ ($\nu=1,2$).
    Any point in the region can be achieved using 
    cost-constrained coding with $L_{1}=L_{2}=3$.
\end{thm}
\begin{proof}
    The equivalence of this rate region to \eqref{eq:MAC_R1_LM}--\eqref{eq:MAC_R12_LM} 
    is proved in Appendix \ref{sub:MAC_PROOFS}.  Here we prove the second claim of the
    theorem by providing a direct derivation.

    The key initial step is to obtain the following non-asymptotic bound
    on the type-12 error event, holding for any codeword distributions $P_{\Xv_{1}}$ and $P_{\Xv_{2}}$:
        \begin{gather}
            \peiiibar(n,M_{1},M_{2})\le\min_{\nu=1,2} \rcuiiinu(n,M_{1},M_{2}), \label{eq:MAC_e12_Bound}
        \end{gather}
    where for $\nu=1,2$ we define
        \begin{align}
            & \rcuiiinu(n,M_{1},M_{2})\defeq \nonumber \\
            &~~~ \EE\Bigg[\min\Bigg\{ 1,(M_{\nu}-1)\EE\Bigg[\min\bigg\{1,(M_{\nu^c}-1) \nonumber \\ &~~~\PP\bigg[\frac{q^{n}(\Xvbar_{1},\Xvbar_{2},\Yv)}{q^{n}(\Xv_{1},\Xv_{2},\Yv)}\ge1\,\bigg|\,\Xvbar_{\nu}\bigg]\bigg\}\,\bigg|\,\Xv_{1},\Xv_{2},\Yv\Bigg]\Bigg\} \Bigg]. \label{eq:MAC_RCU12_1}
        \end{align}
    To prove this, we first write
        \begin{align}
            &\peiiibar \nonumber \\ &=\PP\Bigg[\bigcup_{i\ne1,j\ne1}\bigg\{\frac{q^{n}(\Xv_{1}^{(i)},\Xv_{2}^{(j)},\Yv)}{q^{n}(\Xv_{1},\Xv_{2},\Yv)}\ge1\bigg\}\Bigg]\label{eq:MAC_pe12_Exact1} \\
            & =\EE\Bigg[\PP\Bigg[\bigcup_{i\ne1,j\ne1}\bigg\{\frac{q^{n}(\Xv_{1}^{(i)},\Xv_{2}^{(j)},\Yv)}{q^{n}(\Xv_{1},\Xv_{2},\Yv)}\ge1\bigg\}\Bigg|\Xv_{1},\Xv_{2},\Yv\Bigg]\Bigg].\label{eq:MAC_pe12_Exact2} 
        \end{align}
    We obtain the above-mentioned bounds by applying Lemma 
    \ref{lem:MU_UpperInd} in Appendix \ref{sub:MU_BOUNDS} to the union in \eqref{eq:MAC_pe12_Exact2} (with
    $Z_{1}(i)=\Xv_{1}^{(i)}$ and $Z_{2}(j)=\Xv_{2}^{(j)}$), and then writing $\min\{1,\alpha,\beta\} \le \min\{1,\alpha\}$ and $\min\{1,\alpha,\beta\} \le \min\{1,\beta\}$.

    Define $Q_{\nu}^{n}(\xv_{\nu})\defeq\prod_{i=1}^{n}Q_{\nu}(x_{\nu,i})$
    for $\nu=1,2$. Expanding \eqref{eq:MAC_RCU12_1} and applying Markov's
    inequality and $\min\{1,\alpha\}\le\alpha^{\rho}$ ($0\le\rho\le1$),
    we obtain\footnote{In the case of continuous alphabets, the summations should be replaced by integrals as necessary.}
    \begin{multline}
        \rcuiiia(n,M_{1}) \\
            \le \sum_{\xv_{1},\xv_{2},\yv}P_{\Xv_{1}}(\xv_{1})P_{\Xv_{2}}(\xv_{2})W^{n}(\yv|\xv_{1},\xv_{2}) \Bigg(M_{1}\sum_{\xvbar_{1}}P_{\Xv_{1}}(\xvbar_{1}) \\ \times
            \bigg(M_{2}\frac{\sum_{\xvbar_{2}}P_{\Xv_{2}}(\xvbar_{2})q^{n}(\xvbar_{1},\xvbar_{2},\yv)^{s}}{q^{n}(\xv_{1},\xv_{2},\yv)^{s}}\bigg)^{\rho_{2}}\Bigg)^{\rho_{1}}\label{eq:CNT_CostDer2}    
    \end{multline}
    for any $\rho_{1}\in[0,1]$, $\rho_{2}\in[0,1]$ and $s\ge0$. For
    $\nu=1,2$, we let $a_{\nu}(x)$ be one of the three
    cost functions in the ensemble, and we define $a_{\nu}^{n}(\xv_{\nu})\defeq\sum_{i=1}^{n}a_{\nu}(x_{\nu,i})$
    and $\phi_{\nu}\defeq\EE_{Q_{\nu}}[a_{\nu}(X_{\nu})]$.
    In accordance with the theorem statement, we assume that 
    $\EE_{Q_{\nu}}[a_{\nu}(X_{\nu})]^2 < \infty$, so that
    Proposition \ref{prop:MAC_SubExp} holds.
    Using the bounds on the cost functions in \eqref{eq:CNT_Domain}, we can
    weaken \eqref{eq:CNT_CostDer2} to 
    \begin{multline}
        \rcuiiia(n,M_{1})\le  e^{2\delta(\rho_{1}+\rho_{1}\rho_{2}+1)} \\ \times \sum_{\xv_{1},\xv_{2},\yv}P_{\Xv_{1}}(\xv_{1})P_{\Xv_{2}}(\xv_{2})
        W^{n}(\yv|\xv_{1},\xv_{2})\Bigg(M_{1}\sum_{\xvbar_{1}}P_{\Xv_{1}}(\xvbar_{1}) \\
        \times\bigg(M_{2}\frac{\sum_{\xvbar_{2}}P_{\Xv_{2}}(\xvbar_{2})q^{n}(\xvbar_{1},\xvbar_{2},\yv)^{s}e^{a_{2}^{n}(\xvbar_{2})}}{q^{n}(\xv_{1},\xv_{2},\yv)^{s}e^{a_{2}^{n}(\xv_{2})}}\bigg)^{\rho_{2}}\frac{e^{a_{1}^{n}(\xvbar_{1})}}{e^{a_{1}^{n}(\xv_{1})}}\Bigg)^{\rho_{1}}.\label{eq:CNT_CostDer3} 
    \end{multline}
    We upper bound \eqref{eq:CNT_CostDer3} by substituting \eqref{eq:CNT_EnsemebleCost}
    and replacing the summations over $\Dc_{\nu,n}$ by summations
    over all sequences on $\Xc_{\nu}^{n}$. Writing the
    resulting terms (e.g.~ $W^{n}(\yv|\xv_{1},\xv_{2})$)
    as a product from $1$ to $n$ and taking the supremum over $(s,\rho_{1},\rho_{2})$
    and the cost functions, we obtain a bound whose exponent is
    \begin{equation}
        \max_{\rho_{1}\in[0,1],\rho_{2}\in[0,1]}\Eiiiazcost(\Qv,\rho_{1},\rho_{2})-\rho_{1}(R_{1}+\rho_{2}R_{2}),
    \end{equation}
    where 
    \begin{multline}
        \Eiiiazcost(\Qv,\rho_{1},\rho_{2})\defeq\sup_{s\ge0,a_{1}(\cdot),a_{2}(\cdot)}
        \\ -\log\EE\Bigg[\bigg(\EE\bigg[\bigg(\frac{\EE\big[q(\Xbar_{1},\Xbar_{2},Y)^{s}e^{a_{2}(\Xbar_{2})}\,|\,\Xbar_{1}\big]}{q(X_{1},X_{2},Y)^{s}e^{a_{2}(X_{2})}}\bigg)^{\rho_{2}} \\ \times\frac{e^{a_{1}(\Xbar_{1})}}{e^{a_{1}(X_{1})}}\,\bigg|\, X_{1},X_{2},Y\bigg]\bigg)^{\rho_{1}}\Bigg]\label{eq:MAC_E0_12_1_LM}
    \end{multline}
    We obtain the condition in \eqref{eq:MAC_R12_1_DualLM} by taking the derivative
    of $\Eiiiazcc$ at zero, analogously to the proof of Theorem \ref{thm:MAC_DualRate}.
    We obtain \eqref{eq:MAC_R12_2_DualLM} analogously by starting with $\rcuiiib$ 
    in place of $\rcuiiia$, and we obtain \eqref{eq:MAC_R1_DualLM}--\eqref{eq:MAC_R2_DualLM}
    via a simpler analysis following the standard single-user setting \cite{JournalSU}.
    
    Finally, we note that $L_{1}=L_{2}=3$ suffices due to the fact that the
    cost functions used in deriving \eqref{eq:MAC_R12_1_DualLM}--\eqref{eq:MAC_R12_2_DualLM}
    may coincide, since the theorem statement only requires one of the two to hold.
\end{proof}

Theorem \ref{thm:MAC_DualRate} extends Lapidoth's MAC rate region to general
alphabets, analogously to the extension of the single-user LM rate to
general alphabet by Ganti \emph{et al.}~\cite{MMRevisited}.  Compared to 
the single-user setting, the extension is non-trivial, requiring refined
union bounds, as well as a technique for handling the two additional in constraints
in \eqref{eq:MAC_R12_LM} one at a time, thus leading to two type-12
conditions in \eqref{eq:MAC_R12_1_DualLM}--\eqref{eq:MAC_R12_2_DualLM}.

\subsection{Matched MAC Error Exponent} \label{sec:MAC_MATCHED}

Here we apply our results to the setting of ML decoding,
where $q(x_{1},x_{2},y)=W(y|x_{1},x_{2})$.  The best known
exponent for the constant-composition 
ensemble was derived by Liu and Hughes \cite{MACExponent4},
and was shown to yield a strict improvement over Gallager's 
exponent for the i.i.d.~ensemble \cite{MACExponent2} even
after the optimization of the input distributions.

We have seen that for a general decoding metric, the overall error exponent $\Ercc$ given in \eqref{eq:MAC_Er_LM}
may be reduced when $\Erccprime$ in \eqref{eq:MAC_Er12'} is used in place of $\Eiiircc$.
The following result shows that the resulting expressions are in fact
identical in the matched case.

\begin{thm} \label{thm:MAC_NaiveML}
    Under ML decoding (i.e.~$q(x_{1},x_{2},y)=W(y|x_{1},x_{2})$), we have
    for any input distributions $(Q_{1},Q_{2})$ and rates $(R_{1},R_{2})$ that
    \begin{multline}
        \min\big\{ \Eircc(\Qv,R_{1}),\Eiircc(\Qv,R_{2}),\Eiiircc(\Qv,R_{1},R_{2})\big\} \\ =\min\big\{ \Eircc(\Qv,R_{1}),\Eiircc(\Qv,R_{2}),\Erccprime(\Qv,R_{1},R_{2})\big\}. \label{eq:MAC_NaiveML}
    \end{multline} 
    Thus, both the left-hand side and right-hand side of \eqref{eq:MAC_NaiveML} equal the overall ensemble-tight error exponent.
\end{thm}
\begin{proof}
    See Appendix \ref{sub:MAC_PROOFS}.
\end{proof}
While it is possible that $\Eiiircc>\Erccprime$
under ML decoding, Theorem \ref{thm:MAC_NaiveML} shows that this
never occurs in the region where $\Eiiircc$ achieves the minimum
in \eqref{eq:MAC_Er_LM}. Thus, combining Theorem \ref{thm:MAC_NaiveML}
and Theorem  \ref{thm:MAC_Primal}, we conclude that the exponent
given in \cite{MACExponent4} is ensemble-tight for the constant-composition
ensemble under ML decoding.

In \cite[Sec.~4.2.4]{Thesis}, \cite{PaperITA}, we show that the error exponent of \cite{MACExponent4}
admits a dual form resembling the i.i.d.~exponent of Gallager \cite{MACExponent2},
but with additional optimization parameters $a_1(\cdot)$ and $a_2(\cdot)$
that are functions of the input alphabets $\Xc_1$ and $\Xv_2$.  
As usual, this dual form can also 
be derived directly via the cost-constrained ensemble, with the analysis
remaining valid for infinite and continuous alphabets.

\subsection{Time-Sharing} \label{sub:MAC_TIME_SHARING}

Thus far, we have focused on the standard random coding ensemble 
described by \eqref{eq:MAC_Distr1}, where the codewords are  
independent.  It is well-known that even in the matched case, the
union of the resulting achievable rate regions over all $(Q_1,Q_2)$
may be non-convex, and time-sharing is needed to achieve the rest
of the capacity region \cite{MACNonConvex}.  There are two distinct ways of doing so:
(i) With explicit time-sharing, one splits the block of length $n$ into
two or more smaller blocks, and uses separate codebooks within each
block; (ii) With coded time-sharing, one still generates a single 
codebook, but the codewords are \emph{conditionally} independent
given some time-sharing sequence $\Uv$ on a time-sharing alphabet $\Uc$.
In particular, in the case of constant-composition random coding, one
may let $\Uv$ be uniform on a type class corresponding to $Q_U \in \Pc(\Uc)$, 
and let each $\Xv_{\nu}$ be uniform on a conditional type class corresponding
to $Q_{\nu} \in \Pc(\Xc_{\nu}\,|\,\Uc)$. 

While both of these schemes yield the entire capacity region in the matched
case \cite[Ch.~4]{NetworkBook}, the coded time-sharing approach is generally
preferable in terms of exponents \cite{MACExponent4}.  Intuitively,
this is because explicit time-sharing shortens the effective block length,
thus diminishing the exponent.  

Surprisingly, however, explicit time-sharing can outperform 
coded time-sharing in the mismatched case, even in terms of the 
achievable rate region.  This is most easily understood via the
dual-domain expressions, and for concreteness we consider the
case $|\Uc|=2$ with $Q_U=(\lambda,1-\lambda)$.  Let $I_{1}(\Qv,s)$
denote the right-hand side of \eqref{eq:MAC_R1_DualLM} with a fixed value 
of $s$ in place of the supremum.  Using explicit time-sharing with two different
input distribution pairs $\Qv^{(1)}$ and $\Qv^{(2)}$, the condition corresponding 
to \eqref{eq:MAC_R1_DualLM} is given by
\begin{equation}
    R_{1}\le\lambda\sup_{s\ge0}I_{1}(\Qv^{(1)},s)+\big(1-\lambda\big)\sup_{s\ge0}I_{1}(\Qv^{(2)},s), \label{eq:TS_ExplicitR}
\end{equation}
whereas coded time-sharing only permits
\begin{equation}
    R_{1}\le\sup_{s\ge0}\left(\lambda I_{1}(\Qv^{(1)},s)+\big(1-\lambda\big) I_{1}(\Qv^{(2)},s)\right). \label{eq:TS_CodedR}
\end{equation}
These are obtained using similar arguments to the case without time-sharing; see \cite[Sec.~4.2.5]{Thesis} for
further details.  Similar observations apply for the other rate conditions,
including the parameters $\rho_1$ and $\rho_2$ in \eqref{eq:MAC_R12_1_DualLM}--\eqref{eq:MAC_R12_2_DualLM}.

It is evident from \eqref{eq:TS_ExplicitR} (and the other analogous rate conditions)
that explicit time-sharing between two points can be used to obtain any pair $(R_1,R_2)$
on the line connecting two achievable pairs corresponding to $\Qv^{(1)}$ and $\Qv^{(2)}$.
On the other hand, the same is only true for coded time-sharing if there
exists a single parameter $s$ simultaneously maximizing both terms
in the objective function of \eqref{eq:TS_CodedR} (and similarly for the
other rate conditions), which is not the case in general.

Building on this insight, in the following section, we compare two forms of
superposition coding for single-user channels.  The standard version can be viewed as analogous to coded time-sharing, whereas the refined version can be viewed as analogous to explicit time-sharing.  As a result, the latter can lead to higher achievable rates.

\section{Superposition Coding} \label{sec:MU_SC}

In this section, we turn to the single-user mismatched channel introduced
in Section \ref{sub:SU_SETUP}, and consider multiuser coding schemes that
can improve on standard schemes with independent codewords.  Some
numerical examples are given in Section \ref{sec:MU_COMPARISONS}.

\subsection{Standard Superposition Coding}

We first discuss a standard form of superposition coding that has had
extensive application in degraded broadcast channels \cite{BCDegraded2,BCDegMsg1,BCExp1} and other network information theory problems \cite{NetworkBook}.
This ensemble was studied in the context of mismatched decoding in 
\cite{MMSomekh,Thesis}, so we do not repeat the details here.

The parameters of the ensemble are an auxiliary alphabet $\Uc$, an auxiliary 
codeword distribution $P_{\Uv}$, and a conditional codeword distribution $P_{\Xv|\Uv}$. 
We fix two rates $R_{0}$ and $R_{1}$.  An auxiliary codebook $\{\Uv^{(i)}\}_{i=1}^{M_{0}}$ 
with $M_{0}\defeq\lfloor e^{nR_{0}} \rfloor$ codewords is generated at random, with each auxiliary codeword
independently distributed according to $P_{\Uv}$.  For each $i=1,\dotsc,M_{0}$, a codebook
$\{\Xv^{(i,j)}\}_{j=1}^{M_{1}}$ with $M_{1} \defeq \lfloor e^{nR_{1}} \rfloor$ codewords is
generated at random, with each codeword conditionally independently distributed according 
to $P_{\Xv|\Uv}$.  The message $m$ at the input to the encoder is indexed as
$(m_{0},m_{1})$, and for any such pair, the corresponding codeword
is $\Xv^{(m_{0},m_{1})}$.

The following achievable rate for DMCs is obtained using constant-composition coding with some
input distribution $Q_{UX} \in \Pc(\Uc \times \Xc)$, in which $P_{\Uv}$ is the uniform 
distribution on a type class corresponding to $Q_U$, and $P_{\Xv|\Uv}$ is the uniform 
distribution on a conditional type class corresponding to $Q_{X|U}$. We define the sets
\begin{equation}
	\SetScc(Q_{UX}) \defeq \big\{ P_{UXY}\in\Pc(\Uc\times\Xc\times\Yc)\,:\, P_{UX}=Q_{UX}\big\} \label{eq:SC_SetS}
\end{equation}
\begin{align}
    \SetTocc(P_{UXY}) & \defeq \Big\{\Ptilde_{UXY}\in\Pc(\Uc\times\Xc\times\Yc)\,:\,
                                \Ptilde_{UX}=P_{UX},\,\nonumber \\
                                & \hspace{-5ex} \Ptilde_{Y}=P_{Y}, \EE_{\Ptilde}[\log q(X,Y)]\ge\EE_{P}[\log q(X,Y)]\Big\}\label{eq:SC_SetT0} \\
    \SetTicc(P_{UXY}) & \defeq \Big\{\Ptilde_{UXY}\in\Pc(\Uc\times\Xc\times\Yc)\,:\,
                               \Ptilde_{UX}=P_{UX}, \nonumber \\
                               &\hspace{-8ex}\Ptilde_{UY}=P_{UY},\,\EE_{\Ptilde}[\log q(X,Y)]\ge\EE_{P}[\log q(X,Y)]\Big\}. \label{eq:SC_SetT1}  
\end{align}

\begin{thm} \emph{\cite{MMSomekh,Thesis}} \label{thm:SC_Rate_SC}
    Suppose that $W$ is a DMC.  For any finite auxiliary alphabet 
    $\Uc$, and input distribution $Q_{UX} \in \Pc(\Uc \times \Xc)$, the rate
    \begin{equation}
        R = R_0 + R_1
    \end{equation} 
    is achievable provided that $(R_0,R_1)$ satisfy
    \begin{align}
        R_{1}       & \le\min_{\Ptilde_{UXY}\in\SetTicc(Q_{UX}\times W)}I_{\Ptilde}(X;Y|U) \label{eq:SC_R1_CC} \\
        R_{0}+R_{1} & \le \min_{\substack{\Ptilde_{UXY}\in\SetTocc(Q_{UX}\times W)\\I_{\Ptilde}(U;Y)\le R_{0}}}I_{\Ptilde}(U,X;Y). \label{eq:SC_Rsum_CC}
    \end{align}
\end{thm}

This rate is also known to be tight with respect to the ensemble average \cite{MMSomekh,Thesis}.
It is known to be at least as high as Lapidoth's expurgated parallel coding
rate \cite{MacMM}, though it is not known whether the improvement can be strict.

Using similar steps to those in the previous section, one can obtain
the following equivalent dual form, which also remains valid in the case
of continuous alphabets \cite[Sec.~5.2.2]{Thesis}.

\begin{thm} \emph{\cite{Thesis}} \label{thm:SC_Rate_Dual}
    The achievable rate conditions in \eqref{eq:SC_R1_CC}--\eqref{eq:SC_Rsum_CC} 
    can be expressed as
    \begin{align}
        R_{1}&\le\sup_{s\ge0,a(\cdot,\cdot)}\EE\left[\log\frac{q(X,Y)^{s}e^{a(U,X)}}{\EE[q(\Xtilde,Y)^{s}e^{a(U,\Xtilde)}\,|\, U,Y]}\right]\label{eq:SC_R1_Dual} \\
        R_{0} &\le \sup_{\rho_{1}\in[0,1],s\ge0,a(\cdot,\cdot)} - \rho_{1}R_{1} \nonumber \\ &+\EE\left[\log\frac{\big(q(X,Y)^{s}e^{a(U,X)}\big)^{\rho_{1}}}{\EE\Big[\Big(\EE\big[q(\Xbar,Y)^{s}e^{a(\Ubar,\Xbar)}\,\big|\,\Ubar\big]\Big)^{\rho_{1}}\,\big|\,Y\Big]}\right] , \label{eq:SC_Rsum_Dual}
    \end{align}
    where $(U,X,Y,\Xtilde,\Ubar,\Xbar)$ is distributed as $Q_{UX}(u,x)W(y|x)Q_{X|U}(\xtilde|u)Q_{UX}(\ubar,\xbar)$.
\end{thm}

We observe that superposition coding has some similarity to the coded time-sharing
ensemble discussed in Section \ref{sub:MAC_TIME_SHARING}, in that both involve
generating codewords $\xv$ conditionally on auxiliary sequences $\uv$ according
to the uniform distribution on a type class.  We saw in Section \ref{sub:MAC_TIME_SHARING}
that better rates are in fact achieved by explicit time-sharing, in which one 
splits the block length into sub-blocks and codes individually on each one.
We now apply this approach to superposition coding, yielding
a refined ensemble that can lead to higher achievable rates
than the standard version.  

\subsection{Refined Superposition Coding} \label{sec:MU_RSC}

The ensemble is defined as follows. We fix a finite alphabet
$\Uc$, an input distribution $Q_{U}\in\Pc(\Uc)$ and the rates $R_{0}$
and $\{R_{1u}\}_{u\in\Uc}$. We write $M_{0} \defeq \lfloor e^{nR_{0}} \rfloor$
and $M_{1u} \defeq \lfloor e^{nR_{1u}} \rfloor$. We let $P_{\Uv}(\uv)$
be the uniform distribution on the type class $T^{n}(Q_{U,n})$,
where $Q_{U,n}$ is a type with the same support as $Q_U$ such that
$\max_u|Q_{U,n}(u) - Q_U(u)| \le \frac{1}{n}$.  We set
\begin{equation} 
    P_{\Uv}(\uv) =\frac{1}{|T^{n}(Q_{U,n})|}\openone\Big\{\uv\in T^{n}(Q_{U,n})\Big\} \label{eq:RSC_PU}
\end{equation}
and generate the length-$n$ auxiliary  
codewords $\{\Uv^{(i)}\}_{i=1}^{M_{0}}$ independently according to $P_{\Uv}$.  The 
difference compared to standard superposition coding is that the codewords are not generated conditionally
independently given $\Uv^{(i)}$.  Instead, we generate a number of partial codewords,
and construct the length-$n$ codeword by placing the entries of a partial codeword in the
indices where $\Uv$ takes a particular value.

More precisely, for each $u\in\Uc$, we define
\begin{equation}
    n_{u} \defeq Q_{U,n}(u)n\label{eq:RSC_BC_n1}
\end{equation}
and fix a partial codeword distribution $P_{\Xv_{u}}\in\Pc(\Xc^{n_{u}})$.
For each $i=1,\dotsc,M_{0}$ and $u\in\Uc$, we generate the length-$n_{u}$ 
partial codewords $\{\Xv_{u}^{(i,j_{u})}\}_{j_{u}=1}^{M_{1u}}$
independently according to $P_{\Xv_{u}}$. For example, when $\Uc=\{1,2\}$ we have
\begin{multline}
    \bigg\{\Big(\Uv^{(i)},\big\{\Xv_{1}^{(i,j_{1})}\big\}_{j_{1}=1}^{M_{11}},\big\{\Xv_{2}^{(i,j_{2})}\big\}_{j_{2}=1}^{M_{12}}\Big)\bigg\}_{i=1}^{M_{0}}  \\
    \sim\prod_{i=1}^{M_{0}}\bigg(P_{\Uv}(\uv^{(i)})\prod_{j_{1}=1}^{M_{11}}P_{\Xv_{1}}(\xv_{1}^{(i,j_{1})})\prod_{j_{2}=1}^{M_{12}}P_{\Xv_{2}}(\xv_{2}^{(i,j_{2})})\bigg). \label{eq:RSC_Distr}
\end{multline}
The message $m$ at the encoder is indexed as $(m_{0},m_{11},\dotsc,m_{1|\Uc|})$.
To transmit a given message, we treat $\Uv^{(m_{0})}$
as a time-sharing sequence; at the indices where $\Uv^{(m_{0})}$
equals $u$, we transmit the symbols of $\Xv_{u}^{(m_{0},m_{1u})}$.
There are $M=M_{0}\prod_{u}M_{1u}$ codewords, and hence the rate is
$R=R_{0}+\sum_{u}Q_{U,n}(u)R_{1u}$.  An example of the construction
of the codeword $\xv$ from the auxiliary sequence $\uv$ and partial codewords
$\xv_{1}$, $\xv_{2}$ and $\xv_{3}$ is shown in Figure \ref{fig:RSC_Codewords},
where we have $\Uc=\{1,2,3\}$ and $\Xc=\{a,b,c\}$.

\begin{figure}
    \begin{centering}
        \includegraphics[width=\columnwidth]{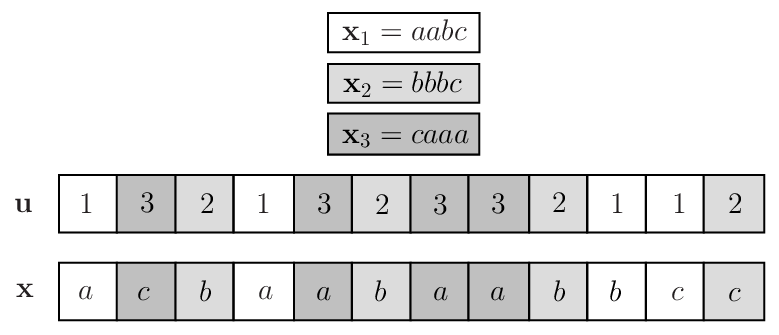}
        \par
    \end{centering}
    
    \caption{The construction of the codeword from the auxiliary sequence $\uv$
             and the partial codewords $\xv_{1}$, $\xv_{2}$ and $\xv_{3}$ for refined
             SC.  Here we have $\Uc = \{1,2,3\}$, $\Xc = \{a,b,c\}$, $n_1 = n_2 = n_3 = 4$, and $n=12$. }
    \label{fig:RSC_Codewords}
\end{figure}

While our main result is stated for an arbitrary finite alphabet $\Uc$, the analysis
will be presented for $\Uc=\{1,2\}$ for clarity.  We proceed by presenting several
definitions for this specific choice. We let $\Xi(\uv,\xv_{1},\xv_{2})$
denote the function for constructing the length-$n$ codeword from
the auxiliary sequence and partial codewords, and we write
\begin{equation}
    \Xv^{(i,j_{1},j_{2})}\defeq\Xi(\Uv^{(i)},\Xv_{1}^{(i,j_{1})},\Xv_{2}^{(i,j_{2})}). \label{eq:RSC_Xfunc}
\end{equation}
We let $\yv_{u}(\uv)$ denote the subsequence of $\yv$ corresponding to the indices
where $\uv$ equals $u$, and similarly for $\Yvu(\uv)$.

We assume without loss of generality that $(m_{0},m_{1},m_{2})=(1,1,1)$.
We let $\Uv$, $\Xv_{1}$, $\Xv_{2}$ and $\Xv$ be the codewords corresponding
to $(1,1,1)$, yielding $\Xv = \Xi(\Uv,\Xv_{1},\Xv_{2})$. We let $\Uvbar$, $\Xvbar_{1}$
and $\Xvbar_{2}$ be the codewords corresponding
to an arbitrary message with $m_{0}\ne1$. For the index $i$ corresponding
to $\Uvbar$, we write $\Xvbar_{1}^{(j_{1})}$,
$\Xvbar_{2}^{(j_{2})}$ and $\Xvbar^{(j_{1},j_{2})}$
in place of $\Xv_{1}^{(i,j_{1})}$, $\Xv_{2}^{(i,j_{2})}$
and $\Xv^{(i,j_{1},j_{2})}$ respectively. It follows that
$\Xvbar^{(j_{1},j_{2})}=\Xi(\Uvbar,\Xvbar_{1}^{(j_{1})},\Xvbar_{2}^{(j_{2})})$.

Upon receiving a realization $\yv$ of the output sequence $\Yv$, 
the decoder forms the estimate
\begin{align}
    & (\hat{m}_{0},\hat{m}_{1},\hat{m}_{2}) \nonumber \\
        &~~= \argmax_{(i,j_{1},j_{2})}q^{n}(\xv^{(i,j_{1},j_{2})},\yv) \\
        &~~= \argmax_{(i,j_{1},j_{2})}q^{n_{1}}\big(\xv_{1}^{(i,j_{1})},\yv_{1}(\uv^{(i)})\big)q^{n_{2}}\big(\xv_{2}^{(i,j_{2})},\yv_{2}(\uv^{(i)})\big),\label{eq:RSC_SplitMetric}
\end{align}
where the objective in \eqref{eq:RSC_SplitMetric} follows by separating
the indices where $u=1$ from those where $u=2$. By writing the objective
in this form, we see that for any given $i$, the pair
$(j_{1},j_{2})$ with the highest metric is the one for which $j_{1}$
maximizes $q^{n_{1}}(\xv_{1}^{(i,j_{1})},\yv_{1}(\uv^{(i)}))$
and $j_{2}$ maximizes $q^{n_{2}}(\xv_{2}^{(i,j_{2})},\yv_{2}(\uv^{(i)}))$.
We thus consider three error events:
\begin{tabbing}
    ~~{\emph{(Type 0)}}~~~ \= $\displaystyle \frac{q^n(\Xv^{(i,j_1,j_2)},\Yv)}{q^n(\Xv,\Yv)} \ge 1$ for some $i \ne 1$, $j_1$, $j_2$; \\
    ~~{\emph{(Type 1)}}~~~ \> $\displaystyle \frac{q^{n_1}(\Xvbar_1^{(1,j_1)},\Yvi(\Uv))}{q^{n_1}(\Xv_1,\Yvi(\Uv))} \ge 1$ for some $j_1 \ne 1$; \\
    ~~{\emph{(Type 2)}}~~~ \> $\displaystyle \frac{q^{n_2}(\Xvbar_2^{(1,j_2)},\Yvii(\Uv))}{q^{n_2}(\Xv_2,\Yvii(\Uv))} \ge 1$ for some $j_2 \ne 1$.
\end{tabbing} 
The corresponding probabilities are denoted by $\peobar(n,M_{0},M_{11},M_{12})$,
$\peibar(n,M_{11})$ and $\peiibar(n,M_{12})$ respectively.  Analogously to
\eqref{eq:MAC_ErrorProbs}, the overall random-coding error
probability $\pebar(n,M_{0},M_{11},M_{12})$ satisfies
\begin{equation}
    \frac{1}{2}\max\{\peobar,\peibar,\peiibar\} \le \pebar \le \peobar + \peibar + \peiibar.
\end{equation}
While our analysis of the error probability will yield non-asymptotic bounds and 
error exponents as intermediate steps, we focus on the resulting achievable rates
for clarity.

\subsection{Rates for DMCs} \label{sec:RSC_RatesDMC}

In this subsection, we assume that the channel is a DMC. 
We fix a joint distribution $Q_{UX}$, and let $Q_{UX,n}$ be a
corresponding type with $\max_{u,x}|Q_{UX,n}(u,x) - Q_{UX}(u,x)| \le \frac{1}{n}$.
We let $P_{\Xv_{u}}$ be the uniform distribution on the type class 
$T^{n_{u}}\big(Q_{X|U,n}(\cdot|u)\big)$, yielding
\begin{multline}
    \hspace*{-2ex}P_{\Xv_{u}}(\xv_{u}) \\ =\frac{1}{\big|T^{n_{u}}\big(Q_{X|U,n}(\cdot|u)\big)\big|}\openone\Big\{\xv_{u}\in T^{n_{u}}\big(Q_{X|U,n}(\cdot|u)\big)\Big\}. \label{eq:RSC_PX}
\end{multline}
Combining this with \eqref{eq:RSC_PU}, we have by symmetry
that each pair $(\Uv^{(i)},\Xv^{(i,j_1,j_2)})$ is uniformly distributed on $T^n(Q_{UX})$.

The main result of this section is stated in the following theorem,
which makes use of the LM rate defined in \eqref{eq:SU_PrimalLM} and 
the set $\SetTocc$ defined in \eqref{eq:SC_SetT0}.

\begin{thm} \label{thm:RSC_Main}
    For any finite set $\Uc$ and input distribution
    $Q_{UX}$, the rate
    \begin{equation}
        R=R_{0}+\sum_{u}Q_{U}(u)R_{1u} \label{eq:RSC_Main_R}
    \end{equation}
    is achievable provided that $R_{0}$ and $\{R_{1u}\}_{u=1}^{|\Uc|}$ satisfy
    \begin{equation}
        R_{1u} \le \LM\big(Q_{X|U}(\cdot|u)\big), \quad u\in\Uc \label{eq:RSC_R1u}
    \end{equation}
    \begin{multline}
        R_{0}\le\min_{\Ptilde_{UXY}\in\mathcal{T}_{0}(Q_{UX}\times W)}I_{\Ptilde}(U;Y) + \\~~ \bigg[\max_{\Kc\subseteq\Uc,\Kc\ne\emptyset}\sum_{u\in\Kc}Q_{U}(u)\Big(I_{\Ptilde}(X;Y|U=u)-R_{1u}\Big)\bigg]^{+}.\label{eq:RSC_R0}  
    \end{multline}
\end{thm}
\begin{proof}
    As mentioned above, the proof is presented only for $\Uc=\{1,2\}$; the same 
    arguments apply in the general case.
    Observe that the type-1 error event corresponds to the error event for the standard
    constant-composition ensemble with rate $R_{11}$, length $n_{1}=nQ_{U}(1)$,
    input distribution $Q_{X|U}(\cdot|1)$, and ties treated as errors. A similar statement holds for
    the type-2 error probability $\peiibar$, and the analysis
    for these error events is identical to the LM rate derivation
    \cite{Hui,Csiszar1}, yielding \eqref{eq:RSC_R1u}.
    
    The error probability for the type-0 event is given by
    \begin{equation}
        \peobar = \PP\Bigg[\bigcup_{i\ne1}\bigcup_{j_{1},j_{2}}\bigg\{\frac{q^{n}(\Xv^{(i,j_{1},j_{2})},\Yv)}{q^{n}(\Xv,\Yv)}\ge1\bigg\}\Bigg],\label{eq:RSC_RefProof1}
    \end{equation}
    where $(\Yv|\Xv=\xv)\sim W^{n}(\cdot|\xv)$. 
    Writing the probability as an expectation given $(\Uv,\Xv,\Yv)$
    and applying the truncated union bound, we obtain
    \begin{multline}
        \peobar = c_{0}\EE\Bigg[\min\Bigg\{1,(M_{0}-1) \\ \times \EE\Bigg[\PP\bigg[\bigcup_{j_{1},j_{2}}\bigg\{\frac{q^{n}(\Xvbar^{(j_{1},j_{2})},\Yv)}{q^{n}(\Xv,\Yv)}\ge1\bigg\}\,\bigg|\,\Uvbar\bigg]\,\Bigg|\,\Uv,\Xv,\Yv\Bigg]\Bigg\}\Bigg], \label{eq:RSC_RefProof2}                                                                                                                                                                                                                                                                                                                           
    \end{multline}
    where $c_{0} \in [\frac{1}{2},1]$, since for independent events the truncated
    union bound is tight to within a factor of $\frac{1}{2}$ \cite[Lemma A.2]{ShulmanThesis}.  
    We have written the probability of the union over $j_{1}$ and $j_{2}$ as an expectation given $\Uvbar$. 
    
    Let the joint types of $(\Uv,\Xv,\Yv)$ and $(\Uvbar,\Xvbar^{(j_{1},j_{2})},\Yv)$
    be denoted by $P_{UXY}$ and $\Ptilde_{UXY}$ respectively. We claim that 
    \begin{equation}
        \frac{q^{n}(\Xvbar^{(j_{1},j_{2})},\Yv)}{q^{n}(\Xv,\Yv)}\ge1\label{eq:RSC_RefProof3}
    \end{equation}
    can be written as
    \begin{equation}
        \Ptilde_{UXY}\in\SetToncc(P_{UXY})\defeq\SetTocc(P_{UXY})\cap\Pc_{n}(\Uc\times\Xc\times\Yc),
    \end{equation}
    where $\SetTocc$ is defined in \eqref{eq:SC_SetT0}.
    The constraint $\Ptilde_{UX}=P_{UX}$ follows from the construction
    of the random coding ensemble, $\Ptilde_{Y}=P_{Y}$ follows
    since $(\Uv,\Xv,\Yv)$ and $(\Uvbar,\Xvbar^{(j_{1},j_{2})},\Yv)$
    share the same $\Yv$ sequence, and $\EE_{\Ptilde}[\log q(X,Y)]\ge\EE_{P}[\log q(X,Y)]$
    coincides with the condition in \eqref{eq:RSC_RefProof3}. Thus, expanding
    \eqref{eq:RSC_RefProof2} in terms of types, we obtain
    \begin{multline}
        \hspace*{-2ex}\peobar = c_{0} \sum_{P_{UXY}}\PP\Big[\big(\Uv,\Xv,\Yv\big)\in T^{n}(P_{UXY})\Big] \\ \times\min\Bigg\{1,
        (M_{0}-1)\sum_{\Ptilde_{UXY}\in\SetToncc(P_{UXY})}\PP\Big[\big(\Uvbar,\yv\big)\in T^{n}(\Ptilde_{UY})\Big] \\
        \times\PP\bigg[\bigcup_{j_{1},j_{2}}\Big\{\big(\uvbar,\Xvbar^{(j_{1},j_{2})},\yv\big)\in T^{n}(\Ptilde_{UXY})\Big\}\bigg]\Bigg\},\label{eq:RSC_RefProof5} 
    \end{multline}
    where we write $(\uvbar,\yv)$ to denote an arbitrary pair such that 
    $\yv\in T^{n}(P_{Y})$ and $(\uvbar,\yv)\in T^{n}(\Ptilde_{UY})$; note that these sequences \emph{implicitly} depend on $P_{UXY}$ and $\Ptilde_{UXY}$.
    
    Similarly to the discussion following \eqref{eq:RSC_SplitMetric},
    we observe that $\big(\uvbar,\Xvbar^{(j_{1},j_{2})},\yv\big)\in T^{n}(\Ptilde_{UXY})$
    if and only if $\big(\Xvbar_{u}^{(j_{u})},\yv_{u}(\uvbar)\big)\in T^{n_{u}}(\Ptilde_{XY|U}(\cdot,\cdot|u))$
    for $u=1,2$. Thus, applying Lemma \ref{lem:MU_MatchingInd}
    in Appendix \ref{sub:MU_BOUNDS} with $Z_{1}(j_{1})=\Xv_{1}^{(j_{1})}$, 
    $Z_{2}(j_{2})=\Xv_{2}^{(j_{2})}$, $\Ac_{1} = T_{\yv_{1}(\uvbar)}^{n_{1}}(\Ptilde_{XY|U}(\cdot,\cdot|1))$, 
    $\Ac_{2} = T_{\yv_{2}(\uvbar)}^{n_{2}}(\Ptilde_{XY|U}(\cdot,\cdot|2))$, and $\Ac = \big\{ (\Xv_1,\Xv_2) \,:\, \Xv_{u}\in T_{\yv_{u}(\uvbar)}^{n_{u}}(\Ptilde_{XY|U}(\cdot,\cdot|u)), \,u=1,2 \big\}$,
    we obtain
    \begin{align}
        &\PP\bigg[\bigcup_{j_{1},j_{2}}\Big\{\big(\uvbar,\Xvbar^{(j_{1},j_{2})},\yv\big)\in T^{n}(\Ptilde_{UXY})\Big\}\bigg] = \\ & \zeta_{0}^{\prime} \min\bigg\{1, \nonumber \min_{u=1,2}M_{1u}\PP\Big[\big(\Xvbar_{u},\yv_{u}(\uvbar)\big)\in T^{n_{u}}\big(\Ptilde_{XY|U}(\cdot,\cdot|u)\big)\Big], \nonumber \\
        &  M_{11}M_{12}\PP\Big[\bigcap_{u=1,2}\Big\{\big(\Xvbar_{u},\yv_{u}(\uvbar)\big)\in T^{n_{u}}\big(\Ptilde_{XY|U}(\cdot,\cdot|u)\big)\Big\}\Big]\bigg\}, \label{eq:RSC_RefProof6} 
    \end{align}
    where $\zeta_{0}^{\prime}\in[\frac{1}{4},1]$.  This is a minimization of
    four terms corresponding to the four subsets of $\{1,2\}$.
    
    Substituting \eqref{eq:RSC_RefProof6} into \eqref{eq:RSC_RefProof5} and
    applying standard properties of types \cite[Ch.~2]{CsiszarBook}, we obtain 
    \begin{multline}
           \lim_{n\to\infty}-\frac{1}{n}\log\peobar = \min_{P_{UXY}\,:\,P_{UX}=Q_{UX}} \\  \min_{\Ptilde_{UXY}\in \SetTocc(P_{UXY})} D(P_{UXY}\|Q_{UX}\times W) +\bigg[I_{\Ptilde}(U;Y) + \\
           \bigg[\max_{\Kc\subseteq\Uc,\Kc\ne\emptyset}\sum_{u\in\Kc}Q_{U}(u)\Big(I_{\Ptilde}(X;Y|U=u)-R_{1u}\Big)\bigg]^{+}-R_{0}\bigg]^{+}, \label{eq:RSC_Exponent} 
    \end{multline}
    where we have replaced the minimizations over types by minimizations 
    over all distributions in the same way as the proof of Theorem \ref{thm:MAC_Primal}.
    By a similar argument to \cite[Lemma 1]{Csiszar1}, the right-hand side
    of \eqref{eq:RSC_Exponent}, with only the second minimization kept, is
    continuous as a function of $P_{UXY}$ when restricted to distributions
    whose support is the same as that of $Q_{UX} \times W$.  It follows that the right-hand
    side of \eqref{eq:RSC_Exponent} is positive whenever \eqref{eq:RSC_R0}
    holds with strict inequality.
\end{proof}
 
The proof of Theorem \ref{thm:RSC_Main} gives an exponentially tight analysis
yielding the exponent in \eqref{eq:RSC_Exponent}. This does not 
prove that the resulting rate is ensemble-tight, since a subexponential decay of the error probability
to zero is possible in principle.  However, the changes required to prove the
tightness of the rate are minimal.  We saw that each condition in \eqref{eq:RSC_Main_R} corresponds
to an error event with independent constant-composition codewords and a reduced
block length, and hence it follows from existing analyses \cite{Merhav,MacMM}
that $\peibar \to 1$ when $R_{11}$ fails this condition, and analogously for $\peiibar$ and $R_{12}$.
To see that $\peobar \to 1$ when \eqref{eq:RSC_R0} fails, we let $\Ec_i$ be
the event that $q^n(\Xv^{(i,j_1,j_2)},\Yv) \ge q^n(\Xv,\Yv)$ for some $(j_1,j_2)$,
let $I_0(P_{UXY})$ denote the right-hand side of \eqref{eq:RSC_R0} with 
$P_{UXY}$ in place of $Q_1 \times Q_2 \times W$, and write 
\begin{align}
    \peobar &= \PP\Big[ \bigcup_{i\ne1} \Ec_i \Big] \label{eq:RSC_Tightness1} \\
            &= \sum_{P_{UXY}} \PP[(\Uv,\Xv,\Yv) \in T^n(P_{UXY})] \nonumber \\
            & \hspace{8ex}\times\Big(1- \big(1 - \PP[\Ec_2\,|\,P_{UXY}])^{M_0 - 1} \Big) \label{eq:RSC_Tightness2} \\
            &\ge \sum_{P_{UXY}} \PP[(\Uv,\Xv,\Yv) \in T^n(P_{UXY})] \nonumber \\
            &\hspace{8ex}\times \Big(1- \big(1 - p_0(n)e^{-nI_0(P_{UXY})} )^{M_0 - 1} \Big), \label{eq:RSC_Tightness3} 
\end{align}
where \eqref{eq:RSC_Tightness2} follows since the events $\Ec_i$ are conditionally
i.i.d.~given that $(\Uv,\Xv,\Yv)$ has a given joint type $P_{UXY}$, and \eqref{eq:RSC_Tightness3}
holds for some subexponential factor $p_0(n)$ by \eqref{eq:RSC_Exponent}.
Next, we observe from the law of large numbers that the joint type of $(\Uv,\Xv,\Yv)$ 
approaches $Q_1 \times Q_2 \times W$ with high probability as $n\to\infty$.  Moreover, by the same argument
as that of the LM rate \cite[Lemma 1]{Csiszar2}, $I_0(P_{UXY})$ is continuous
in $P_{UXY}$.  Combining these observations, we readily obtain from \eqref{eq:RSC_Tightness3}
that $\peobar \to 1$ if $R_0 > I_0(Q_1 \times Q_2 \times W)$, as desired.

\subsection{Comparison to Standard Superposition Coding}

In this subsection, we show that the conditions in \eqref{eq:RSC_R1u}--\eqref{eq:RSC_R0}
can be weakened to \eqref{eq:SC_R1_CC}--\eqref{eq:SC_Rsum_CC} upon identifying 
\begin{equation}
    R_{1}=\sum_{u}Q_{U}(u)R_{1u}. \label{eq:RSC_Weaken0}
\end{equation}

\begin{prop} \label{prop:SC_Comparison2}
    For any finite auxiliary alphabet $\Uc$ and input distribution $Q_{UX}$, 
    the rate $\max_{R_{0},R_{11},\dotsc,R_{1|\Uc|}}R_{0}+\sum_{u}Q_{U}(u)R_{1u}$ 
    resulting from Theorem \ref{thm:RSC_Main} is at least as high as the rate 
    $\max_{R_{0},R_{1}}R_{0}+R_{1}$ resulting from Theorem \ref{thm:SC_Rate_SC}.
\end{prop}
\begin{proof}
    We begin by weakening \eqref{eq:RSC_R0} to \eqref{eq:SC_Rsum_CC}.  We lower bound the 
    right-hand side of \eqref{eq:RSC_R0} by replacing the maximum over $\Kc$ by the  
    particular choice $\Kc=\Uc$, yielding
    \begin{equation}
        R_{0}\le\min_{\Ptilde_{UXY}\in\SetTocc(Q_{UX}\times W)}I_{\Ptilde}(U;Y) + \Big[I_{\Ptilde}(X;Y|U)-R_{1}\Big]^{+}, \label{eq:RSC_Weaken1}
    \end{equation}
    where we have used \eqref{eq:RSC_Weaken0} and the definition of conditional mutual information.  
    We can weaken \eqref{eq:RSC_Weaken1} to \eqref{eq:SC_Rsum_CC} using the chain rule for mutual
    information, and noting that \eqref{eq:RSC_Weaken1} is always satisfied when the minimizing
    $\Ptilde_{UXY}$ satisfies $I_{\Ptilde}(U;Y) > R_{0}$.
    
    Next, we show that highest value of $R_{1}$ permitted by the $|\Uc|$ conditions in \eqref{eq:RSC_R1u}, 
    denoted by $R_{1}^{*}$, can be lower bounded by the right-hand side of \eqref{eq:SC_R1_CC}. 
    From \eqref{eq:RSC_Weaken0} and \eqref{eq:RSC_R1u}, we have
    \begin{equation}
        R_{1}^{*} = \sum_{u}Q_{U}(u)I_{\Ptilde^{*}}(X;Y|U=u), \label{eq:RSC_Weaken2}
    \end{equation}
    where $\Ptilde_{XY|U}^{*}(\cdot,\cdot|u)$ is the distribution that achieves the minimum in \eqref{eq:SU_PrimalLM} 
    under $Q_{X|U}(\cdot|u)$.  Defining the joint distribution $\Ptilde_{UXY}^{*}$ accordingly with $\Ptilde_{U}^{*}=Q_{U}$,
    we can write \eqref{eq:RSC_Weaken2} as
    \begin{equation}
        R_{1}^{*} = I_{\Ptilde^{*}}(X;Y|U).
    \end{equation}
    Therefore, we can lower bound $R_{1}^{*}$ by the right-hand side of \eqref{eq:SC_R1_CC} provided that
    $\Ptilde_{UXY}^{*} \in \SetTicc(Q_{UX}\times W)$.  The constraints $\Ptilde_{UX}^{*}=Q_{UX}$ and
    $\Ptilde_{UY}^{*}=P_{UY}$ in \eqref{eq:SC_SetT1} are satisfied since we have chosen 
    $\Ptilde_{U}^{*}=Q_{U}$, and since the constraints in \eqref{eq:SU_PrimalLM} imply 
    $\Ptilde_{X|U}^{*}(\cdot|u)=Q_{X|U}(\cdot|u)$ and $\Ptilde_{Y|U}^{*}(\cdot|u)=P_{Y|U}(\cdot|u)$ for all $u\in\Uc$.  The constraint 
    $\EE_{\Ptilde^{*}}[\log q(X,Y)]\ge\EE_{P}[\log q(X,Y)]$ is satisfied since, from 
    \eqref{eq:SU_PrimalLM}, we have $\EE_{\Ptilde^{*}}[\log q(X,Y)\,|\,U=u]\ge\EE_{P}[\log q(X,Y)\,|\,U=u]$ for all $u\in\Uc$.
\end{proof} 

Intuitively, one can think of the gain of the refined superposition coding
ensemble as being due to a stronger
dependence among the codewords.  For standard SC, the codewords $\{\Xv^{(i,j)}\}_{j=1}^{M_{1}}$ are
conditionally independent given $\Uv^{(i)}$, whereas for refined superposition coding
this is generally not the case.
The additional structure leads to further constraints in the minimizations, and 
maxima over more terms in the objective functions, both leading to higher overall rates.

It should be noted, however, that the exponents for  standard superposition coding
may be higher, particularly at 
low to moderate rates.  In particular, we noted in the proof of Theorem \ref{thm:RSC_Main} that
the type-1 and type-2 error events are equivalent to a single-user channel, but the corresponding
block lengths are only $n_{1}$ and $n_{2}$.  Thus, if either $Q_{U}(1)$ or $Q_{U}(2)$ is close to zero,
the corresponding exponent is small.

Finally, we recall that the standard superposition coding rate is at least as high as Lapidoth's expurgated parallel coding rate \cite{MMSomekh}, though no
example of strict improvement is known.

\subsection{Dual Expressions and General Alphabets}

In this subsection, we present a dual expression for the rate given in 
Theorem \ref{thm:RSC_Main} in the case that $|\Uc|=2$, as well as extending
the result to general alphabets $\Xc$ and $\Yc$.

With $\Uc=\{1,2\}$, the condition in \eqref{eq:RSC_R0} is given by
\begin{multline}
    R_{0}\le\min_{\Ptilde_{UXY}\in\SetTocc(Q_{UX}\times W)}I_{\Ptilde}(U;Y) \\ + \Big[\max\Big\{Q_{U}(1)\big(I_{\Ptilde}(X;Y|U=1)-R_{11}\big),
    \\ Q_{U}(2)\big(I_{\Ptilde}(X;Y|U=2)-R_{12}\big), I_{\Ptilde}(X;Y|U)-R_{1}\Big\}\Big]^{+},\label{eq:RSC_AltR0}
\end{multline}
where 
\begin{equation}
    R_{1}\defeq\sum_{u}Q_{U}(u)R_{1u}. \label{eq:RSC_R1Expr}
\end{equation}

Since the right-hand side of \eqref{eq:RSC_R1u} is the LM rate, we can  
use the dual expression in \eqref{eq:SU_DualLM}.  The main result of this
subsection gives a dual expression for \eqref{eq:RSC_AltR0}, and extends its
validity to memoryless MACs with infinite or continuous alphabets.

We again use cost-constrained random coding.
We consider the ensemble given in \eqref{eq:RSC_Distr}, with $P_{\Xv_{u}}$ given by
\begin{equation}
    P_{\Xv_{u}}(\xv_{u}) = \frac{1}{\mu_{u,n_{u}}}\prod_{i=1}^{n_{u}}Q_{X|U}(x_{u,i}|u_i)\openone\big\{\xv_{u}\in \Dc_{u,n_{u}}\big\}, \label{eq:RSC_DistrCost}
\end{equation}
where
\begin{multline}
    \Dc_{u,n_{u}} \defeq\Bigg\{ \xv_{u}:\left|\frac{1}{n_{u}}\sum_{i=1}^{n_{u}}a_{u,l}(x_{u,i})-\phi_{u,l}\right|\le\frac{\delta}{n_{u}}, \\ l=1,\dotsc,L_{u}\Bigg\} \label{eq:RSC_SetD}
\end{multline}
\begin{equation}
    \phi_{u,l}\defeq\EE_{Q_u}\big[a_{u,l}(X_u)\,|\,U=u\big], \label{eq:RSC_Phi}
\end{equation}
and where $\mu_{u,n_{u}}$, $\{a_{u,l}\}$ and $\delta$ are defined analogously to \eqref{eq:CNT_Domain},
and $n_{u}$ is defined in \eqref{eq:RSC_BC_n1}. 

\begin{thm} \label{thm:RSC_Dual}
    The condition in \eqref{eq:RSC_AltR0} holds if and only if the following holds for at least one of $u=1,2$:
    \begin{align}
        & R_{0} \le \sup_{s\ge0,\rho_{1}\in[0,1],\rho_{2}\in[0,1],a(\cdot,\cdot)} - \sum_{u'=1,2}\rho_u(u')Q_{U}(u')R_{1u'} \nonumber \\
        &+\EE\left[\log\frac{\big(q(X,Y)^{s_u(U)}e^{a(U,X)}\big)^{\rho_u(U)}}{\EE\Big[\Big(\EE\big[q(\Xbar,Y)^{s_u(\Ubar)}e^{a(\Ubar,\Xbar)}\,\big|\,\Ubar\big]\Big)^{\rho_u(\Ubar)}\,\Big|\,Y\Big]}\right]  \label{eq:RSC_DualR0_1}
    \end{align}
    where 
    \begin{gather}
        \rho_1(1)=\rho_{1},~ \rho_1(2)=\rho_{1}\rho_{2},~ s_1(1)=\rho_{2}s,~ s_1(2)=s \label{eq:RSC_DualVars1} \\
        \rho_2(1)=\rho_{1}\rho_{2},~ \rho_2(2)=\rho_{2},~ s_2(1)=s,~ s_2(2)=\rho_{1}s \label{eq:RSC_DualVars2}
    \end{gather} 
    and $(U,X,Y,\Ubar,\Xbar)\sim Q_{UX}(u,x)W(y|x)Q_{UX}(\ubar,\xbar)$.
    
    Moreover, for any mismatched memoryless channel (possibly having infinite 
    or continuous alphabets) and input 
    distribution $Q_{UX}$ $(\Uc={1,2})$, the rate 
    $R=R_{0}+\sum_{u=1,2}Q_{U}(u)R_{1u}$ is achievable for any triplet 
    $(R_{0},R_{11},R_{12})$ satisfying \eqref{eq:RSC_R1u} (with $\LM$ 
    defined in \eqref{eq:SU_DualLM}) and 
    \eqref{eq:RSC_DualR0_1} for at least one of $u=1,2$.  The supremum in
    \eqref{eq:SU_DualLM} is subject to $\EE_{Q}[a(X)^2]<\infty$, and 
    that in \eqref{eq:RSC_DualR0_1}
    is subject to $\EE_Q[a(U,X)^2]<\infty$.  Furthermore, the rate is achievable
    using cost-constrained coding in \eqref{eq:RSC_DistrCost} with $L_{1}=L_{2}=2$.
\end{thm}
\begin{proof}
    Both the proof of the primal-dual equivalence is 
    and the direct derivation of \eqref{eq:RSC_DualR0_1}
    are given in Appendix \ref{sub:RSC_PROOFS}.  The choice $L_{1}=L_{2}=2$ suffices since for $u=1,2$, one cost is required for \eqref{eq:RSC_R1u}
    and another for \eqref{eq:RSC_DualR0_1}.  
    It suffices to let the cost functions for 
    \eqref{eq:RSC_DualR0_1} with $u=1$ and $u=2$ coincide, since the 
    theorem only requires that one of the two hold.
\end{proof}

The condition in \eqref{eq:RSC_DualR0_1} bears a strong 
resemblance to the standard superposition coding
condition in \eqref{eq:SC_Rsum_Dual}; the latter
can be recovered by setting $\rho_{2}=1$ in the condition with $u=1$, or 
or $\rho_{1}=1$ in the condition with $u=2$.

\section{Numerical Examples} \label{sec:MU_COMPARISONS}

\subsection{Error Exponent for the Multiple-Access Channel} \label{sub:MAC_NUMERICAL}

We revisit the parallel BSC example given by Lapidoth \cite{MacMM},
consisting of binary inputs $\Xc_1=\Xc_2=\{0,1\}$ and a pair of binary
outputs $\Yc = \{0,1\}^2$.  The output is given by $Y = (Y_1,Y_2)$, where
for $\nu=1,2$, $Y_\nu$ is generated by passing $X_\nu$ through a 
binary symmetric channel (BSC) with some crossover probability $\delta_\nu < 0.5$.
The mismatched decoder assumes that both crossover probabilities are 
equal to $\delta<0.5$. 
The decoder assumes that both crossover
probabilities are equal. The corresponding
decoding rule is equivalent to minimizing sum of $t_{1}$ and $t_{2}$,
where $t_{\nu}$ is the number of bit flips from the input sequence
$\xv_{\nu}$ to the output sequence $\yv_{\nu}$. As noted in \cite{MacMM}, 
this decision rule is in fact equivalent to ML. 

We let both $Q_{1}$ and $Q_{2}$ be equiprobable on $\{0,1\}$.
With this choice, it was shown in \cite{MacMM} that the right-hand
side of \eqref{eq:MAC_R12'} is no greater than 
\begin{equation}
    2\left(1-H_{2}\left(\frac{\delta_{1}+\delta_{2}}{2}\right)\right)\,\,\mathrm{bits/use}, \label{eq:MAC_Rworse2}
\end{equation}
where $H_{2}(\cdot)$ is the binary entropy function in bits.
In fact, this is the same rate that would be obtained by considering
the corresponding \emph{single-user} channel with $X=(X_1,X_2)$, and
applying the LM rate with a uniform distribution on the quaternary
input alphabet \cite{MacMM}.

On the other hand, the refined condition in \eqref{eq:MAC_R12_LM}
can be used to prove the achievability of \emph{any} $(R_{1},R_{2})$
within the rectangle with corners $(0,0)$ and $(C_{1},C_{2})$, where
$C_{\nu}\defeq1-H_{2}(\delta_{\nu})$ \cite{MacMM}.  This implies that
the mismatched capacity region coincides with the (matched) capacity region.

We evaluate the error exponents using the optimization software YALMIP
\cite{YALMIP}. Figure \ref{fig:MAC_Exponents} plots each of the
exponents as a function of $\alpha$, where the rate pair is
$(R_{1},R_{2})=(\alpha C_{1},\alpha C_{2})$. While the overall
error exponent $\Ercc(\Qv,R_{1},R_{2})$ in \eqref{eq:MAC_Er_LM}
is unchanged at low to moderate values of $\alpha$ when $\Erccprime$ in \eqref{eq:MAC_Er12'}
is used in place of $\Eiiircc$, this is not true for high values
of $\alpha$. Furthermore, consistent with the preceding discussion,
$\Erccprime$ is non-zero only for $\alpha<0.865$,
whereas $\Eiiircc$ is positive for all $\alpha<1$.
The fact that $\Eiiircc$
and $\Erccprime$ coincide at low values of $\alpha$
is consistent with \cite[Cor. 5]{MACExponent5}, which states
that $\Erccprime$ is ensemble-tight at low rates.

\begin{figure}
    \begin{centering}
        \includegraphics[width=0.4\paperwidth]{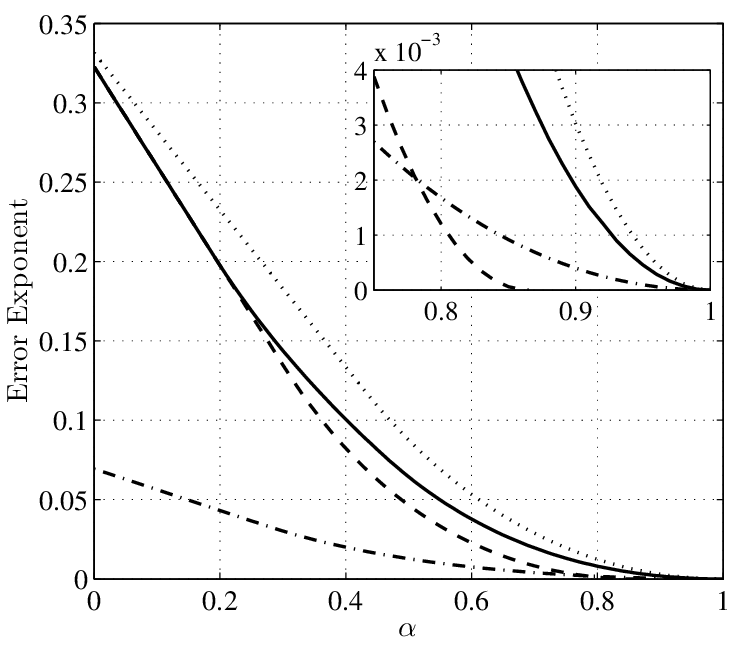}
        \par
    \end{centering}
    
    \caption{Error exponents $\Eircc$ (dotted), $\Eiircc$
        (dash-dot), $\Eiiircc$ (solid) and $\Erccprime$
        (dashed) for the parallel channel 
        using $\delta_{1}=0.05$, $\delta_{2}=0.25$ and equiprobable input
        distributions on $\{0,1\}$. The rate pair is given by $(R_1,R_2)=(\alpha C_1,\alpha C_2).$ \label{fig:MAC_Exponents}}
\end{figure}

\subsection{Achievable Rates for Single-User Channels} \label{sec:MU_NUMERICAL}

In this subsection, we provide examples comparing the two versions of superposition coding and the LM rate.  
We do not explicitly give values for Lapidoth's rate \cite{MacMM}, 
since for each example given, we found it to coincide with the 
superposition coding rate (see Theorem \ref{thm:SC_Rate_SC}).

\subsubsection{Sum Channel}

We first consider a sum-channel analog of the parallel-channel example given in Section \ref{sub:MAC_NUMERICAL}.
Given two channels $(W_{1},W_{2})$ respectively defined on the alphabets
$(\Xc_{1},\Yc_{1})$ and $(\Xc_{2},\Yc_{2})$,
the sum channel is defined to be the channel $W(y|x)$ with
$|\Xc|=|\Xc_{1}|+|\Xc_{2}|$ and $|\Yc|=|\Yc_{1}|+|\Yc_{2}|$
such that one of the two subchannels is used on each transmission
\cite{Shannon}. One can similarly combine two metrics $q_{1}(x_{1},y_{1})$
and $q_{2}(x_{2},y_{2})$ to form a sum metric $q(x,y)$.  Assuming without loss
of generality that $\Xc_{1}$ and $\Xc_{2}$ are disjoint and $\Yc_{1}$ 
and $\Yc_{2}$ are disjoint, we have
\begin{equation}
    q(x,y) =
        \begin{cases}
            q_1(x_{1},y_{1}) & x_{1}\in\Xc_{1}\text{ and }y_{1}\in\Yc_{1} \\
            q_2(x_{2},y_{2}) & x_{2}\in\Xc_{2}\text{ and }y_{2}\in\Yc_{2} \\
            0 & \text{otherwise},
        \end{cases}
\end{equation}
and similarly for $W(y|x)$. Let $\hat{Q}_{1}$ and $\hat{Q}_{2}$ be
the distributions that maximize the LM rate in \eqref{eq:SU_PrimalLM} 
on the respective subchannels. We set $\Uc=\{1,2\}$, $Q_{X|U}(\cdot|1)=(\hat{Q}_{1},\boldsymbol{0})$
and $Q_{X|U}(\cdot|2)=(\boldsymbol{0},\hat{Q}_{2})$, where $\boldsymbol{0}$
denotes the zero vector. We leave $Q_{U}$ to be specified.

Combining the constraints $\Ptilde_{UX}=Q_{UX}$ and $\EE_{\Ptilde}[\log q(X,Y)]\ge\EE_{P}[\log q(X,Y)]$
in \eqref{eq:SC_SetT0}, we find that the minimizing
$\Ptilde_{UXY}$ in \eqref{eq:RSC_R0} only has
non-zero values for $(u,x,y)$ such that (i) $u=1$, $x\in\Xc_{1}$
and $y\in\Yc_{1}$, or (ii) $u=2$, $x\in\Xc_{2}$
and $y\in\Yc_{2}$.  It follows that $U$ is a deterministic
function of $Y$ under the minimizing $\Ptilde_{UXY}$, and hence
\begin{equation}
    I_{\Ptilde}(U;Y)=H(Q_{U})-H_{\Ptilde}(U|Y)=H(Q_{U}).
\end{equation}
Therefore, the right-hand side of \eqref{eq:RSC_R0} is lower bounded
by $H(Q_{U})$. Using \eqref{eq:RSC_Main_R}, it follows that we can
achieve the rate
\begin{multline}
    H(Q_U) + Q_U(1)I_{1}^{\mathrm{LM}}(\hat{Q}_{1}) + Q_U(2)I_{2}^{\mathrm{LM}}(\hat{Q}_{2}) \\ = \log\big(e^{I_{1}^{\mathrm{LM}}(\hat{Q}_{1})}+e^{I_{2}^{\mathrm{LM}}(\hat{Q}_{2})}\big) \label{eq:SC_SumChRate}
\end{multline}
where $I_{\nu}^{\mathrm{LM}}$ is the LM rate for subchannel $\nu$, 
and the equality follows by optimizing $Q_U$ in the same
way as \cite[Sec. 16]{Shannon}, yielding $Q_U(1) = \frac{e^{I_{1}^{\mathrm{LM}}(\hat{Q}_{1})}}{e^{I_{1}^{\mathrm{LM}}(\hat{Q}_{1})}+e^{I_{2}^{\mathrm{LM}}(\hat{Q}_{2})}}$.  
Using similar arguments to \cite{MacMM}, it can be shown that the 
LM rate with an optimized input distribution can be strictly 
less than \eqref{eq:SC_SumChRate} even for simple examples (e.g.~binary symmetric subchannels).

\subsubsection{Zero Undetected Error Capacity}

It was shown by Csisz\'{a}r and Narayan \cite{Csiszar2} that two special cases of the
mismatched capacity are the zero-undetected erasures capacity \cite{ZUEC} and the 
zero-error capacity \cite{ShannonZero}.  Here we consider the zero-undetected erasures capacity, defined to 
be the highest achievable rate in the case that the decoder is required to know with
certainty whether or not an error has occurred.  For any DMC, the zero-undetected erasures capacity is equal to
the mismatched capacity under the decoding metric $q(x,y)=\openone\{W(y|x)>0\}$ \cite{Csiszar2}. 

We consider an example from \cite{ZUEC2}, where $\Xc=\Yc=\{0,1,2\}$, and
the channel is described by the entries of
\begin{align} 
    \Wv & = \left[\begin{array}{ccc}
        0.75 & 0.25 & 0 \\
        0 & 0.75 & 0.25 \\
        0.25 & 0 & 0.75 \\
    \end{array}\right]
\end{align}
where $x$ indexes the rows and $y$ indexes the columns.

Using an exhaustive search to three decimal places, we found the optimized LM rate to be
$R_{\mathrm{LM}}^{*}=0.599$ bits/use, using the input distribution
$Q=(0.449,0.551,0)$.  It was stated
in \cite{ZUEC2} that the rate obtained by considering the second-order
product of the channel and metric (see \cite{Csiszar2})
is equal to $R_{\mathrm{LM2}}^{*}=0.616$ bits/use.  Using local optimization
techniques, we verified that this rate is achieved with $Q=(0,0.250,0,0.319,0,0,0,0.181,0.250)$,  
where the order of the inputs is $(0,0),(0,1),(0,2),(1,0),\dotsc,(2,2)$.  

The global optimization of \eqref{eq:SC_R1_CC}--\eqref{eq:SC_Rsum_CC}
over $\Uc$ and $Q_{UX}$ appears to be difficult. Setting $|\Uc|=2$
and applying local optimization techniques using a number of starting
points, we obtained an achievable rate of $R_{\mathrm{sc}}^{*}=0.695$ bits/use,
with $Q_{U}=(0.645,0.355)$, $Q_{X|U}(\cdot|1)=(0.3,0.7,0)$ and $Q_{X|U}(\cdot|2)=(0,0,1)$.
Thus, superposition coding not only yields an improvement over the single-letter 
LM rate, but also over the two-letter version.  Note that since the decoding
metric is the erasures-only metric, applying the LM rate
to the $k$-th order product channel achieves the mismatched capacity
in the limit as $k\to\infty$ \cite{Csiszar2}; however, in this example, 
a significant gap remains for $k = 2$.

\subsubsection{A Case where Refined Superposition Coding Outperforms Standard Superposition Coding}

Here we consider the channel and decoding metric described by the entries
of
\begin{align}
    \Wv & = \left[\begin{array}{cccc}
        0.99 & 0.01 & 0 & 0\\
        0.01 & 0.99 & 0 & 0\\
        0.1~ & 0.1~ & 0.7 & \,0.1~\\
        0.1 & 0.1 & 0.1 & \,0.7~
    \end{array}\right] \label{eq:Example_W} \\
    \qv &= \left[\begin{array}{cccc}
        1 & 0.5 & 0 & 0\\
        0.5 & 1 & 0 & 0\\
        0.05 & 0.15 & 1 & 0.05\\
        0.15 & 0.05 & 0.5 & 1
    \end{array}\right].\label{eq:Example_q}
\end{align}
We have intentionally chosen a highly asymmetric channel and metric,
since such examples often yield larger gaps between the various
achievable rates. Using an exhaustive search to three decimal places, we found the optimized 
LM rate to be $R_{\mathrm{LM}}^{*}=1.111$ bits/use,
which is achieved by the input distribution $Q_{X}^{*}=(0.403,0.418,0,0.179)$.

Setting $|\Uc|=2$ and applying local optimization techniques using a number of starting
points, we obtained an achievable rate of $R_{\mathrm{rsc}}^{*}=1.313$ bits/use,
with $Q_{U}=(0.698,0.302)$, $Q_{X|U}(\cdot|1)=(0.5,0.5,0,0)$ and $Q_{X|U}(\cdot|u)=(0,0,0.528,0.472)$.
We denote the corresponding input distribution by $Q_{UX}^{(1)}$.

Applying similar techniques to the standard superposition coding rate, we obtained an achievable
rate of $R_{\mathrm{sc}}^{*}=1.236$ bits/use, with $Q_{U}=(0.830,0.170)$,
$Q_{X|U}(\cdot|1)=(0.435,0.450,0.115,0)$ and $Q_{X|U}(\cdot|2)=(0,0,0,1)$. We
denote the corresponding input distribution by $Q_{UX}^{(2)}$.

The achievable rates for this example are summarized in Table \ref{tab:Rates},
where $Q_{UX}^{(\mathrm{LM})}$ denotes the distribution in which
$U$ is deterministic and the $X$-marginal maximizes the LM rate.
While the achievable rate of Theorem \ref{thm:RSC_Main} coincides with
that of Theorem \ref{thm:SC_Rate_SC} under $Q_{UX}^{(2)}$,
the former is significantly higher under $Q_{UX}^{(1)}$. Both types
of superposition coding yield a strict improvement over the LM rate.

Our parameters may not be globally optimal, and thus we cannot conclude
from this example that refined superposition coding yields a strict improvement
over standard superposition coding (and hence over Lapidoth's rate \cite{MacMM}) after optimizing 
$\Uc$ and $Q_{UX}$. However, 
improvements for a fixed set of random-coding parameters are still of 
interest due to the fact that global optimizations are prohibitively complex in general.

\begin{center}
    \begin{table}[t]   
        \begin{center}  
            \caption{Achievable rates (bits) for the mismatched channel  \eqref{eq:Example_W}--\eqref{eq:Example_q}.}   
            \label{tab:Rates}   
            \begin{tabular}{@{}ccc@{}}     
                \toprule  Input Distribution & Refined SC & Standard SC \\  
                \midrule         
                $Q_{UX}^{(1)}$               & 1.313      & 1.060       
                \vspace{0.5mm}\\     
                $Q_{UX}^{(2)}$               & 1.236      & 1.236       
                \vspace{0.5mm}\\     
                $Q_{UX}^{(\mathrm{LM})}$     & 1.111      & 1.111 \\
                \bottomrule   
            \end{tabular}     
        \end{center} 
    \end{table}
    \par
\end{center}

\section{Conclusion} \label{sec:MU_CONCLUSION}

We have provided techniques for studying multiuser random-coding ensembles
for channel coding problems with mismatched decoding.  The key initial step
in each case is the application of a refined bound on the probability
of a multiply-indexed union (\emph{cf.}~Appendix \ref{sub:MU_BOUNDS}),
from which one can apply constant-composition coding and the method 
of types to obtain primal expressions and prove ensemble tightness, or
cost-constrained random coding to obtain dual expressions and 
continuous-alphabet generalizations.  We have demonstrated our techniques
on both the mismatched MAC and the single-user channel with refined
superposition coding, with the latter providing a new achievable rate 
at least as good as all previous rates in the literature.

After the initial preparation of this work, the superposition
coding rate from Theorems \ref{thm:SC_Rate_SC}--\ref{thm:SC_Rate_Dual}
was used to find an example for which the LM rate is strictly smaller than the
mismatched capacity for a binary-input DMC \cite{PaperBI}, thus
providing a counter-example to the converse reported in \cite{Balikirsky}.
Another work building on this paper is \cite{TanRelay},
which considers the matched relay channel, and shows that 
the utility of our refined union bounds is not restricted to 
mismatched decoders.  

\appendices
\numberwithin{equation}{section}
\setcounter {equation}{0}

\section{Upper and Lower Bounds on the Probability of a Multiply-Indexed Union} \label{sub:MU_BOUNDS}

Bounds on the random-coding error probability in channel
coding problems are often obtained using 
the truncated union bound, which states that for any set of
events $\{A_{i}\}_{i=1}^{N}$, 
\begin{equation}
    \PP\Big[\bigcup_{i}A_{i}\Big] \le \min\Big\{1,\sum_{i}\PP[A_{i}]\Big\}. \label{eq:MU_TruncatedBound}
\end{equation}
In this paper, we are also interested in lower bounds on the probability of a union,
which are used to prove ensemble tightness results.  In particular, we make use
of de Caen's lower bound \cite{Caen}, which states that
\begin{equation}
    \PP\Big[\bigcup_{i}A_{i}\Big] \ge \sum_{i}\frac{\PP[A_{i}]^{2}}{\sum_{i^{\prime}}\PP[A_{i} \cap A_{i^{\prime}}]}. \label{eq:MU_deCaenBound}
\end{equation}
In the case that the events are pairwise independent and identically distributed,
\eqref{eq:MU_deCaenBound} proves the tightness of  
\eqref{eq:MU_TruncatedBound} to within a factor of $\frac{1}{2}$; 
see the proof of \cite[Thm.~1]{PaperSU}.  

In this section, we provide a number of upper and lower bounds on the probability
of a multiply-indexed union.  In several cases of interest, the upper and lower bounds
coincide to within a constant factor, and generalize the above-mentioned tightness result
of \cite{PaperSU} to certain settings where pairwise independence need not hold.

\begin{lem} \label{lem:MU_UpperInd}
    Let $\{Z_{1}(i)\}_{i=1}^{N_{1}}$ and $\{Z_{2}(j)\}_{j=1}^{N_{2}}$ be independent
    sequences of identically distributed random variables on the alphabets $\Zc_{1}$ and $\Zc_{2}$ 
    respectively, with $Z_1(i) \sim P_{Z_1}$ and $Z_2(j) \sim P_{Z_2}$.  For any set
    $\Ac \subseteq \Zc_{1}\times\Zc_{2}$, we have:
    
    1) A general upper bound is given by
    \begin{align}
        & \PP\bigg[\bigcup_{i,j}\Big\{\big(Z_{1}(i),Z_{2}(j)\big)\in\Ac\Big\}\bigg] \le\min\Bigg\{1, \nonumber \\
        & \quad N_{1}\EE\bigg[\min\Big\{1,N_{2}\PP\big[(Z_{1},Z_{2})\in\Ac\,\big|\, Z_{1}\big]\Big\}\bigg],\nonumber \\
        & \quad N_{2}\EE\bigg[\min\Big\{1,N_{1}\PP\big[(Z_{1},Z_{2})\in\Ac\,\big|\, Z_{2}\big]\Big\}\bigg]\Bigg\}\label{eq:MU_GeneralUB}                     
    \end{align}
    where $(Z_{1},Z_{2})\sim P_{Z_{1}}\times P_{Z_{2}}$.
    
    2) If $\{Z_{1}(i)\}_{i=1}^{N_{1}}$ and $\{Z_{2}(j)\}_{j=1}^{N_{2}}$ are
    pairwise independent, then we have the lower bound
    \begin{align}
        &\PP\bigg[\bigcup_{i,j}\Big\{\big(Z_{1}(i),Z_{2}(j)\big)\in\Ac\Big\}\bigg]\ge\frac{1}{4}\min\Bigg\{1, \nonumber \\
        & \qquad N_{1}\frac{\PP\big[(Z_{1},Z_{2})\in\Ac\big]^{2}}{\PP\big[(Z_{1},Z_{2})\in\Ac\,\cap\,(Z_{1},Z_{2}^{\prime})\in\Ac\big]},\\
        &\qquad N_{2}\frac{\PP\big[(Z_{1},Z_{2})\in\Ac\big]^{2}}{\PP\big[(Z_{1},Z_{2})\in\Ac\,\cap\,(Z_{1}^{\prime},Z_{2})\in\Ac\big]},  \nonumber \\ 
        & \qquad N_{1}N_{2}\PP\big[(Z_{1},Z_{2})\in\Ac\big]\Bigg\}, \label{eq:MU_GeneralLB}
    \end{align}
    where $(Z_{1},Z_{1}^{\prime},Z_{2},Z_{2}^{\prime}) \sim P_{Z_{1}}(z_{1})P_{Z_{1}}(z_{1}^{\prime})P_{Z_{2}}(z_{2}) P_{Z_{2}}(z_{2}^{\prime})$.
\end{lem}
\begin{proof}
    We first prove \eqref{eq:MU_GeneralUB}.  Applying the union bound to the union over $i$ gives
    \begin{align}
        & \PP\bigg[\bigcup_{i,j}\Big\{\big(Z_{1}(i),Z_{2}(j)\big)\in\Ac\Big\}\bigg] \nonumber \\ 
        & \quad\le N_{1}\PP\bigg[\bigcup_{j}\Big\{\big(Z_{1},Z_{2}(j)\big)\in\Ac\Big\}\bigg]\label{eq:GenProof1} \\
        &\quad=N_{1}\EE\bigg[\PP\bigg[\bigcup_{j}\Big\{\big(Z_{1},Z_{2}(j)\big)\in\Ac\Big\}\,\bigg|\, Z_{1}\bigg]\bigg]. 
    \end{align}
    Applying the truncated union bound to the union over $j$, we recover 
    the second term in the outer minimization in \eqref{eq:MU_GeneralUB}.
    The third term is obtained similarly by applying the union bounds
    in the opposite order, and the first term is trivial.
    
    To prove \eqref{eq:MU_GeneralLB}, we make use of de Caen's bound in \eqref{eq:MU_deCaenBound}. Noting by symmetry that each 
    term in the outer summation is equal, and splitting the inner summation according to which
    of the $(i,j)$ indices coincide with $(i^{\prime},j^{\prime})$, we obtain  
    \begin{align}
        &\PP\bigg[\bigcup_{i,j}\Big\{\big(Z_{1}(i),Z_{2}(j)\big)\in\Ac\Big\}\bigg] \nonumber \\ 
        &~~\ge N_{1}N_{2}\PP\big[(Z_{1},Z_{2})\in\Ac\big]^{2} \nonumber\\ &\qquad\times\bigg((N_{1}-1)(N_{2}-1)\PP\big[(Z_{1},Z_{2})\in\Ac\big]^2 \nonumber\\
           &\qquad~~~+ (N_{2}-1)\PP\big[(Z_{1},Z_{2})\in\Ac\,\cap\,(Z_{1},Z_{2}^{\prime})\in\Ac\big]
           \nonumber\\ &\qquad~~~ +(N_{1}-1)\PP\big[(Z_{1},Z_{2})\in\Ac\,\cap\,(Z_{1}^{\prime},Z_{2})\in\Ac\big] \nonumber \\
           &\hspace{24ex}+ \PP\big[(Z_{1},Z_{2})\in\Ac\big]
           \bigg)^{-1}.
    \end{align}
    The lemma follows by upper bounding $N_{\nu} - 1$ by $N_{\nu}$ for $\nu=1,2$, and upper bounding the four terms in the $(\cdot)^{-1}$ by four times the maximum of those terms.
\end{proof}

The following lemma gives conditions under which a weakened version of \eqref{eq:MU_GeneralUB}
matches \eqref{eq:MU_GeneralLB} to within a factor of four.  Recall that $\nu^c$ denotes the 
item in $\{1,2\}$ differing from $\nu$

\begin{lem} \label{lem:MU_MatchingInd}
    Let $\{Z_{1}(i)\}_{i=1}^{N_{1}}$ and $\{Z_{2}(j)\}_{j=1}^{N_{2}}$ be independent
    sequences of identically distributed random variables on the alphabets $\Zc_{1}$ and $\Zc_{2}$ 
    respectively, with $Z_1(i) \sim P_{Z_1}$ and $Z_2(j) \sim P_{Z_2}$.   Fix a set 
    $\Ac \subseteq \Zc_{1}\times\Zc_{2}$, and define 
    \begin{equation}
        \Ac_{\nu} \defeq \Big\{ z_{\nu}\in\Zc_{\nu}\,:\,(z_{1},z_{2})\in\Ac\mathrm{~for~some~}z_{\nu^c}\in\Zc_{\nu^c}\Big\} \label{eq:MU_SetA1} 
    \end{equation}
    for $\nu=1,2$.
    \begin{enumerate}
    \item A general upper bound is given by 
    \begin{align}
        &\hspace*{-2ex}\PP\bigg[\bigcup_{i,j}\Big\{\big(Z_{1}(i),Z_{2}(j)\big)\in\Ac\Big\}\bigg] \le \min\Big\{1,  \nonumber \\ & \hspace*{-2ex}N_{1}\PP\big[Z_{1}\in\Ac_{1}\big],N_{2}\PP\big[Z_{2}\in\Ac_{2}\big],N_{1}N_{2}\PP\big[(Z_{1},Z_{2})\in\Ac\big]\Big\}, \label{eq:MU_MatchingUB}
    \end{align}
    where $(Z_{1},Z_{2})\sim P_{Z_{1}}\times P_{Z_{2}}$.
    \item If (i) $\{Z_{1}(i)\}_{i=1}^{N_{1}}$ are pairwise independent, (ii) $\{Z_{2}(j)\}_{j=1}^{N_{2}}$
    are pairwise independent, (iii) $\PP\big[(z_{1},Z_{2})\in\Ac\big]$
    is the same for all $z_{1}\in\Ac_{1}$, and (iv) $\PP\big[(Z_{1},z_{2})\in\Ac\big]$
    is the same for all $z_{2}\in\Ac_{2}$, then
    \begin{align}
        &\hspace*{-2ex}\PP\bigg[\bigcup_{i,j}\Big\{\big(Z_{1}(i),Z_{2}(j)\big)\in\Ac\Big\}\bigg] \ge \frac{1}{4}\min\Big\{1, \nonumber \\ & \hspace*{-2ex}N_{1}\PP\big[Z_{1}\in\Ac_{1}\big],N_{2}\PP\big[Z_{2}\in\Ac_{2}\big],N_{1}N_{2}\PP\big[(Z_{1},Z_{2})\in\Ac\big]\Big\}. \label{eq:MU_MatchingLB}
    \end{align}
    \end{enumerate}
\end{lem}
\begin{proof}
    We obtain \eqref{eq:MU_MatchingUB} by weakening \eqref{eq:MU_GeneralUB} in multiple ways.
    The second term in \eqref{eq:MU_MatchingUB} follows since the inner probability in the second
    term of \eqref{eq:MU_GeneralUB} is zero whenever $\PP[Z_{1}\notin\Ac]$, and since
    $\min\{1,\zeta\} \le 1$.  The third term in \eqref{eq:MU_MatchingUB} is obtained similarly,
    and the fourth term follows from the fact that $\min\{1,\zeta\} \le \zeta$.
        
    The lower bound in \eqref{eq:MU_MatchingLB} follows from \eqref{eq:MU_GeneralLB},
    and since the additional assumptions in the second part of the lemma statement imply 
    \begin{align}
        & \frac{\PP\big[(Z_{1},Z_{2})\in\Ac\big]^{2}}{\PP\big[(Z_{1},Z_{2})\in\Ac\,\cap\,(Z_{1},Z_{2}^{\prime})\in\Ac\big]} \nonumber \\ & \qquad=\frac{\PP\big[Z_{1}\in\Ac_{1}\big]^{2}\PP\big[(z_{1},Z_{2})\in\Ac\big]^{2}}{\PP\big[Z_{1}\in\Ac_{1}\big]\PP\big[(z_{1},Z_{2})\in\Ac\big]^{2}}, \\
                    &\qquad =\PP\big[Z_{1}\in\Ac_{1}\big]                   
    \end{align}
    where $z_{1}$ is an arbitrary element of $\Ac_{1}$. The
    third term in the minimization in \eqref{eq:MU_GeneralLB} can be handled similarly.
\end{proof}

A generalization of Lemma \ref{lem:MU_MatchingInd} to 
the probability of a union indexed by $K$ values can be found in \cite[Appendix D]{Thesis}.

\section{Equivalent Forms of Convex Optimization Problems} \label{sec:MU_DUALITY}

The achievable rates and error exponents derived in this paper are presented in both primal and
dual forms, analogously to the LM rate in \eqref{eq:SU_PrimalLM}--\eqref{eq:SU_DualLM}.
The corresponding proofs of equivalence are more involved 
than that of the LM rate (see \cite{Merhav}).  Here we provide two lemmas
that are useful in proving the equivalences.  The following lemma
generalizes the result that \eqref{eq:SU_PrimalLM} and \eqref{eq:SU_DualLM}
are equivalent, and is proved using Lagrange duality \cite[Ch.~5]{Convex}.

\begin{lem} \label{lem:MU_DualityLemma}
    Fix the finite alphabets $\Zc_{1}$ and $\Zc_{2}$, the non-negative
    functions $f(z_{1},z_{2})$ and $g(z_{1},z_{2})$, the distributions $P_{Z_{1}}\in\Pc(\Zc_1)$
    and $P_{Z_{2}}\in\Pc(\Zc_2)$, and a constant $\beta$.  Then
    \begin{equation}
        \min_{\substack{\Ptilde_{Z_{1}Z_{2}}\,:\,\Ptilde_{Z_{1}}=P_{Z_{1}},\Ptilde_{Z_{2}}=P_{Z_{2}}, \\ \EE_{\Ptilde}[\log f(Z_{1},Z_{2})]\ge\beta}} I_{\Ptilde}(Z_{1};Z_{2}) - \EE_{\Ptilde}[\log g(Z_{1},Z_{2})] \label{eq:MU_Primal}
    \end{equation}
    is equal to
    \begin{multline}
        \sup_{\lambda\ge0,\mu_{1}(\cdot)} \sum_{z_{1}}P_{Z_{1}}(z_{1})\mu_{1}(z_{1}) + \lambda\beta \\ -\sum_{z_{2}}P_{Z_{2}}(z_{2})\log\sum_{\zbar_{1}}P_{Z_{1}}(\zbar_{1})f(\zbar_{1},z_{2})^{\lambda}g(\zbar_{1},z_{2})e^{\mu_{1}(\zbar_{1})}, \label{eq:MU_Dual}
    \end{multline}
    where the supremum over $\mu_{1}(\cdot)$ is taken over all real-valued functions on $\Zc_{1}$.
\end{lem}
\begin{proof}
    The Lagrangian \cite[Sec. 5.1.1]{Convex} of the optimization problem in \eqref{eq:MU_Primal}
    is given by
    \begin{multline}
        \hspace*{-1.5ex}L = \sum_{z_{1},z_{2}}\Ptilde_{Z_{1}Z_{2}}(z_{1},z_{2})\bigg(\log\frac{\Ptilde(z_{1},z_{2})}{P_{Z_{1}}(z_{1})P_{Z_{2}}(z_{2})} - \log g(z_{1},z_{2}) \\
             - \lambda\log f(z_{1},z_{2})\bigg) 
            + \sum_{z_{1}}\mu_{1}(z_{1})\big(P_{Z_{1}}(z_{1})-\Ptilde_{Z_{1}}(z_{1})\big) \\ + \sum_{z_{2}}\mu_{2}(z_{2})\big(P_{Z_{2}}(z_{2})-\Ptilde_{Z_{2}}(z_{2})\big) + \lambda\beta, \label{eq:MU_DualityPf1}
    \end{multline}
    where $\lambda\ge0$, $\mu_{1}(\cdot)$ and $\mu_{2}(\cdot)$ are Lagrange multipliers.
    Since the objective in \eqref{eq:MU_Primal} is convex and the constraints are affine,  
    the optimal value is equal to $L$ for some choice of $\Ptilde_{Z_{1}Z_{2}}$ and the
    Lagrange multipliers satisfying the Karush-Kuhn-Tucker (KKT) conditions \cite[Sec. 5.5.3]{Convex}.  
    
    We proceed to simplify \eqref{eq:MU_DualityPf1} using the KKT conditions.  Setting
    $\frac{\partial L}{\partial \Ptilde(z_{1},z_{2})} = 0$ yields
    \begin{multline}
        1 + \log\frac{\Ptilde_{Z_{1}Z_{2}}(z_{1},z_{2})}{P_{Z_{1}}(z_{1})P_{Z_{2}}(z_{2})f(z_{1},z_{2})^{\lambda}g(z_{1},z_{2})} \\ - \mu_{1}(z_{1}) - \mu_{2}(z_{2}) = 0. \label{eq:MU_DualityPf2}
    \end{multline}
    Solving for $\Ptilde_{Z_{1}Z_{2}}(z_{1},z_{2})$ applying 
    the constraint $\Ptilde_{Z_{2}}=P_{Z_{2}}$, and then solving
    for $\mu_2(z_2)$, we obtain
    \begin{equation}
        \mu_{2}(z_{2}) = 1 - \log\sum_{\zbar_{1}}P_{Z_{1}}(\zbar_{1})f(\zbar_{1},z_{2})^{\lambda}g(\zbar_{1},z_{2})e^{\mu_{1}(\zbar_{1})}. \label{eq:MU_DualityPf5}
    \end{equation}
    
    Substituting \eqref{eq:MU_DualityPf2} into \eqref{eq:MU_DualityPf1} yields
    \begin{equation}
        L = -1 + \sum_{z_{1}}\mu_{1}(z_{1})P_{Z_{1}}(z_{1}) + \sum_{z_{2}}\mu_{2}(z_{2})P_{Z_{2}}(z_{2}) + \lambda\beta,
    \end{equation}
    and applying \eqref{eq:MU_DualityPf5} yields \eqref{eq:MU_Dual} with the supremum omitted. 
    It follows that \eqref{eq:MU_Dual} is an upper bound to \eqref{eq:MU_Primal}.
    
    To obtain a matching lower bound, we make use of the log-sum inequality \cite[Thm. 2.7.1]{Cover}
    similarly to \cite[Appendix A]{Merhav}.  For any $\Ptilde_{Z_{1}Z_{2}}$ satisfying the constraints in 
    \eqref{eq:MU_Primal}, we can lower bound the objective in \eqref{eq:MU_Primal} as follows: 
    \begin{align}
        & \sum_{z_{1},z_{2}}\Ptilde_{Z_{1}Z_{2}}(z_{1},z_{2})\log\frac{\Ptilde(z_{1},z_{2})}{P_{Z_{1}}(z_{1})P_{Z_{2}}(z_{2})g(z_{1},z_{2})} \\
        & \ge \sum_{z_{1},z_{2}}\Ptilde_{Z_{1}Z_{2}}(z_{1},z_{2})\nonumber \\ 
        &~~~\times\log\frac{\Ptilde(z_{1},z_{2})}{P_{Z_{1}}(z_{1})P_{Z_{2}}(z_{2})f(z_{1},z_{2})^{\lambda}g(z_{1},z_{2})} + \lambda\beta \label{eq:MU_DualityPf8} \\ 
        &= \sum_{z_{1},z_{2}}\Ptilde_{Z_{1}Z_{2}}(z_{1},z_{2}) \nonumber \\ &~~~\times\log\frac{\Ptilde(z_{1},z_{2})}{P_{Z_{1}}(z_{1})P_{Z_{2}}(z_{2})f(z_{1},z_{2})^{\lambda}g(z_{1},z_{2})e^{\mu_{1}(\zbar_{1})}} \nonumber \\
        &~~~+ \sum_{z_{1}} P_{Z_{1}}(z_{1})\mu_{1}(z_{1}) + \lambda\beta, \label{eq:MU_DualityPf9}  
    \end{align}
    where \eqref{eq:MU_DualityPf8} holds for any $\lambda\ge0$ due to the constraint 
    $\EE_{\Ptilde}[\log f(Z_{1},Z_{2})]\ge\beta$, and \eqref{eq:MU_DualityPf9} holds 
    for any $\mu_{1}(\cdot)$ by an expansion of the logarithm.  Applying 
    the log-sum inequality, we can lower bound \eqref{eq:MU_DualityPf9} by
    the objective in \eqref{eq:MU_Dual}.  Since $\lambda\ge0$ and $\mu_{1}(\cdot)$ 
    are arbitrary, the proof is complete.
\end{proof}

When using Lemma \ref{lem:MU_DualityLemma}, we will typically be interested the 
case that either $g(\cdot,\cdot)=1$, or $f(\cdot,\cdot) = 1$ and $\beta = 0$.

The following lemma will allow certain convex optimization problems to be
expressed in a form where, after some manipulations, Lemma \ref{lem:MU_DualityLemma}
can be applied.

\begin{lem} \label{lem:MU_MinMaxLemma}
    Fix a positive integer $d$ and let $\Dc$ be a convex subset of $\RR^d$.
    Let $f(\zv)$, $g(\zv)$, $g_1(\zv)$ and $g_2(\zv)$ be convex 
    functions mapping $\RR^d$ to $\RR$ such that
    \begin{equation}
        g_1(\zv)+g_2(\zv) \le g(\zv) \label{eq:MU_g1g2g}
    \end{equation}
    for all $\zv\in\Dc$. Then
    \begin{equation}
        \min_{\zv\in\Dc} f(\zv) + \Big[\max\big\{g_1(\zv),g_2(\zv),g(\zv)\big\}\Big]^{+} \label{eq:MU_Equiv1}
    \end{equation}
    is equal to
    \begin{multline}
        \max\bigg\{\min_{\zv\in\Dc} f(\zv) + \Big[\max\big\{g_1(\zv),g(\zv)\big\}\Big]^{+}, \\ \min_{\zv\in\Dc} f(\zv) + \Big[\max\big\{g_2(\zv),g(\zv)\big\}\Big]^{+} \bigg\} \label{eq:MU_Equiv2}.
    \end{multline}
\end{lem} 
\begin{proof}
    We define the following functions ($\nu=1,2$):
    \begin{gather}
        \Phi_0(\zv) \defeq f(\zv) + \big[g(\zv)\big]^{+} \label{eq:MAC_Phi0} \\
        \Phi_{\nu}(\zv) \defeq f(\zv) + \big[\max\big\{g_{\nu}(\zv),g(\zv)\big\}\big]^{+}. \label{eq:MAC_Phi1} 
    \end{gather}
    Since $f(\cdot)$, $g(\cdot)$, $g_1(\cdot)$ and $g_2(\cdot)$ are convex
    by assumption, it follows from the composition rules in
    \cite[Sec. 3.2.4]{Convex} that $\Phi_0(\cdot)$, $\Phi_1(\cdot)$ and
    $\Phi_2(\cdot)$ are also convex.
    
    We wish to show that 
    \begin{equation}
        \min_{\zv\in\Dc} \max\big\{\Phi_1(\zv),\Phi_2(\zv)\big\} = \max\bigg\{\min_{\zv\in\Dc} \Phi_1(\zv), \min_{\zv\in\Dc} \Phi_2(\zv)\bigg\}. \label{eq:MU_EquivStatement}
    \end{equation}
    We define the following regions for $\nu=1,2$:
    \begin{gather}
        \Rc_{\nu} = \big\{\zv:\,\Phi_{\nu}(\zv)>\Phi_{0}(\zv)\big\}. \label{eq:MAC_Region1}
    \end{gather}
    The key observation is that $\mathcal{R}_{1}$ and $\mathcal{R}_{2}$
    are disjoint. To see this, we observe from \eqref{eq:MAC_Phi0}--\eqref{eq:MAC_Phi1}
    that any $\zv\in\Rc_{1} \cap \Rc_{2}$ satisfies $g_1(\zv) > g(\zv)$ and 
    $g_2(\zv) > g(\zv)$.  Combined with \eqref{eq:MU_g1g2g}, these imply 
    $g_1(\zv)<0$ and $g_2(\zv)<0$, and it follows from \eqref{eq:MAC_Phi0}--\eqref{eq:MAC_Phi1}
    that $\Phi_0(\zv)=\Phi_1(\zv)=\Phi_2(\zv)$, in contradiction with the 
    assumption that $\zv\in\Rc_{1} \cap \Rc_{2}$.  Thus, $\Rc_{1} \cap \Rc_{2}$ is empty,
    which implies that $g_1(\zv)$ and $g_2(\zv)$ cannot simultaneously be the unique  
    maximizers in \eqref{eq:MAC_Phi1} for both $\nu=1$ and $\nu=2$.  Combining
    this with \eqref{eq:MAC_Phi0}, we obtain
    \begin{equation}
        \Phi_{0}(\zv) = \min\big\{\Phi_{1}(\zv),\Phi_{2}(\zv)\big\}. \label{eq:MAC_Phi0_Alt} \\
    \end{equation}
    
    To prove \eqref{eq:MU_EquivStatement}, we use a proof by contradiction.  Let the left-hand side
    and right-hand side be denoted by $f^{*}$ and $\tilde{f}^{*}$ respectively.  The inequality
    $f^{*}\ge\tilde{f}^{*}$ holds by definition, so we assume that $f^{*} > \tilde{f}^{*}$.  
    Let $\zv_{1}^{*}$ and $\zv_{2}^{*}$ minimize $\Phi_{1}$ and $\Phi_{2}$ respectively on the
    right-hand side of \eqref{eq:MU_EquivStatement}, so that
    \begin{equation}
        \tilde{f}^{*} = \max\big\{\Phi_{1}(\zv_{1}^{*}),\Phi_{2}(\zv_{2}^{*})\big\}.\label{eq:MAC_f2_star}
    \end{equation}
    The assumption $f^{*} > \tilde{f}^{*}$ implies that
    \begin{align}
        \Phi_{2}(\zv_{1}^{*}) &> \Phi_{1}(\zv_{1}^{*}) \label{eq:MAC_Ineq1} \\
        \Phi_{1}(\zv_{2}^{*}) &> \Phi_{2}(\zv_{2}^{*}). \label{eq:MAC_Ineq2}
    \end{align}
    Next, we define
    \begin{align}
        \hat{\Phi}_{\nu}(\lambda)\defeq\Phi_{\nu}\big(\lambda \zv_{1}^{*}+(1-\lambda)\zv_{2}^{*}\big)\label{eq:MAC_Phi0'}
    \end{align}
    for $\lambda\in[0,1]$ and $\nu=0,1,2$. Since any convex function is also convex when 
    restricted to a straight line \cite[Section 3.1.1]{Convex}, it follows 
    that $\hat{\Phi}_{0}$, $\hat{\Phi}_{1}$ and $\hat{\Phi}_{2}$ 
    are convex in $\lambda$. From \eqref{eq:MAC_Ineq1}--\eqref{eq:MAC_Ineq2}, we have 
    \begin{align}
        \hat{\Phi}_{2}(1) &> \hat{\Phi}_{1}(1)\label{eq:MAC_Ineq3} \\
        \hat{\Phi}_{1}(0) &> \hat{\Phi}_{2}(0).\label{eq:MAC_Ineq4}
    \end{align}
    Since $\hat{\Phi}_{1}$ and $\hat{\Phi}_{2}$ are convex, they
    are also continuous (at least in the region that they are finite), 
    and it follows that the two must intersect somewhere
    in $(0,1)$, say at $\lambda^{*}$. Therefore, 
    \begin{flalign}
        \hat{\Phi}_{0}(\lambda^{*}) & =\min\big\{\hat{\Phi}_{1}(\lambda^{*}),\hat{\Phi}_{2}(\lambda^{*})\big\}\label{eq:MAC_Phi0Bound1}\\
        & =\max\big\{\hat{\Phi}_{1}(\lambda^{*}),\hat{\Phi}_{2}(\lambda^{*})\big\}\label{eq:MAC_Phi0Bound2}\\
        & \ge \min_{\zv\in\Dc} \max\big\{\Phi_1(\zv),\Phi_2(\zv)\big\} \label{eq:MAC_Phi0Bound3}\\
        & = f^{*},\label{eq:MAC_Phi0Bound4}
    \end{flalign}
    where \eqref{eq:MAC_Phi0Bound1} follows from \eqref{eq:MAC_Phi0_Alt}. 
    Finally, we have the following contradiction: (i) Combining 
    \eqref{eq:MAC_Phi0Bound4} with the assumption  that $f^{*} > \tilde{f}^{*}$,  we have 
    \begin{equation}
        \hat{\Phi}_{0}(\lambda^{*}) > \tilde{f}^{*} =\max\{\hat{\Phi}_{1}(1),\hat{\Phi}_{2}(0)\}, \label{eq:MAC_PhiBound1} \\        
    \end{equation}
    where the equality follows from \eqref{eq:MAC_f2_star}; 
    (ii) From \eqref{eq:MAC_Phi0_Alt}, we have $\hat{\Phi}_{0}(\lambda)=\min\{\hat{\Phi}_{1}(\lambda),\hat{\Phi}_{2}(\lambda)\}$,
    and it follows from \eqref{eq:MAC_Ineq3}--\eqref{eq:MAC_Ineq4} that
    $\hat{\Phi}_{0}(1)=\hat{\Phi}_{1}(1)$ and $\hat{\Phi}_{0}(0)=\hat{\Phi}_{2}(0)$.
    Using the convexity of $\hat{\Phi}_{0}$ and Jensen's inequality,
    we have 
    \begin{align}
        \hat{\Phi}_{0}(\lambda^{*}) & \le \lambda^{*}\hat{\Phi}_{1}(1)+(1-\lambda^{*})\hat{\Phi}_{2}(0) \\
                                       & \le \max\{\hat{\Phi}_{1}(1),\hat{\Phi}_{2}(0)\}.\label{eq:MAC_PhiBound2}
    \end{align}
\end{proof}

\section{Multiple-Access Channel Proofs} \label{sub:MAC_PROOFS}

\subsection{Preliminary Lemma for Proving Theorem \ref{thm:MAC_DualRate}}

The following lemma expresses \eqref{eq:MAC_R12_LM} in a form that is
more amenable to Lagrange duality techniques.

\begin{lem} \label{lem:MAC_AltRate}
    The achievable rate condition in \eqref{eq:MAC_R12_LM} holds if the following holds for at least one of $\nu=1,2$:
    \begin{multline}
        R_{1}+R_{2} \le \min_{\substack{\Ptilde_{X_{1}X_{2}Y}\in\SetTiiicc(Q_{1}\times Q_{2}\times W)\\I_{\Ptilde}(X_{\nu};Y)\le R_{\nu}}} \\ D(\Ptilde_{X_{1}X_{2}Y}\|Q_{1}\times Q_{2}\times P_{Y})\label{eq:MAC_R12_1_LM} 
    \end{multline}
\end{lem}
\begin{proof}
    We first write the condition in \eqref{eq:MAC_R12_LM} as
    \begin{multline}
        \hspace*{-2.5ex}0 \le \min_{\Ptilde_{X_{1}X_{2}Y}\in\SetTiiicc(Q_{1}\times Q_{2}\times W)} 
        \max\big\{ D(\Ptilde_{X_{1}X_{2}Y}\|Q_{1}\times Q_{2}\times P_{Y}) \\ - (R_1 + R_2), I_{\Ptilde}(X_{1};Y) -  R_{1},\, I_{\Ptilde}(X_{2};Y) - R_{2} \big\}, \label{eq:MAC_Rate_nu_0}
    \end{multline}
    where the equivalence is seen by noting that this condition is always satisfied
    when the minimizer satisfies $I_{\Ptilde}(X_{1};Y) > R_{1}$ or $I_{\Ptilde}(X_{2};Y) > R_{2}$.
    Next, we claim that this condition is equivalent to the following holding for
    at least one of $\nu = 1,2$:
    \begin{multline}
        0 \le \min_{\Ptilde_{X_{1}X_{2}Y}\in\SetTiiicc(Q_{1}\times Q_{2}\times W)} 
        \max\big\{ \\ D(\Ptilde_{X_{1}X_{2}Y}\|Q_{1}\times Q_{2}\times P_{Y}) - (R_1 + R_2), I_{\Ptilde}(X_{\nu};Y) -  R_{\nu}\big\}. \label{eq:MAC_Rate_nu}
    \end{multline}
    This is seen by applying Lemma \ref{lem:MU_MinMaxLemma} with the following identifications ($\nu=1,2$):
    \begin{align} 
        f(\zv)   &= 0 \\
        g(\zv)   &= D\big(\Ptilde_{X_{1}X_{2}Y}\|Q_{1}\times Q_{2}\times P_{Y}\big)-R_1-R_2 \\
        g_{\nu}(\zv) &= I_{\Ptilde}(X_{\nu};Y)-R_{\nu}.
    \end{align}
    From the last two lines and the identity
    \begin{multline}
        D\big(\Ptilde_{X_{1}X_{2}Y}\|Q_{1}\times Q_{2}\times P_{Y}\big) = I_{\Ptilde}(X_{1};Y) \\ + I_{\Ptilde}(X_{2};Y) + I_{\Ptilde}(X_1;X_2|Y),
    \end{multline} 
    which holds under the constraints present
    in the definition of $\SetTiiicc$ in \eqref{eq:MAC_SetT12}, 
    we see that the condition in \eqref{eq:MU_g1g2g} is satisfied.

    Finally, the lemma follows from \eqref{eq:MAC_Rate_nu} by reversing 
    the step used to obtain \eqref{eq:MAC_Rate_nu_0}.
\end{proof}

\subsection{Proof of First Part of Theorem \ref{thm:MAC_DualRate}} \label{sub:MAC_DUAL_PROOF}

Each expression in the theorem statement is 
derived similarly, so we focus on \eqref{eq:MAC_R12_1_DualLM}. 
We claim that \eqref{eq:MAC_R12_1_LM} holds if and only if
\begin{multline}
    R_1  \le \max_{\rho_{2}\in[0,1]} \min_{\Ptilde_{X_{1}X_{2}Y}\in\SetTiiicc(P_{X_{1}X_{2}Y})} I_{\Ptilde}(X_{1};Y) \\ +\rho_{2}I_{\Ptilde}(X_{2};X_{1},Y)  - \rho_{2}R_{2}, \label{eq:MAC_DualProof1}
\end{multline}
where here and in the remainder of the proof we write $P_{X_1X_2Y} \triangleq Q_1 \times Q_2 \times W$.
To see this, we first note that by the identity 
\begin{equation}
D\big(P_{X_{1}X_{2}Y}\|Q_{1}\times Q_{2}\times P_{Y}\big) = I_{P}(X_{1};Y) + I_{P}(X_2;X_1,Y), \label{eq:MAC_ExpandedD}
\end{equation} 
\eqref{eq:MAC_Rate_nu} (with $\nu = 1$) is equivalent to
\begin{multline}
    R_1 \le \min_{\Ptilde_{X_{1}X_{2}Y}\in\SetTiiicc(Q_{1}\times Q_{2}\times W)} I_{\Ptilde}(X_1;Y) \\ + \big[I_{\Ptilde}(X_2;X_1,Y) - R_2\big]^+.
\end{multline}
Next, we apply the identity $[\alpha]^{+}=\max_{0\le\rho_{1}\le1}\rho_1\alpha$.  
The resulting objective is linear in $\rho_{1}$ and jointly
convex in $(P_{X_{1}X_{2}Y},\Ptilde_{X_{1}X_{2}Y})$, so we can apply Fan's minimax theorem
\cite{Minimax} to interchange the maximization and minimizations, thus yielding \eqref{eq:MAC_DualProof1}.

We define the sets
\begin{align}
    \SetTiiicc'(P_{X_{1}X_{2}Y},\hat{P}_{X_{1}Y}) \defeq \Big\{\Ptilde_{X_{1}X_{2}Y}\in\Pc(\Xc_{1}\times\Xc_{2}\times\Yc)\,:\,&\nonumber \\
     \Ptilde_{X_{2}}=P_{X_{2}}, \Ptilde_{X_{1}Y}=\hat{P}_{X_{1}Y},& \nonumber \\ \EE_{\Ptilde}[\log q(X_{1},X_{2},Y)]\ge\EE_{P}[\log q(X_{1},X_{2},Y)]\Big\}&\label{eq:MAC_SetU} 
\end{align}
\vspace*{-5ex}
\begin{multline}
        \SetTiiicc''(P_{X_{1}X_{2}Y}) \defeq\Big\{\hat{P}_{X_{1}Y}\in\Pc(\Xc_{1}\times\Yc)\,:\,\\ \hat{P}_{X_{1}}=P_{X_{1}},\hat{P}_{Y}=P_{Y}\Big\}.\label{eq:MAC_SetV}
\end{multline}
It follows that $\Ptilde_{X_{1}X_{2}Y}\in\SetTiiicc(P_{X_{1}X_{2}Y})$ (see \eqref{eq:MAC_SetT12})
if and only if $\Ptilde_{X_{1}X_{2}Y}\in\SetTiiicc'(P_{X_{1}X_{2}Y},\hat{P}_{X_{1}Y})$
for some $\hat{P}_{X_{1}Y}\in\SetTiiicc''(P_{X_{1}X_{2}Y})$. We
can therefore replace the minimization over $\Ptilde_{X_{1}X_{2}Y}\in\SetTiiicc(P_{X_{1}X_{2}Y})$
in \eqref{eq:MAC_DualProof1} with minimizations over $\hat{P}_{X_{1}Y}\in\SetTiiicc''(P_{X_{1}X_{2}Y})$
and $\Ptilde_{X_{1}X_{2}Y}\in\SetTiiicc'(P_{X_{1}X_{2}Y},\hat{P}_{X_{1}Y})$.

We prove the theorem by performing the minimization in several steps, and performing
multiple applications of Lemma \ref{lem:MU_DualityLemma}.  Each such application will yield
an overall optimization of the form $\sup\min\sup\{\cdot\}$, and we will implicitly use
Fan's minimax theorem \cite{Minimax} to obtain an equivalent expression of the form $\sup\sup\min\{\cdot\}$.
Thus, we will leave the optimization of the dual variables (i.e.~the suprema) until the final step.

\subsubsection*{Step 1}

We first consider the minimization of the term $I_{\Ptilde}(X_{1};X_{2},Y)$
over $\Ptilde_{X_{1}X_{2}Y}$ when $P_{X_{1}X_{2}Y}\in\SetScc(\Qv)$
and $\hat{P}_{X_{1}Y}\in\SetTiiicc''(P_{X_{1}X_{2}Y})$ are fixed,
and thus all of the terms in the objective in \eqref{eq:MAC_DualProof1}
other than $I_{\Ptilde}(X_{1};X_{2},Y)$ are fixed. The minimization 
is given by
\begin{equation}
    \Fsf_{1} \defeq \min_{\Ptilde_{X_{1}X_{2}Y}\in\SetTiiicc'(P_{X_{1}X_{2}Y},\hat{P}_{X_{1}Y})} I_{\Ptilde}(X_{1};X_{2},Y).\label{eq:MAC_DualS1_1}
\end{equation}
Applying Lemma \ref{lem:MU_DualityLemma} with $P_{Z_{1}}=P_{X_{2}}$, $P_{Z_{2}}=\hat{P}_{X_{1}Y}$
and $\mu_{1}(\cdot)=a_{2}(\cdot)$, we obtain the dual expression
\begin{multline}
    \hspace*{-2.5ex}\Fsf_{1} = - \sum_{x_{1},y}\hat{P}_{X_{1}Y}(x_{1},y)\log\sum_{\xbar_{2}}P_{X_{2}}(\xbar_{2})q(x_{1},\xbar_{2},y)^{s}e^{a_{2}(\xbar_{2})} \\
            + s\sum_{x_{1},x_{2},y}P_{X_{1}X_{2}Y}(x_{1},x_{2},y)\log q(x_{1},x_{2},y) + \\ \sum_{x_{2}}P_{X_{2}}(x_{2})a_{2}(x_{2}).\label{eq:MAC_DualS1_8}
\end{multline}

\subsubsection*{Step 2}

After Step 1, the overall objective (see \eqref{eq:MAC_DualProof1})
is given by
\begin{equation}
    I_{\hat{P}}(X_{1};Y) + \rho_{2}\Fsf_{1} - \rho_2 R_2,\label{eq:MAC_DualS2_1}
\end{equation}
where we have replaced $I_{\Ptilde}(X_{1};Y)$ by $I_{\hat{P}}(X_{1};Y)$
due to the constraint $\Ptilde_{X_{1}Y}=\hat{P}_{X_{1}Y}$ in
\eqref{eq:MAC_SetU}. Since the only terms involving $\hat{P}_{X_{1}Y}$
are $I_{\hat{P}}(X_{1};Y)$ and the first term in 
\eqref{eq:MAC_DualS1_8}, we consider the minimization
\begin{multline}
    \Fsf_{2} \defeq \min_{\hat{P}_{X_{1}Y}\in\SetTiiicc''(P_{X_{1}X_{2}Y})} I_{\hat{P}}(X_1;Y)  
                 -\rho_{2}\sum_{x_{1},y}\hat{P}_{X_{1}Y}(x_{1},y)\\ \times\log\sum_{\xbar_{2}}P_{X_{2}}(\xbar_{2})q(x_{1},\xbar_{2},y)^{s}e^{a_{2}(\xbar_{2})}.\label{eq:MAC_DualS2_2}
\end{multline}
Applying Lemma \ref{lem:MU_DualityLemma} with $P_{Z_{1}}=P_{X_{1}}$, $P_{Z_{2}}=P_{Y}$
and $\mu_{1}(\cdot)=a_{1}(\cdot)$, we obtain
\begin{multline}
    \Fsf_{2} = \sum_{x_{1}}P_{X_{1}}(x_{1})a_{1}(x_{1})-\sum_{y}P_{Y}(y)\log\sum_{\xbar_{1}}P_{X_{1}}(\xbar_{1}) \\ \times\bigg(\sum_{\xbar_{2}}P_{X_{2}}(\xbar_{2})q(\xbar_{1},\xbar_{2},y)^{s}e^{a_{2}(\xbar_{2})}\bigg)^{\rho_{2}}e^{a_{1}(\xbar_{1})}.\label{eq:MAC_DualS2_9}
\end{multline}

\subsubsection*{Step 3}

From \eqref{eq:MAC_DualS1_8}, \eqref{eq:MAC_DualS2_1} and \eqref{eq:MAC_DualS2_9},
the overall objective is now given by
\begin{multline}
    \Fsf_{3} \defeq \Fsf_{2} - \rho_2 R_2 \\ + \rho_{2}\sum_{x_{1},x_{2},y}P_{X_{1}X_{2}Y}(x_{1},x_{2},y)\log q(x_{1},x_{2},y)^{s}e^{a_{2}(x_{2})}. \label{eq:MAC_DualS3_1}
\end{multline}
Substituting \eqref{eq:MAC_DualS2_9} and performing some 
rearrangements, we obtain the objective in \eqref{eq:MAC_R12_1_DualLM},
and conclude the proof by taking the supremum over $\rho_2$, $s$,
$a_1(\cdot)$ and $a_2(\cdot)$.

\subsection{Proof of Theorem \ref{thm:MAC_NaiveML}}

We begin with the following proposition, which  
shows that the exponents $(\Eircc,\Eiircc,\Erccprime)$ (see \eqref{eq:MAC_Er1_LM}
and \eqref{eq:MAC_Er12'}) under ML decoding coincide 
with those by Liu and Hughes in the absence of time-sharing \cite{MACExponent4}.
\begin{prop} \label{prop:MAC_ExpML}
    Under ML decoding (i.e.~$q(x_{1},x_{2},y)=W(y|x_{1},x_{2})$), $\Enurcc$ and $\Erccprime$
    can be expressed as 
    \begin{align}
        &\hspace*{-2ex}\Enurcc(\Qv,R_{\nu})            =\min_{P_{X_{1}X_{2}Y}\in\SetScc(\Qv)}D(P_{X_{1}X_{2}Y}\|Q_{1}\times Q_{2}\times W)\nonumber \\ 
        &\qquad\qquad\qquad\qquad\quad+\big[I_{P}(X_{\nu};X_{\nu^c},Y)-R_{\nu}\big]^{+} \label{eq:MAC_Er1_ML} \\
        &\hspace*{-2ex}\Erccprime(\Qv,R_{1},R_{2})  =\min_{P_{X_{1}X_{2}Y}\in\SetScc(\Qv)}D(P_{X_{1}X_{2}Y}\|Q_{1}\times Q_{2}\times W) \nonumber \\
        &\quad+\big[D(P_{X_{1}X_{2}Y}\|Q_{1}\times Q_{2}\times P_{Y})-(R_{1}+R_{2})\big]^{+}. \label{eq:MAC_Er12Naive_ML} 
    \end{align}
\end{prop}
\begin{proof} 
    The proof is similar to that of \cite[Lemma 9]{GallagerCC}, so
    we provide only an outline, and we focus on the type-12 exponent.
    Consider any pair $(P_{X_1X_2Y},\Ptilde_{X_1X_2Y})$ satisfying the constraints
    of \eqref{eq:MAC_Er12'}.  If $D(\Ptilde_{X_{1}X_{2}Y}\|Q_{1}\times Q_{2}\times P_{Y}) \ge D(P_{X_{1}X_{2}Y}\|Q_{1}\times Q_{2}\times P_{Y})$,
    we can lower bound the objective of \eqref{eq:MAC_Er12'} by that of \eqref{eq:MAC_Er12Naive_ML}.
    In the remaining case, we may use the constraint $\EE_{\Ptilde}[\log W] \ge \EE_{P}[\log W]$
    to lower bound the objective in \eqref{eq:MAC_Er12'} by that of 
    \eqref{eq:MAC_Er12Naive_ML} with $\Ptilde_{X_1X_2Y}$ in place of $P_{X_1X_2Y}$.
    This proves that \eqref{eq:MAC_Er12Naive_ML} lower bounds \eqref{eq:MAC_Er12'},
    and the matching upper bound follows immediately from the fact that 
    $\Ptilde_{X_1X_2Y} = P_{X_1X_2Y}$ satisfies the constraints of the minimization in \eqref{eq:MAC_Er12'}.
\end{proof}

We know that $\Eiiircc\ge \Erccprime$
always holds, and hence the left-hand side of \eqref{eq:MAC_NaiveML}
is greater than or equal to the right-hand side. It remains to 
prove the reverse inequality.
From the definition of $\SetTiiicc(P_{X_{1}X_{2}Y})$, $\Ptilde_{X_{1}X_{2}Y}=P_{X_{1}X_{2}Y}$
always satisfies the constraints of \eqref{eq:MAC_Er12_LM}, and hence
\begin{equation}
    \Eiiircc(\Qv,R_{1},R_{2})\le F_{12}(\Qv,R_{1},R_{2}),\label{eq:MAC_NaiveML1}
\end{equation}
where
\begin{align}
    & F_{12}(\Qv,R_{1},R_{2})\defeq\min_{P_{X_{1}X_{2}Y}\in\SetScc(\Qv)}D(P_{X_{1}X_{2}Y}\|Q_{1}\times Q_{2}\times W)\nonumber\\
    &\qquad+\Big[\max\Big\{I_{P}(X_{1};Y)-R_{1},I_{P}(X_{2};Y)-R_{2}, \nonumber \\ 
    &\qquad D\big(P_{X_{1}X_{2}Y}\|Q_{1}\times Q_{2}\times P_{Y}\big)-R_{1}-R_{2}\Big\}\Big]^{+}.\label{eq:MAC_NaiveML2}
\end{align}
We will prove \eqref{eq:MAC_NaiveML} by showing that
\begin{multline}
    \min\big\{ \Eircc(\Qv,R_{1}),\Eiircc(\Qv,R_{2}),F_{12}(\Qv,R_{1},R_{2})\big\} \\ \le\min\big\{ \Eircc(\Qv,R_{1}),\Eiircc(\Qv,R_{2}),\Erccprime(\Qv,R_{1},R_{2})\big\}.\label{eq:MAC_NaiveML3}
\end{multline}
It suffices to show that whenever $F_{12}$ exceeds $\Erccprime$,
$F_{12}$ also greater than or equal to either $\Eircc$
or $\Eiircc$. Comparing \eqref{eq:MAC_Er12Naive_ML}
and \eqref{eq:MAC_NaiveML2}, the objective in \eqref{eq:MAC_NaiveML2} 
only exceeds that of \eqref{eq:MAC_Er12Naive_ML}
when the maximum in \eqref{eq:MAC_NaiveML2} is achieved by $I_{P}(X_{1};Y)-R_{1}$
or $I_{P}(X_{2};Y)-R_{2}$. We show that the former implies $F_{12}\ge \Eiircc$;
it can similarly be shown that the latter implies $F_{12}\ge \Eircc$.
If $I_{P}(X_{1};Y)-R_{1}$ achieves the maximum, we have 
\begin{equation}
    I_{P}(X_{1};Y)-R_{1}\ge D\big(P_{X_{1}X_{2}Y}\|Q_{1}\times Q_{2}\times P_{Y}\big)-R_{1}-R_{2}. \label{eq:MAC_NaiveML4}
\end{equation}
Using the identity \eqref{eq:MAC_ExpandedD},
we can write \eqref{eq:MAC_NaiveML4} as $I_{P}(X_{2};X_{1},Y) \le R_{2}$.
For any $P_{X_{1}X_{2}Y}$ satisfying this property, the objective
in \eqref{eq:MAC_Er1_ML} (with $\nu=2$) equals $D(P_{X_{1}X_{2}Y}\|Q_{1}\times Q_{2}\times W)$, 
and thus cannot exceed the objective in \eqref{eq:MAC_NaiveML2}. It follows
that $F_{12}\ge \Eiircc$.

\section{Refined Superposition Coding Proofs} \label{sub:RSC_PROOFS}

\subsection{A Preliminary Lemma}

Similarly to Lemma \ref{lem:MAC_AltRate} for the MAC, the following lemma gives
an alternative expression for \eqref{eq:RSC_AltR0} that
is more amenable to Lagrange duality techniques.

\begin{lem}
    The condition in \eqref{eq:RSC_AltR0} holds if and only if the following holds for at least one of $u=1,2$:
    \begin{multline}
        R_{0} \le\min_{\Ptilde_{UXY}\in\SetTocc(Q_{UX}\times W)}I_{\Ptilde}(U;Y) + \Big[\max\Big\{ \\ Q_{U}(u)\big(I_{\Ptilde}(X;Y|U=u)-R_{1u}\big),I_{\Ptilde}(X;Y|U)-R_{1}\Big\}\Big]^{+} \label{eq:RSC_AltR0_1} 
    \end{multline}
\end{lem} 
\begin{proof}
    This is a special case of Lemma \ref{lem:MU_MinMaxLemma} in Appendix \ref{sec:MU_DUALITY} with the following identifications ($u=1,2$):
    \begin{align} 
        f(\zv)   &= I_{\Ptilde}(U;Y) \\
        g(\zv)   &= I_{\Ptilde}(X;Y|U)-R_{1} \\
        g_u(\zv) &= Q_{U}(u)\big(I_{\Ptilde}(X;Y|U=u)-R_{1u}\big),
    \end{align}
    where we recall that $R_1 = \sum_{u} Q_U(u)R_{1u}$. In this case, the condition in \eqref{eq:MU_g1g2g} holds with equality.
\end{proof}

\subsection{Proof of First Part of Theorem \ref{thm:RSC_Dual}} \label{sub:RSC_DUAL1}

We show the equivalence of \eqref{eq:RSC_AltR0_1} ($\nu=1$) and \eqref{eq:RSC_DualR0_1} ($u=1$);
identical arguments apply for $\nu=u=2$.
The primal expression is written in terms of a minimization over $\Ptilde_{UXY}$.
It is convenient to ``split'' this distribution into three distributions:
$\Ptilde_{UY}$, $\hat{P}_{XY}\defeq\Ptilde_{XY|U}(\cdot,\cdot|1)$
and $\hat{\hat{P}}_{XY}\defeq\Ptilde_{XY|U}(\cdot,\cdot|2)$.  Using 
a similar argument to the start of Section \ref{sub:MAC_DUAL_PROOF},
we can write the right-hand side of \eqref{eq:RSC_AltR0_1} as
\begin{multline}
   \hspace*{-2ex}\sup_{\rho_{1}\in[0,1],\rho_{2}\in[0,1]}\min_{\Ptilde_{UY},\hat{P}_{XY},\hat{\hat{P}}_{XY}}I_{\Ptilde}(U;Y)+\rho_{1}Q_{U}(1)I_{\hat{P}}(X;Y) \\ +\rho_{1}\rho_{2}Q_{U}(2)I_{\hat{\hat{P}}}(X;Y) - \rho_{1}R_{11} - \rho_{1}\rho_{2}R_{12}.
\end{multline}
Defining $P_{UXY} \triangleq Q_{UX}\times W$, the minimization is subject 
to the constraints (i) $\Ptilde_{U}=Q_{U}$, (ii) $\hat{P}_{X}=Q_{X|U}(\cdot|1)$,
(iii) $\hat{\hat{P}}_{X}=Q_{X|U}(\cdot|2)$, (iv) $\Ptilde_{Y}=P_{Y}$,
(v) $\hat{P}_{Y}=\Ptilde_{Y|U}(\cdot|1)$, (vi) $\hat{\hat{P}}_{Y}=\Ptilde_{Y|U}(\cdot|2)$,
(vii) $Q_{U}(1)\EE_{\hat{P}}[\log q(X,Y)]+Q_{U}(2)\EE_{\hat{\hat{P}}}[\log q(X,Y)]\ge\EE_{P}[\log q(X,Y)]$.

Similarly to Section \ref{sub:MAC_DUAL_PROOF}, we apply the minimization 
in several steps, making repeated use of Lemma \ref{lem:MU_DualityLemma}.
We implicitly apply Fan's minimax theorem \cite{Minimax} after each step,
so that the supremum over the dual variables can be left until the end.
We provide less detail than the amount given in Section \ref{sub:MAC_DUAL_PROOF},
since the general steps are similar.

\subsubsection*{Step 1}

For given joint distributions $\Ptilde_{UY}$ and $\hat{P}_{XY}$,
the minimization $\min_{\hat{\hat{P}}_{XY}}I_{\hat{\hat{P}}}(X;Y)$ subject
to the constraints (iii), (vi) and (vii) has a dual expression given by
\begin{equation}
    \Fsf_{1} \triangleq -\Fsf_{1,1} + \Fsf_{1,2} + \Fsf_{1,3} - s Q_{U}(1)\Fsf_{1,4},
\end{equation}
where
\begin{align}
    \Fsf_{1,1} &\defeq \sum_{y}\Ptilde_{Y|U}(y|2)\log\sum_{\xbar}Q_{X|U}(\xbar|2)q(\xbar,y)^{s Q_{U}(2)}e^{a_{2}(\xbar)} \\
    \Fsf_{1,2} &\defeq \sum_{x_{2}}Q_{X|U}(x|2)a_{2}(x) \\
    \Fsf_{1,3} &\defeq s\sum_{x,y}P_{XY}(x,y)\log q(x,y) \\
    \Fsf_{1,4} &\defeq \sum_{x,y}\hat{P}_{XY}(x,y)\log q(x,y),
\end{align}
and where $s\ge0$ and $a_2(\cdot)$ are dual variables.

\subsubsection*{Step 2}

For a given joint distribution $\Ptilde_{UY}$, the minimization 
$\min_{\hat{P}_{XY}}I_{\hat{P}}(X;Y)-s\rho_{2}Q_{U}(2)\Fsf_{1,4}(\hat{P}_{XY})$
subject to (ii) and (v) has a dual expression given by
\begin{equation}
    \Fsf_{2} \triangleq -\Fsf_{2,1} + \Fsf_{2,2},
\end{equation}
where
\begin{align}
    \Fsf_{2,1} &\defeq \sum_{y}\Ptilde_{Y|U}(y|1)\log\sum_{\xbar}Q_{X|U}(\xbar|1)q(\xbar,y)^{s\rho_{2}Q_{U}(2)}e^{a_{1}(\xbar)} \\
    \Fsf_{2,2} &\defeq \sum_{x}Q_{X|U}(x|1)a_{1}(x),
\end{align}
and where $a_1(\cdot)$ is a dual variable.

\subsubsection*{Step 3}

Next, we consider the minimization 
$\min_{\Ptilde_{UY}}I_{\Ptilde}(U;Y)-\rho_{1}Q_{U}(1)\Fsf_{2,1}-\rho_{1}\rho_{2}Q_{U}(2)\Fsf_{1,1}$
subject to (i) and (iv).  The objective can equivalently be expressed as
\begin{multline}
    \Fsf_{3} \triangleq I_{\Ptilde}(U;Y) - \sum_{u}\rho_1(u)\sum_{y}\Ptilde_{UY}(u,y) \\ \times\log\sum_{\xbar}Q_{X|U}(\xbar|u)q(\xbar,y)^{s_1(u)}e^{a(u,\xbar)}
\end{multline}
using the definitions in \eqref{eq:RSC_DualVars1} along with $a(u,x) \triangleq a_{u}(x)$.  
The dual expression is given by
\begin{multline}
    \Fsf_{3} = \sum_{u}Q_{U}(u)b(u)-\sum_{y}P_{Y}(y)\log\sum_{\ubar}Q_{U}(\ubar) \\ \times\bigg(\sum_{\xbar}Q_{X|U}(\xbar|\ubar)q(\xbar,y)^{s_1(\ubar)}e^{a(\ubar,\xbar)}\bigg)^{\rho_1(u)}e^{b(\ubar)},
\end{multline}
where $b(\cdot)$ is a dual variable.

\subsubsection*{Step 4}

The final objective is given by
\begin{multline}
    \Fsf_{4} \defeq \Fsf_{3}+\rho_{1}Q_{U}(1)\Fsf_{2,1}+\rho_{1}\rho_{2}Q_{U}(2)\big(\Fsf_{1,2}+\Fsf_{1,3}) \\ - \sum_{u=1,2}\rho_1(u)Q_U(u)R_{1u}.
\end{multline}
After applying some algebraic manipulations, we obtain the dual expression
\begin{multline}
    \hspace*{-1.5ex}- \sum_{u=1,2}\rho_1(u)Q_U(u)R_{1u} + \sum_{u,x,y}P_{UXY}(u,x,y) \\ \hspace*{-1.5ex}\times\log\frac{\Big(q(x,y)^{s_1(u)}e^{a(u,x)}\Big)^{\rho_1(u)}e^{b(u)}}{\sum_{\ubar}Q_{U}(\ubar)\Big(\sum_{\xbar}Q_{X|U}(\xbar|\ubar)q(\xbar,y)^{s_1(\ubar)}e^{a(\ubar,\xbar)}\Big)^{\rho_1(\ubar)}e^{b(\ubar)}}. \label{eq:RSC_DualPfFinal}
\end{multline}
To conclude the proof, we show that the variable $b(u)$ can be removed
from the numerator and denominator in \eqref{eq:RSC_DualPfFinal} without
affecting the dual optimization.  For $\rho_1>0$ and $\rho_2>0$,
this follows by factoring $b(u)$ into $a(u,x)$.  Using the identity
$\EE[e^{b(U)}] \ge e^{\EE[b(U)]}$ (by Jensen's inequality), we find that the
optimal value of the objective is zero when $\rho_1=0$, regardless
of whether $b(u)$ is present.  For the remaining case, namely $\rho_1>0$ and
$\rho_2=0$, the objective depends on $a(u,x)$ only for $u=1$.  Moreover,
since \eqref{eq:RSC_DualPfFinal} depends on $b(\cdot)$ only through
the difference $b(2)-b(1)$, we may set $b(2)=0$ without loss of generality.
The remaining parameter $b(1)$ can be factored into $a(1,x)$.

\subsection{Proof of Second Part of Theorem \ref{thm:RSC_Dual}} \label{sub:RSC_DUAL2}

We focus on the derivation of \eqref{eq:RSC_DualR0_1} with $u=1$, since the case $u=2$ is handled similarly.  The ideas used in the derivation are similar to those for 
the MAC (see the proof of Theorem \ref{thm:MAC_DualRate}), but the details are more involved.

Applying Lemma \ref{lem:MU_UpperInd} to the union in \eqref{eq:RSC_RefProof2}, 
with $Z_{1}(i)=\Xv_{1}^{(1,i)}$ and $Z_{2}(j)=\Xv_{2}^{(1,j)}$, we obtain
\begin{multline}
    \peobar \le \EE\Bigg[\min\Bigg\{1,(M_{0}-1)\EE\Bigg[\min\Bigg\{1, 
        M_{11}\EE\bigg[\min\bigg\{1, \\ M_{12}\PP\bigg[\frac{q^{n}\big(\Xvbar,\Yv\big)}{q^{n}(\Xv,\Yv)}\ge1\,\Big|\,\Xvbar_{1}\bigg]\bigg\}\,\Big|\,\Uvbar\bigg]\Bigg\}\,\Bigg|\,\Uv,\Xv,\Yv\Bigg]\Bigg\}\Bigg].
\end{multline}
Using \eqref{eq:RSC_SplitMetric}, Markov's inequality, and $\min\{1,\zeta\} \le \zeta^{\rho}$ 
($\rho\in[0,1]$), we obtain\footnote{In the case of continuous alphabets, the summations should be replaced by integrals as necessary.}
\begin{align}
    &\peobar \le (M_{0}M_{11}^{\rho_{1}}M_{12}^{\rho_{1}\rho_{2}})^{\rho_{0}} \sum_{\uv,\xv_{1},\xv_{2}}P_{\Uv}(\uv)P_{\Xv_{1}}(\xv_{1})P_{\Xv_{2}}(\xv_{2}) \nonumber \\ 
    &\quad\times \sum_{\yv}W^{n}(\yv|\Xi(\uv,\xv_{1},\xv_{2}))\Bigg(\sum_{\uvbar}P_{\Uv}(\uvbar) \nonumber \\
    &\quad\times \bigg(\sum_{\xvbar_{1}}P_{\Xv_{1}}(\xvbar_{1})\bigg(\frac{q^{n_{1}}\big(\xvbar_{1},\yv_{1}(\uvbar))}{q^{n_{1}}(\xv_{1},\yv_{1}(\uv))}\bigg)^{\rho_{2}s}\bigg)^{\rho_{1}} \nonumber \\ 
    &\quad\times\bigg(\sum_{\xvbar_{2}}P_{\Xv_{2}}(\xvbar_{2})\bigg(\frac{q^{n_{2}}\big(\xvbar_{2},\yv_{2}(\uvbar))}{q^{n_{2}}(\xv_{2},\yv_{2}(\uv))}\bigg)^{s}\bigg)^{\rho_{1}\rho_{2}}\Bigg)^{\rho_{0}},
\end{align}
where $s\ge0$ and $\rho_{1},\rho_{2}\in[0,1]$ are arbitrary. Using the definition  
of the ensemble in \eqref{eq:RSC_DistrCost}--\eqref{eq:RSC_Phi}, we obtain
\begin{align}
    &\peobar \,\,\dot{\le}\, (M_{0}M_{11}^{\rho_{1}}M_{12}^{\rho_{1}\rho_{2}})^{\rho_{0}} \nonumber \\
    &\times\sum_{\uv,\xv_{1},\xv_{2}}P_{\Uv}(\uv)P_{\Xv_{1}}(\xv_{1})P_{\Xv_{2}}(\xv_{2})\sum_{\yv}W^{n}(\yv|\Xi(\uv,\xv_{1},\xv_{2})) \nonumber \\ 
    & \times \Bigg(\sum_{\uvbar}P_{\Uv}(\uvbar)\bigg(\sum_{\xvbar_{1}}P_{\Xv_{1}}(\xvbar_{1})\nonumber \\
    &\quad\times\bigg(\frac{q^{n_{1}}\big(\xvbar_{1},\yv_{1}(\uvbar))}{q^{n_{1}}(\xv_{1},\yv_{1}(\uv))}\bigg)^{\rho_{2}s}\frac{e^{a_{1}^{n_{1}}(\xvbar_{1})}}{e^{a_{1}^{n_{1}}(\xv_{1})}}\bigg)^{\rho_{1}} \nonumber \\
    & \quad\times \bigg(\sum_{\xvbar_{2}}P_{\Xv_{2}}(\xvbar_{2})\bigg(\frac{q^{n_{2}}\big(\xvbar_{2},\yv_{2}(\uvbar))}{q^{n_{2}}(\xv_{2},\yv_{2}(\uv))}\bigg)^{s}\frac{e^{a_{2}^{n_{2}}(\xvbar_{2})}}{e^{a_{2}^{n_{2}}(\xv_{2})}}\bigg)^{\rho_{1}\rho_{2}}\Bigg)^{\rho_{0}}, \label{eq:RSC_DirectDual3}
\end{align}
where for $u=1,2$, $a_{u}(\cdot)$ is one of the $L_{u}=2$ cost functions in \eqref{eq:RSC_SetD},
and $a_{u}^{n_{u}}(\xv_{u}) \defeq \sum_{i=1}^{n_{u}}a_{u}(x_{u,i})$.
For each $(\uv,\xv_{1},\xv_{2},\yv)$, we write the argument to the summation over
$\yv$ in \eqref{eq:RSC_DirectDual3} as a product of two terms, namely
\begin{align}
    \Tsf_{1} &\defeq W^{n}(\yv|\Xi(\uv,\xv_{1},\xv_{2})) \nonumber \\ 
             & \qquad~~ \times \Big(q^{n_{1}}(\xv_{1},\yv_{1}(\uv))^{-\rho_{1}\rho_{2}s}e^{-\rho_{1}a_{1}^{n_{1}}(\xv_{1})} \nonumber \\ 
             & \qquad~~ \times q^{n_{2}}(\xv_{2},\yv_{2}(\uv))^{-\rho_{1}\rho_{2}s}e^{-\rho_{1}\rho_{2}a_{2}^{n_{2}}(\xv_{2})}\Big)^{\rho_0} \\
    \Tsf_{2} &\defeq \Bigg(\sum_{\uvbar}P_{\Uv}(\uvbar) \nonumber \\ 
    &\hspace*{-4ex} \times\bigg(\sum_{\xvbar_{1}}P_{\Xv_{1}}(\xvbar_{1})q^{n_{1}}\big(\xvbar_{1},\yv_{1}(\uvbar))^{\rho_{2}s}e^{a_{1}^{n_{1}}(\xvbar_{1})}\bigg)^{\rho_{1}} \nonumber \\ 
    & \hspace*{-4ex}\times \bigg(\sum_{\xvbar_{2}}P_{\Xv_{2}}(\xvbar_{2})q^{n_{2}}\big(\xvbar_{2},\yv_{2}(\uvbar))^{s}e^{a_{2}^{n_{2}}(\xvbar_{2})}\bigg)^{\rho_{1}\rho_{2}}\Bigg)^{\rho_{0}}.
\end{align} 
Since $P_{\Xv_{u}}(\xv_u)$ is upper bounded by a subexponential prefactor times 
$\prod_{i=1}^{n}Q_{X|U}(x_{u,i}|u)$ for $u=1,2$ (see Proposition \ref{prop:MAC_SubExp}), we have
\begin{align}
    &\sum_{\xvbar_{1}}P_{\Xv_{1}}(\xvbar_{1})q^{n_{1}}\big(\xvbar_{1},\yv_{1}(\uvbar))^{\rho_{2}s}e^{a_{1}^{n_{1}}(\xvbar_{1})} \nonumber \\
    & \qquad\dot{\le}\, \prod_{i=1}^{n_{1}}\sum_{\xbar_{1}}Q_{X|U}(\xbar_{1}|1)q(\xbar_{1},y_{1,i}(\uvbar))^{\rho_{2}s}e^{a_{1}(\xbar_{1})} \label{eq:RSC_DirectDual5}\\
    & \sum_{\xvbar_{2}}P_{\Xv_{2}}(\xvbar_{2})q^{n_{2}}\big(\xvbar_{2},\yv_{2}(\uvbar))^{s}e^{a_{2}^{n_{2}}(\xvbar_{2})} \nonumber \\
    & \qquad\dot{\le}\, \prod_{i=1}^{n_{2}}\sum_{\xbar_{2}}Q_{X|U}(\xbar_{2}|2)q(\xbar_{2},y_{2,i}(\uvbar))^{s}e^{a_{2}(\xbar_{2})}, \label{eq:RSC_DirectDual6}
\end{align}
where for $u=1,2$, $y_{u,i}(\uvbar)$ is the $i$-th entry of $\yv_{u}(\uvbar)$.
Using the definitions in \eqref{eq:RSC_DualVars1} along with $a(u,x)\triangleq a_{u}(x)$, we therefore obtain
\begin{align}
      &\hspace*{-1.2ex} \bigg(\sum_{\xvbar_{1}}P_{\Xv_{1}}(\xvbar_{1})q^{n_{1}}\big(\xvbar_{1},\yv_{1}(\uvbar))^{\rho_{2}s}e^{a_{1}^{n_{1}}(\xvbar_{1})}\bigg)^{\rho_{1}} \nonumber \\ 
      &\hspace*{-1.2ex} \,~\times \bigg(\sum_{\xvbar_{2}}P_{\Xv_{2}}(\xvbar_{2})q^{n_{2}}\big(\xvbar_{2},\yv_{2}(\uvbar))^{s}e^{a_{2}^{n_{2}}(\xvbar_{2})}\bigg)^{\rho_{1}\rho_{2}} \\
      &\hspace*{-1.2ex} \,\,\dot{\le}\, \bigg(\prod_{i=1}^{n_{1}}\sum_{\xbar_{1}}Q_{X|U}(\xbar_{1}|1)q(\xbar_{1},y_{1,i}(\uvbar))^{\rho_{2}s}e^{a_{1}(\xbar_{1})}\bigg)^{\rho_{1}} \nonumber \\
      &\hspace*{-1.2ex}\,~\times\bigg(\prod_{i=1}^{n_{2}}\sum_{\xbar_{2}}Q_{X|U}(\xbar_{2}|2)q(\xbar_{2},y_{2,i}(\uvbar))^{s}e^{a_{2}(\xbar_{2})}\bigg)^{\rho_{1}\rho_{2}} \\
      &\hspace*{-1.2ex} = \prod_{i=1}^{n}\bigg(\sum_{\xbar}Q_{X|U}(\xbar|\ubar_{i})q(\xbar,y_{i})^{s_1(\ubar_{i})}e^{a(\ubar_{i},\xbar)}\bigg)^{\rho_1(\ubar_{i})}.
\end{align}
Hence, and using the fact that $P_{\Uv}(\uv) \,\,\dot{\le}\, Q_U^n(\uv)$ 
(see \cite[Ch.~2]{CsiszarBook}), we obtain
\begin{multline}
      \Tsf_{2} \,\,\dot{\le}\, \prod_{i=1}^{n}\Bigg(\sum_{\ubar}Q_{U}(\ubar) \\ \times\bigg(\sum_{\xbar}Q(\xbar|\ubar)q(\xbar,y_{i})^{s_1(\ubar)}e^{a(\ubar,\xbar)}\bigg)^{\rho_1(\ubar)}\Bigg)^{\rho_{0}}. \label{eq:RSC_T2Final}                      
\end{multline}
A similar argument (without the need for the $\dot{\le}$ steps) gives 
\begin{equation}
      \Tsf_{1} = \prod_{i=1}^{n}W(y_i|x_{i})\Big(q(x_{i},y_i)^{-\rho_1(u_{i})s_1(u_{i})}e^{-\rho_1(u_{i})a(u_{i},x_{i})}\Big)^{\rho_0}, \label{eq:RSC_T1Final}                                                                                                                                                                                                                  
\end{equation}
where we have used the fact that $W^{n}(\yv|\Xi(\uv,\xv_{1},\xv_{2})) = W^{n_{1}}(\yv_{1}(\uv)|\xv_{1})W^{n_{2}}(\yv_{2}(\uv)|\xv_{2})$.
Substituting \eqref{eq:RSC_T2Final} and \eqref{eq:RSC_T1Final} into \eqref{eq:RSC_DirectDual3}, 
we obtain
\begin{align}
    \peobar \,\,\dot{\le}\, &(M_{0}M_{11}^{\rho_{1}} M_{12}^{\rho_{1}\rho_{2}})^{\rho_{0}} \sum_{\uv,\xv}P_{\Uv\Xv}(\uv,\xv)\prod_{i=1}^{n}\sum_{y}W(y|x_{i}) \nonumber \\ 
    &\times \bigg(\sum_{\ubar}Q_{U}(\ubar)\bigg(\sum_{\xbar}Q_{X|U}(\xbar|\ubar_{i}) \nonumber\\ &\times\bigg(\frac{q(\xbar,y_{i})}{q(x_{i},y_{i})}\bigg)^{s_1(\ubar_{i})}\frac{e^{a(\ubar_{i},\xbar)}}{e^{a(u_{i},x_{i})}}\bigg)^{\rho_1(\ubar_{i})}\bigg)^{\rho_{0}}, \label{eq:RSC_DirectDual12}
\end{align}
where
\begin{multline}
    P_{\Uv\Xv}(\uv,\xv) \defeq \sum_{\xv_{1},\xv_{2}}P_{\Uv}(\uv)P_{\Xv_{1}}(\xv_{1})P_{\Xv_{2}}(\xv_{2}) \\ \times\openone\{\xv = \Xi(\uv,\xv_{1},\xv_{2})\}.
\end{multline}
If $P_{\Uv\Xv}$ were i.i.d.~on $Q_{UX}$, then \eqref{eq:RSC_DirectDual12} would yield an
error exponent that is positive when \eqref{eq:RSC_DualR0_1} ($u=1$) holds with strict inequality, 
by taking $\rho_0\to0$ similarly to Theorem \ref{thm:MAC_DualRate}.  The same can be done 
in the present setting by upper bounding $P_{\Uv\Xv}$ by a subexponential prefactor
times $Q_{UX}^{n}$, analogously to \eqref{eq:RSC_DirectDual5}--\eqref{eq:RSC_DirectDual6}.
More precisely, we have
\begin{align}
    & P_{\Uv\Xv}(\uv,\xv) \nonumber \\  &\,\,\dot{\le}\,\,\sum_{\xv_{1},\xv_{2}}P_{\Uv}(\uv)\bigg(\prod_{i=1}^{n}Q_{X|U}(x_{1,i}|1)\bigg) \nonumber \\ 
    &\qquad\times\bigg(\prod_{i=1}^{n}Q_{X|U}(x_{2,i}|2)\bigg)\openone\{\xv = \Xi(\uv,\xv_{1},\xv_{2})\} \\
                        &= P_{\Uv}(\uv)Q_{X|U}^{n}(\xv|\uv) \\
                        &\,\dot{\le}\,\, Q_{U}^{n}(\uv)Q_{X|U}^{n}(\xv|\uv) 
                        = Q_{UX}^{n}(\uv,\xv).
\end{align}


\bibliographystyle{IEEEtran}
\bibliography{12-Paper,18-MultiUser,18-SingleUser}

\begin{thebibliography}{10}
\providecommand{\url}[1]{#1}
\csname url@samestyle\endcsname
\providecommand{\newblock}{\relax}
\providecommand{\bibinfo}[2]{#2}
\providecommand{\BIBentrySTDinterwordspacing}{\spaceskip=0pt\relax}
\providecommand{\BIBentryALTinterwordstretchfactor}{4}
\providecommand{\BIBentryALTinterwordspacing}{\spaceskip=\fontdimen2\font plus
\BIBentryALTinterwordstretchfactor\fontdimen3\font minus
  \fontdimen4\font\relax}
\providecommand{\BIBforeignlanguage}[2]{{%
\expandafter\ifx\csname l@#1\endcsname\relax
\typeout{** WARNING: IEEEtran.bst: No hyphenation pattern has been}%
\typeout{** loaded for the language `#1'. Using the pattern for}%
\typeout{** the default language instead.}%
\else
\language=\csname l@#1\endcsname
\fi
#2}}
\providecommand{\BIBdecl}{\relax}
\BIBdecl

\bibitem{Hui}
J.~Hui, ``Fundamental issues of multiple accessing,'' Ph.D. dissertation, MIT,
  1983.

\bibitem{Csiszar1}
I.~Csisz\'{a}r and J.~K\"{o}rner, ``Graph decomposition: A new key to coding
  theorems,'' \emph{IEEE Trans. Inf. Theory}, vol.~27, no.~1, pp. 5--12, Jan.
  1981.

\bibitem{Csiszar2}
I.~Csisz\'{a}r and P.~Narayan, ``Channel capacity for a given decoding
  metric,'' \emph{IEEE Trans. Inf. Theory}, vol.~45, no.~1, pp. 35--43, Jan.
  1995.

\bibitem{Merhav}
N.~Merhav, G.~Kaplan, A.~Lapidoth, and S.~Shamai, ``On information rates for
  mismatched decoders,'' \emph{IEEE Trans. Inf. Theory}, vol.~40, no.~6, pp.
  1953--1967, Nov. 1994.

\bibitem{ConverseMM}
V.~Balakirsky, ``A converse coding theorem for mismatched decoding at the
  output of binary-input memoryless channels,'' \emph{IEEE Trans. Inf. Theory},
  vol.~41, no.~6, pp. 1889--1902, Nov. 1995.

\bibitem{MacMM}
A.~Lapidoth, ``Mismatched decoding and the multiple-access channel,''
  \emph{IEEE Trans. Inf. Theory}, vol.~42, no.~5, pp. 1439--1452, Sept. 1996.

\bibitem{MMRevisited}
A.~Ganti, A.~Lapidoth, and E.~Telatar, ``Mismatched decoding revisited:
  {G}eneral alphabets, channels with memory, and the wide-band limit,''
  \emph{IEEE Trans. Inf. Theory}, vol.~46, no.~7, pp. 2315--2328, Nov. 2000.

\bibitem{JournalSU}
J.~Scarlett, A.~Martinez, and A.~{Guill{\'e}n i F\`{a}bregas}, ``Mismatched
  decoding: Error exponents, second-order rates and saddlepoint
  approximations,'' \emph{IEEE Trans. Inf. Theory}, vol.~60, no.~5, pp.
  2647--2666, May 2014.

\bibitem{MMSomekh}
A.~Somekh-Baruch, ``On achievable rates and error exponents for channels with
  mismatched decoding,'' \emph{IEEE Trans. Inf. Theory}, vol.~61, no.~2, pp.
  727--740, Feb. 2015.

\bibitem{Compound}
G.~Kaplan and S.~Shamai, ``Information rates and error exponents of compound
  channels with application to antipodal signaling in a fading environment,''
  \emph{Arch. Elek. \"{U}ber.}, vol.~47, no.~4, pp. 228--239, 1993.

\bibitem{Variations}
S.~Shamai and I.~Sason, ``Variations on the {G}allager bounds, connections, and
  applications,'' \emph{IEEE Trans. Inf. Theory}, vol.~48, no.~12, pp.
  3029--3051, Dec. 2002.

\bibitem{MACExponent4}
Y.~Liu and B.~Hughes, ``A new universal random coding bound for the
  multiple-access channel,'' \emph{IEEE Trans. Inf. Theory}, vol.~42, no.~2,
  pp. 376--386, March 1996.

\bibitem{Thesis}
J.~Scarlett, ``Reliable communication under mismatched decoding,'' Ph.D.
  dissertation, University of Cambridge, 2014, [Online:
  http://itc.upf.edu/biblio/1061].

\bibitem{CsiszarBook}
I.~Csisz\'{a}r and J.~K\"{o}rner, \emph{Information Theory: Coding Theorems for
  Discrete Memoryless Systems}, 2nd~ed.\hskip 1em plus 0.5em minus 0.4em\relax
  Cambridge University Press, 2011.

\bibitem{GallagerCC}
R.~Gallager, ``Fixed composition arguments and lower bounds to error
  probability,'' \url{http://web.mit.edu/gallager/www/notes/notes5.pdf}.

\bibitem{TwoChannels}
P.~Elias, ``Coding for two noisy channels,'' in \emph{Third London Symp. Inf.
  Theory}, 1955.

\bibitem{MACExponent5}
A.~Nazari, A.~Anastasopoulos, and S.~Pradhan, ``Error exponent for
  multiple-access channels: Lower bounds,'' \emph{IEEE Trans. Inf}, vol.~60,
  no.~9, pp. 5095--5115, Sept. 2014.

\bibitem{DyachkovCC}
A.~G. D'yachkov, ``Bounds on the average error probability for a code ensemble
  with fixed composition,'' \emph{Prob. Inf. Transm.}, vol.~16, no.~4, pp.
  3--8, 1980.

\bibitem{NetworkBook}
A.~{El Gamal} and Y.~H. Kim, \emph{Network Information Theory}.\hskip 1em plus
  0.5em minus 0.4em\relax Cambridge University Press, 2011.

\bibitem{MACExponent2}
R.~Gallager, ``A perspective on multiaccess channels,'' \emph{IEEE Trans. Inf.
  Theory}, vol.~31, no.~2, pp. 124--142, March 1985.

\bibitem{Gallager}
------, \emph{Information Theory and Reliable Communication}.\hskip 1em plus
  0.5em minus 0.4em\relax John Wiley \& Sons, 1968.

\bibitem{PaperExpurg}
J.~Scarlett, L.~Peng, N.~Merhav, A.~Martinez, and A.~{Guill\'{e}n i
  F\`{a}bregas}, ``Expurgated random-coding ensembles: Exponents, refinements
  and connections,'' \emph{IEEE Trans. Inf. Theory}, vol.~60, no.~8, pp.
  4449--4462, Aug. 2014.

\bibitem{PaperITA}
J.~Scarlett, A.~Martinez, and A.~{Guill\'{e}n i F\`{a}bregas},
  ``Cost-constrained random coding and applications,'' in \emph{Inf. Theory and
  Apps. Workshop}, San Diego, CA, Feb. 2013.

\bibitem{MACNonConvex}
M.~Bierbaum and H.~Wallmeier, ``A note on the capacity region of the
  multiple-access channel,'' \emph{IEEE Trans. Inf. Theory}, vol.~25, no.~4,
  pp. 484--484, July 1979.

\bibitem{BCDegraded2}
R.~Gallager, ``Capacity and coding for degraded broadcast channels,''
  \emph{Prob. Peredachi Inf.}, vol.~10, no.~3, pp. 3--14, 1974.

\bibitem{BCDegMsg1}
J.~K\"orner and K.~Marton, ``General broadcast channels with degraded message
  sets,'' \emph{IEEE Trans. Inf. Theory}, vol.~23, no.~1, pp. 60--64, Jan.
  1977.

\bibitem{BCExp1}
J.~K\"orner and A.~Sgarro, ``Universally attainable error exponents for
  broadcast channels with degraded message sets,'' \emph{IEEE Trans. Inf.
  Theory}, vol.~26, no.~6, pp. 670--679, Nov. 1980.

\bibitem{ShulmanThesis}
N.~Shulman, ``Communication over an unknown channel via common broadcasting,''
  Ph.D. dissertation, Tel Aviv University, 2003.

\bibitem{YALMIP}
J.~L\"{o}fberg, ``{YALMIP} : A toolbox for modeling and optimization in
  {MATLAB},'' in \emph{Proc. CACSD Conf.}, Taipei, 2004.

\bibitem{Shannon}
C.~E. Shannon, ``A mathematical theory of communication,'' \emph{Bell Syst.
  Tech. Journal}, vol.~27, pp. 379--423, July and Oct. 1948.

\bibitem{ZUEC}
M.~S. Pinsker and A.~Sheverdjaev, ``Zero error capacity with erasure,''
  \emph{Prob. Inf. Transm.}, vol.~6, no.~1, pp. 20--24, 1970.

\bibitem{ShannonZero}
C.~E. Shannon, ``The zero error capacity of a noisy channel,'' \emph{IRE Trans.
  Inf. Theory}, vol.~2, no.~3, pp. 8--19, Sept. 1956.

\bibitem{ZUEC2}
R.~Ahlswede, N.~Cai, and Z.~Zhang, ``Erasure, list, and detection zero-error
  capacities for low noise and a relation to identification,'' \emph{IEEE
  Trans. Inf. Theory}, vol.~42, no.~1, pp. 55--62, Jan. 1996.

\bibitem{PaperBI}
J.~Scarlett, A.~Somekh-Baruch, A.~Martinez, and A.~Guill\'en~i F\`abregas, ``A
  counter-example to the mismatched decoding converse for binary-input discrete
  memoryless channels,'' \emph{IEEE Trans. Inf. Theory}, vol.~61, no.~10, pp.
  5387--5395, Oct. 2015.

\bibitem{Balikirsky}
V.~Balakirsky, ``Coding theorem for discrete memoryless channels with given
  decision rule,'' in \emph{Algebraic Coding}.\hskip 1em plus 0.5em minus
  0.4em\relax Springer Berlin / Heidelberg, 1992, vol. 573, pp. 142--150.

\bibitem{TanRelay}
V.~Y.~F. Tan, ``On the reliability function of the discrete memoryless relay
  channel,'' \emph{IEEE Trans. Inf. Theory}, vol.~61, no.~4, pp. 1550--1573,
  April 2015.

\bibitem{Caen}
D.~{de Caen}, ``A lower bound on the probability of a union,'' \emph{Discrete
  Math.}, vol. 169, pp. 217--220, 1997.

\bibitem{PaperSU}
J.~Scarlett, A.~Martinez, and A.~{Guill\'{e}n i F\`{a}bregas}, ``Ensemble-tight
  error exponents for mismatched decoders,'' in \emph{Allerton Conf. on Comm.,
  Control and Comp.}, Monticello, IL, Oct. 2012, pp. 1951--1958.

\bibitem{Convex}
S.~Boyd and L.~Vandenberghe, \emph{Convex Optimization}.\hskip 1em plus 0.5em
  minus 0.4em\relax Cambridge University Press, 2004.

\bibitem{Cover}
T.~M. Cover and J.~A. Thomas, \emph{Elements of Information Theory}.\hskip 1em
  plus 0.5em minus 0.4em\relax John Wiley \& Sons, Inc., 2006.

\bibitem{Minimax}
K.~Fan, ``Minimax theorems,'' \emph{Proc. Nat. Acad. Sci.}, vol.~39, pp.
  42--47, 1953.

\end{thebibliography}

\vspace*{-2cm}
\begin{IEEEbiographynophoto}{Jonathan Scarlett}
(S'14 -- M'15) was born in Melbourne, Australia, in 1988. In 2010, he received 
the B.Eng. degree in electrical engineering and the B.Sci. degree in 
computer science from the University of Melbourne, Australia. In 2011, 
he was a research assistant at the Department of Electrical \& Electronic 
Engineering, University of Melbourne.  From October 2011 to August 2014,  
he was a Ph.D. student in the Signal Processing and Communications Group
at the University of Cambridge, United Kingdom. He
is now a post-doctoral researcher with the Laboratory for Information
and Inference Systems at the \'Ecole Polytechnique F\'ed\'erale de Lausanne,
Switzerland.  His research interests are in the areas of information theory, 
signal processing, machine learning, and high-dimensional statistics. 
He received the Cambridge Australia Poynton International Scholarship, and the 'EPFL Fellows' postdoctoral fellowship co-funded by Marie Curie.
\end{IEEEbiographynophoto}

\vspace*{-2cm}
\begin{IEEEbiographynophoto}{Alfonso Martinez}
(SM'11) was born in Zaragoza, Spain, in October 1973. He
is currently a Ram\'on y Cajal Research Fellow at Universitat Pompeu Fabra,
Barcelona, Spain. He obtained his Telecommunications Engineering degree
from the University of Zaragoza in 1997. In 1998-2003 he was a Systems Engineer
at the research centre of the European Space Agency (ESAESTEC) in
Noordwijk, The Netherlands. His work on APSK modulation was instrumental
in the definition of the physical layer of DVB-S2. From 2003 to 2007 he was
a Research and Teaching Assistant at Technische Universiteit Eindhoven, The
Netherlands, where he conducted research on digital signal processing for
MIMO optical systems and on optical communication theory. Between 2008
and 2010 he was a post-doctoral fellow with the Information-Theoretic Learning
Group at Centrum Wiskunde \& Informatica (CWI), in Amsterdam, The
Netherlands. In 2011 he was a Research Associate with the Signal Processing
and Communications Lab at the Department of Engineering, University
of Cambridge, Cambridge, U.K. His research interests lie in the fields of
information theory and coding, with emphasis on digital modulation and the
analysis of mismatched decoding; in this area he has coauthored a monograph
on ``Bit-Interleaved Coded Modulation.'' More generally, he is intrigued by the
connections between information theory, optical communications, and physics,
particularly by the links between classical and quantum information theory.
\end{IEEEbiographynophoto}

\vspace*{-2cm}
\begin{IEEEbiographynophoto}{Albert Guill\'en i F\`abregas}
	(S'01 -- M'05 -- SM'09) received the Telecommunication
	Engineering Degree and the Electronics Engineering Degree from
	Universitat Polit\`ecnica de Catalunya and Politecnico di Torino, respectively
	in 1999, and the Ph.D. in Communication Systems from \'Ecole Polytechnique
	F\'ed\'erale de Lausanne (EPFL) in 2004.
	Since 2011 he has been an ICREA Research Professor at Universitat Pompeu
	Fabra. He is also an Adjunct Researcher at the University of Cambridge.
	He has held appointments at the New Jersey Institute of Technology, Telecom
	Italia, European Space Agency (ESA), Institut Eur\`ecom, University of South
	Australia, University of Cambridge, as well as visiting appointments at EPFL,
	\'Ecole Nationale des T\'el\'ecommunications (Paris), Universitat Pompeu Fabra,
	University of South Australia, Centrum Wiskunde \& Informatica and Texas
	A\&M University in Qatar. His research interests are in the areas of information
	theory, coding theory and communication theory.
	Dr. Guill\'en i F\`abregas is a Member of the Young Academy of Europe, and
	received the Starting Grant from the European Research Council, the Young
	Authors Award of the 2004 European Signal Processing Conference, the 2004
	Best Doctoral Thesis Award from the Spanish Institution of Telecommunications
	Engineers, and a Research Fellowship of the Spanish Government
	to join ESA. He is also an Associate Editor of the IEEE Transactions on Information Theory, an Editor of the Foundations and Trends in
	Communications and Information Theory, Now Publishers and was an Editor
	of the IEEE Transactions on Wireless Communications.
\end{IEEEbiographynophoto}

\end{document}